\documentclass[a4paper,11pt]{llncs}
\usepackage{microtype}
\usepackage{xcolor}
\usepackage{wrapfig}
\usepackage{microtype}
\usepackage{graphicx}
\usepackage{graphics}
\usepackage{epsfig}
\usepackage{epstopdf}
\usepackage{amsmath}
\usepackage{latexsym}
\usepackage{multirow}
\usepackage{amsfonts}
\usepackage{cite}
\usepackage{amssymb}
\usepackage{caption}
\usepackage{algorithm}
\usepackage{algorithmic}
\usepackage{enumerate}
\usepackage{diagbox}
 \usepackage{multirow}  
\usepackage{makecell}
\usepackage{graphicx}
\usepackage{hyperref}

\usepackage[a4paper,margin=2.7cm,footskip=0.25in]{geometry}
\makeatletter
\newcommand{\rmnum}[1]{\romannumeral #1}
\newcommand{\Rmnum}[1]{\expandafter\@slowromancap\romannumeral #1@}
\makeatother

\spnewtheorem{claim}{Claim}{\bfseries}{\rmfamily}

\spnewtheorem{remark}{Remark}{\bfseries}{\rmfamily}

\spnewtheorem{property}{Property}{\bfseries}{\rmfamily}

\begin{document}

\title{ Sublinear Time Algorithms for Several  Geometric Optimization (With Outliers) Problems In Machine Learning\thanks{Part of this work has appeared in~\cite{DBLP:conf/esa/Ding20,DBLP:conf/esa/Ding21}.
}}
\author{Hu Ding}
\institute{
 School of Computer Science and Engineering, University of Science and Technology of China \\
 He Fei, China\\
  \email{huding@ustc.edu.cn}\\
}
%
\maketitle

\thispagestyle{empty}



\begin{abstract}
In this paper, we study several important geometric optimization problems arising in machine learning. First, we revisit the Minimum Enclosing Ball (MEB) problem in Euclidean space $\mathbb{R}^d$. The problem has been extensively studied before, but real-world machine learning tasks often need to handle large-scale datasets so that we cannot even afford linear time algorithms. Motivated by the recent studies on {\em beyond worst-case analysis}, we introduce the notion of stability for MEB, which is natural and easy to understand. Roughly speaking, an instance of MEB is stable, if the radius of the resulting ball cannot be significantly reduced by removing a small fraction of the input points. Under the stability assumption, we present two sampling algorithms for computing radius-approximate MEB with sample complexities independent of the number of input points $n$. In particular, the second algorithm has the sample complexity even independent of the dimensionality $d$. We also consider the general case without the stability assumption. We present a hybrid algorithm that can output either a radius-approximate MEB or a covering-approximate MEB.  Our algorithm improves the running time and the number of passes for the previous sublinear MEB algorithms.  Our method relies on two novel techniques, the Uniform-Adaptive Sampling method and Sandwich Lemma.  Furthermore, we observe that these two techniques can be generalized to design sublinear time algorithms for a broader range of geometric optimization problems with outliers in high dimensions, including MEB with outliers, one-class and two-class linear SVMs   with outliers,  $k$-center clustering with outliers, and flat fitting with outliers. Our proposed algorithms also work fine for  kernels.

\end{abstract}

%
%
%
  

\pagestyle{plain}
\pagenumbering{arabic}
\setcounter{page}{1}

\section{Introduction}
\label{sec-intro}
Many real-world machine learning tasks can be formulated as geometric optimization problems in Euclidean space. We start with a fundamental geometric optimization problem,  {\em Minimum Enclosing Ball (MEB)}, which has attracted a lot of attentions in past years. 
Given a set $P$ of $n$ points in Euclidean space $\mathbb{R}^d$, where $d$ could be quite high, the problem of MEB is to find 
a ball with minimum radius to cover all the points in $P$ \cite{badoiu2003smaller,DBLP:journals/jea/KumarMY03,DBLP:conf/esa/FischerGK03}. MEB finds 
several important applications in machine learning~\cite{DBLP:journals/ijcga/NielsenN09}. For example, the popular classification model  {\em Support  Vector Machine (SVM)} can be formulated as an MEB problem in high dimensional space, if all the mapped points have the same norm by using kernel method, {\em e.g.,} the popular radial basis function kernel; this SVM is called ``Core Vector Machine (CVM)'' which is currently one of the most important SVM training methods for large-scale data sets, since it was proposed in 2005~\cite{DBLP:journals/jmlr/TsangKC05}. Hence 
fast MEB algorithms can be used to speed up 
its training procedure \cite{DBLP:journals/jmlr/TsangKC05,C10,DBLP:journals/jacm/ClarksonHW12}.
Recently, 
MEB has also been studied for preserving privacy~\cite{DBLP:conf/pods/NissimSV16,DBLP:conf/ipsn/FeldmanXZR17} and quantum cryptography~\cite{gyongyosi2014geometrical}.  


Usually, we consider the approximate solutions of MEB. If a ball covers all the $n$ points but has a radius larger than the optimal one, we call it a  \textbf{``radius-approximate solution''}; if a ball has the radius no larger than the optimal one but covers less than $n$ points, we call it a \textbf{``covering-approximate solution''} instead (the formal definitions are shown in Section~\ref{sec-pre}). 
In the era of big data, the dataset could be so large that we cannot even afford linear time algorithms. This motivates us to ask 
the following questions:

\vspace{0.03in}
{\em Is it possible to develop approximation algorithms for MEB  that run in sublinear time in the input size? And how about other high-dimensional geometric optimization problems arising in machine learning?}
\vspace{0.03in}

It is common to assume that the input data is represented by a  $n\times d$ matrix, and  any algorithm having complexity $o(nd)$ can be considered as a sublinear time algorithm. 
In practice,  data items are usually represented as sparse vectors in $\mathbb{R}^d$; so it can be fast to perform the operations, like distance computing, even though the dimensionality $d$ is high ({\em e.g.,} if each vector has $s\ll d$ non-zero entries, the time for computing the distance is $O(s)$ rather than $O(d)$; see the concluding remarks of~\cite{DBLP:journals/jacm/ClarksonHW12}). Moreover, the number of input points $n$ is often much larger than the dimensionality $d$ in many real-world scenarios.  \textbf{Therefore, we are interested in designing the algorithms that have complexities sublinear in $n$ (or linear in $n$ but with  small factor before it).} 



\subsection{Our Main Ideas and Results}
\label{sec-overview}
Our idea for designing sublinear time MEB  algorithms is inspired by the recent studies on optimization with respect to stable instances, under the umbrella of {\em beyond worst-case analysis} \cite{DBLP:journals/cacm/Roughgarden19}. 
For example, several recent works introduced 
the notion of stability for problems like clustering and max-cut~\cite{balcan2013clustering,DBLP:journals/cpc/BiluL12,DBLP:conf/focs/AwasthiBS10}. 
In this paper, we give the notion of  \textbf{``stability''} for MEB. Roughly speaking, an instance of MEB  is stable, if the radius of the resulting ball cannot be significantly reduced by removing a small fraction of the input points ({\em e.g.,} the radius cannot be reduced by $10\%$ if only $1\%$ of the points are removed). The rationale behind 
this notion is quite natural: if the given instance is not stable, the small fraction of points causing  significant reduction in the radius should be viewed as outliers (or we may need 
multiple balls to cover the input points as the $k$-center clustering problem~\cite{hochbaum1985best,gonzalez1985clustering}). To the best of our knowledge, this is the first study on MEB  from the perspective of stability.

We prove an important implication of the stability assumption: informally speaking, if an instance of MEB is stable, its center should reveal a certain extent of robustness in the space (Section~\ref{sec-stable}). 
Using 
this implication,
we propose two sampling algorithms for computing $(1+\epsilon)$-radius approximate MEB with sublinear time complexities (Section~\ref{sec-sub}); in particular, our second algorithm has the sample size ({\em i.e.,} the number of sampled points) independent of the number of input points $n$ and dimensionality $d$ (to the best of our knowledge, this is the first algorithm achieving $(1+\epsilon)$-radius approximation with such a sublinear complexity). 

Moreover, we have an interesting observation: the ideas developed under the stability assumption can even help us to solve the general instance without the stability assumption, if we relax the requirement slightly. We  introduce \textbf{a hybrid approach} that can output either a radius-approximate MEB or a covering-approximate MEB, depending upon whether the input instance is sufficiently stable\footnote{We do not need to explicitly know whether the instance is stable or not, when running our algorithm. } (Section~\ref{sec-final-extent}). Also, a byproduct is that we can  infer the stability degree of the given instance from the output. It is worth noting that the simple uniform sampling idea based on VC-dimension~\cite{vapnik2015uniform,haussler1987} can only yield a ``bi-criteria'' approximation, which has errors on both the radius and the number of covered points (see the discussion on our first sampling algorithm in Section~\ref{sec-firstesa}). 
Comparing with the sublinear time MEB algorithm proposed by Clarkson {\em et al.}~\cite{DBLP:journals/jacm/ClarksonHW12}, we reduce the total running time from $\tilde{O}(\epsilon^{-2}n+\epsilon^{-1}d+M)$ to $O(n+h(\epsilon, \delta)\cdot d+M)$, where $M$ is the number of non-zero entries in the input $n\times d$ matrix and $h(\epsilon, \delta)$ is a factor depending on the pre-specified radius error bound $\epsilon$ and covering error bound $\delta$. Thus, our improvement is significant if $n\gg d$. The only tradeoff is that we allow a covering approximation for unstable instance (given the lower bound proved by~\cite{DBLP:journals/jacm/ClarksonHW12}, it is quite unlikely to reduce the term $\epsilon^{-2}n$ if we keep restricting the output to be $(1+\epsilon)$-radius approximation). Moreover, our algorithm only needs \textbf{uniform sampling and a single pass over the data}; on the other hand, the algorithm of~\cite{DBLP:journals/jacm/ClarksonHW12} needs  $\tilde{O}(\epsilon^{-1})$ passes (the details are  shown in Table~\ref{tab-1}). 
In addition to the stability idea, our method also relies on two key techniques, the novel ``\textbf{Uniform-Adaptive Sampling}'' method and ``\textbf{Sandwich Lemma}''.  Roughly speaking, the Uniform-Adaptive Sampling method can help us to bound the error induced in each ``randomized greedy selection'' step; the Sandwich Lemma enables us to estimate the objective value of each candidate and select the best one in sublinear time.

\begin{table}[h]
 \renewcommand{\arraystretch}{1.3}
 \newcommand{\tabincell}[2]{\begin{tabular}{@{}#1@{}}#2\end{tabular}} 
 \center
 \begin{tabular}{p{1.2cm}<{\centering}|p{1.5cm}<{\centering}|c|c|c|p{2.8cm}<{\centering}}
  \hline
  \multicolumn{2}{c|}{\textbf{Results}} & \textbf{Quality} & \textbf{Time }& \textbf{Number of passes} & \textbf{Extendibility for MEB with outliers} \\ \hline
  \multicolumn{2}{l|}{Clarkson {\em et al.}~\cite{DBLP:journals/jacm/ClarksonHW12}} & $(1+\epsilon)$-rad.  & $\tilde{O}(\epsilon^{-2}n+\epsilon^{-1}d+M)$ & $\tilde{O}(\epsilon^{-1})$ & N/A\\ \hline
  \multicolumn{2}{l|}{\tabincell{c}{Core-sets methods\\ \cite{badoiu2003smaller,DBLP:journals/jea/KumarMY03,panigrahy2004minimum,C10}}} & $(1+\epsilon)$-rad. &  {\tabincell{c}{roughly $O(\epsilon^{-1}nd)$ \\  or $O(\epsilon^{-1}(n+d+M))$\\ if $M=o(nd)$}} & $O(\epsilon^{-1})$ & bi-criteria approx.~\cite{BHI} \\ \hline
  \multicolumn{2}{l|}{Numerical method~\cite{DBLP:conf/soda/SahaVZ11}} & $(1+\epsilon)$-rad. & {\tabincell{c}{$\tilde{O}(\epsilon^{-1/2}nd)$ or \\ $\tilde{O}(\epsilon^{-1/2}(n+d+M))$\\ if $M=o(nd)$}}& $O(\epsilon^{-1/2})$ & N/A\\ \hline
    \multicolumn{2}{l|}{Numerical method~\cite{DBLP:conf/icalp/ZhuLY16}} & $(1+\epsilon)$-rad. &   $\tilde{O}(nd+n\sqrt{d}/\sqrt{\epsilon})$ & $\tilde{O}(d+\sqrt{d/\epsilon})$ & N/A\\ \hline
  \multicolumn{2}{l|}{Streaming algorithm~\cite{DBLP:journals/comgeo/ChanP14,DBLP:journals/algorithmica/AgarwalS15}} & $1.22$-rad. & $O(nd/\epsilon^5)$ & one pass & N/A \\ \hline
  \multirow{2}{*}{\tabincell{c}{\\ \\\textbf{This} \\ \textbf{paper}}} & stable instance   & $(1+\epsilon)$-rad. & \tabincell{c}{ $O(C_1\cdot d)$ (Sec.~\ref{sec-secondesa}) } & uniform sampling & N/A  \\ \cline{2-6}  
  & general instance & \tabincell{c}{$(1+\epsilon)$-rad. \\  or $(1-\delta)$-cov.} &  \tabincell{c}{$O\big((n+C_2)d\big)$ or \\ $O(n+C_2\cdot d+M)$\\  if $M=o(nd)$  (Sec.~\ref{sec-unknown2})} & \tabincell{c}{uniform sampling\\ plus a single pass} &  \tabincell{c}{$(1+\epsilon)$-rad. \\  or $(1-\delta)$-cov. \\(Sec.~\ref{sec-unknown})}  \\  \hline
 \end{tabular}
 \vspace{0.1in}
 \caption{The existing and our results for computing MEB in high dimensions. In the table, ``rad.'' and ``cov.'' stand for ``radius approximation'' and ``covering approximation'', respectively. ``$M$''  is the number of non-zero entries in the input $n\times d$ matrix. The factor $C_1$ depends on $\epsilon$ and the stability degree of the given instance; the factor $C_2$ depends on $\epsilon$ and $\delta$. }
  \label{tab-1}
\end{table}

Finally, we present several extensions of our result. In practice,  we may assume the presence of outliers in 
given datasets. In particular,  as the rapid development of machine learning, the field of {\em  adversarial machine learning}  has attracted a great amount of attentions~\cite{biggio2018wild,DBLP:journals/cacm/GoodfellowMP18}. A small set of outliers could be added by some adversarial attacker to make the model severely deviate and cause unexpected error (the seminal paper~\cite{DBLP:conf/icml/BiggioNL12} on poisoning attacks against SVM has just received the {\em ICML2022 Test of Time award}). To defend such poisoning attacks, we often design robust algorithms that are resilient against outliers~\cite{DBLP:conf/sp/JagielskiOBLNL18}. However, the presence of outliers makes the problem  not only non-convex but also highly combinatorial in high dimensions; for example, if $m$ of the input $n$ data items are outliers ($m<n$), we have to consider an exponentially large number ${n\choose m}$ of  different possible cases when optimizing the objective function. So we consider to design sublinear time algorithms for the following problems.

\vspace{0.1in}
\textbf{MEB with outliers.} 
MEB with outliers is a natural generalization of the MEB problem, where the goal is to find the minimum ball covering at least a certain  fraction of input points.  We can apply MEB with outliers to solve many practical problems ({\em e.g.,} outlier recognition) in data mining and data analysis~\cite{tan2006introduction}. 
We define the stability for MEB with outliers, and propose the sublinear time  approximation algorithm. Our algorithm is the first sublinear time  algorithm  for the MEB with outliers problem (comparing with the previous linear time algorithms like~\cite{BHI}), to the best of our knowledge.

%
 
  \vspace{0.1in}
  \textbf{Other enclosing with outliers problems. }    Besides MEB with outliers, we observe that our proposed techniques can be used to solve a broader range of enclosing with outliers problems. We define a general optimization problem called \textbf{ minimum enclosing ``x'' (MEX) with Outliers}, where the ``x'' could be a specified kind of shape ({\em e.g.,} the shape is a ball for MEB with outliers). We prove that it is possible to generalize the Uniform-Adaptive Sampling method and  Sandwich Lemma to adapt the shape ``x'', as long as it satisfies several properties. 
 In particular we focus on the MEX with outlier problems including flat fitting, $k$-center clustering, and SVM with outliers; a common characteristic of these problems is that each of them has an iterative algorithm based on greedy selection for its vanilla version (without outliers) that is similar to the MEB algorithm of~\cite{badoiu2003smaller}. Though these problems have been widely studied before, the research in terms of their sublinear time algorithms is till quite limited.

 \begin{remark}
 Because the geometric optimization problems studied in this paper are motivated from machine learning applications, we also take into account the \textbf{kernels}~\cite{DBLP:books/lib/ScholkopfS02}. Our proposed algorithms only need to conduct the basic operations, like computing the distance and inner product, on the data items. Therefore, 
 they also work fine for kernels. 
 \end{remark}
 
 \textbf{The rest of the paper is organized as follows.} In Section~\ref{sec-related}, we summarize the previous results that are related to our work. In Section~\ref{sec-pre}, we present the important definitions and briefly introduce the coreset construction method for MEB from~\cite{badoiu2003smaller} (which will be used in our following algorithms  and analysis). In Section~\ref{sec-stable}, we prove the implication of MEB stability. Further,  in Section~\ref{sec-sub} we propose two sublinear time MEB algorithms for stable instance. In Section~\ref{sec-final-extent}, we propose two key techniques, Uniform-Adaptive sampling and Sandwich lemma, and then present our sublinear time algorithm for general MEB without the stability assumption. In Section~\ref{sec-unknown}, we extend the idea of Section~\ref{sec-final-extent} to the MEB with outliers problem.  Finally, we present  the generalized Uniform-Adaptive sampling and Sandwich lemma, together with the applications in several enclosing with outliers problems (including flat fitting, $k$-center clustering, and SVM with outliers) in Section~\ref{sec-ext}.

\subsection{Previous Work}
\label{sec-related}

The works most related to ours  are \cite{DBLP:journals/jacm/ClarksonHW12,alon2003testing}. Clarkson {\em et al.}~\cite{DBLP:journals/jacm/ClarksonHW12} developed an elegant perceptron framework for solving several optimization problems arising in machine learning, such as MEB.  Given a set of $n$ points in $\mathbb{R}^d$ represented as  
an $n\times d$ matrix with $M$ non-zero entries, 
their framework can compute the MEB in $\tilde{O}(\frac{n}{\epsilon^2}+\frac{d}{\epsilon})$ time~\footnote{The asymptotic notation $\tilde{O}(f)=O\big(f\cdot \mathtt{polylog}(\frac{nd}{\epsilon})\big)$.}. Note that the parameter ``$\epsilon$'' is an additive error ({\em i.e.,} the resulting radius is $r+\epsilon$ if $r$ is the radius of the optimal MEB) which can be converted into a relative error ({\em i.e.,} $(1+\epsilon)r$) in $O(M)$ preprocessing time. Thus, if $M=o(nd)$, the running time is still sublinear in the input size $nd$ (please see Table~\ref{tab-1}). 
The framework of~\cite{DBLP:journals/jacm/ClarksonHW12} also inspires the sublinear time algorithms for training SVMs~\cite{DBLP:conf/nips/HazanKS11} and approximating Semidefinite Programs~\cite{DBLP:conf/nips/GarberH11}. Hayashi and Yoshida~\cite{DBLP:conf/nips/HayashiY16} presented a sampling-based method for minimizing quadratic functions of which the MEB objective is a special case, but it yields a large additive error $O(\epsilon n^2)$.

Alon {\em et al.}~\cite{alon2003testing} studied the following property testing problem: given a set of $n$ points in some metric space, 
determine 
whether the instance is $(k, b)$-clusterable, where an instance is called $(k, b)$-clusterable if it can be covered by $k$ balls with radius (or diameter) $b>0$. They proposed several sampling algorithms to answer the question ``approximately''. Specifically, they distinguish between the case that the instance is $(k, b)$-clusterable and the case that it is $\epsilon$-far away from $(k, b')$-clusterable, where $\epsilon\in (0,1)$ and $b'\geq b$. ``$\epsilon$-far'' means that more than $\epsilon n$ points should be removed so that it becomes $(k, b')$-clusterable. 
Note that their method cannot yield a single criterion radius-approximation or covering-approximation algorithm for the MEB problem, since it will introduce unavoidable errors on the radius and the number of covered points due to the relaxation of ``$\epsilon$-far''. But  it is possible to convert it into a ``\textbf{bi-criteria}'' approximation, where it allows approximations on both the radius and the number of uncovered outliers ({\em e.g.,} discard more than the pre-specified number of outliers).

 \vspace{0.1in}

\textbf{MEB and core-set.} A {\em core-set} is a small set of points that approximates the structure/shape of a much larger point set~\cite{agarwal2005geometric,DBLP:journals/corr/Phillips16,DBLP:journals/widm/Feldman20}. The core-set idea  has also been used to compute approximate  MEB  in high dimensional space~\cite{BHI,DBLP:journals/jea/KumarMY03,panigrahy2004minimum,DBLP:conf/isaac/KerberS13}. B\u{a}doiu and Clarkson \cite{badoiu2003smaller} showed that it is possible to find a core-set of size $\lceil2/\epsilon\rceil$ that yields a $(1+\epsilon)$-radius approximate MEB. Several other methods can yield even lower core-set sizes, such as~\cite{coresets1,DBLP:conf/isaac/KerberS13}. 
In fact, the algorithm for computing the core-set of MEB is a {\em Frank-Wolfe}  algorithm~\cite{frank1956algorithm}, which has been systematically studied by Clarkson~\cite{C10}. 
Other MEB algorithms that do not rely on core-sets include 
~\cite{DBLP:conf/esa/FischerGK03,DBLP:conf/soda/SahaVZ11,DBLP:conf/icalp/ZhuLY16}.  Agarwal and Sharathkumar~\cite{DBLP:journals/algorithmica/AgarwalS15} presented a streaming $(\frac{1+\sqrt{3}}{2}+\epsilon)$-radius approximation algorithm for computing MEB; later, Chan and Pathak~\cite{DBLP:journals/comgeo/ChanP14} proved that the same algorithm actually yields an approximation ratio less than $1.22$. Very recently,  Cohen-Addad {\em et al.}~\cite{cohen2021improved} proposed the sublinear time algorithm for computing high dimensional power means ({\em e.g.,} geometric median and mean points) by using core-sets.

 \vspace{0.1in}

\textbf{MEB with outliers and $k$-center clustering with outliers.} 
The MEB with outliers problem can be viewed as the case $k=1$ of the $k$-center clustering with outliers problem~\cite{charikar2001algorithms}. 
B\u{a}doiu {\em et al.}~\cite{BHI} extended their core-set idea to the problems of MEB and $k$-center clustering with outliers, and achieved linear time bi-criteria approximation algorithms (if $k$ is assumed to be a constant).  Huang {\em et al.}~\cite{DBLP:conf/focs/HuangJLW18} and  Ding {\em et al.}~\cite{DBLP:journals/corr/abs-1901-08219} respectively showed that simple uniform sampling approach can yield bi-criteria approximation of $k$-center clustering with outliers.  
Several algorithms for the low dimensional MEB with outliers have also been developed~\cite{aggarwal1991finding,efrat1994computing,har2005fast,matouvsek1995enclosing}. 
There also exist a number of works on streaming MEB and $k$-center clustering with outliers~\cite{charikar2003better,mccutchen2008streaming,zarrabistreaming,DBLP:journals/corr/abs-1802-09205}. 
Other related topics include robust optimization~\cite{DBLP:journals/ior/BertsimasS04}, robust fitting~\cite{har2004shape,agarwal2008robust}, and optimization with uncertainty~\cite{DBLP:journals/mp/CalafioreC05}.

 \vspace{0.1in}


\textbf{SVM with outliers.}  Given two point sets $P_1$ and $P_2$ in $\mathbb{R}^d$, the problem of {\em Support Vector Machine (SVM)} is to find the largest margin to separate $P_1$ and $P_2$ (if they are separable)~\cite{journals/tist/ChangL11}. 
SVM can be formulated as a quadratic programming problem, and a number of efficient techniques have been developed in the past, such as the soft margin SVM~\cite{mach:Cortes+Vapnik:1995,platt99}, $\nu$-SVM~\cite{bb57389,conf/nips/CrispB99}, and Core-SVM~\cite{tkc-cvmfstv-05}. There also exist a number of works on designing robust algorithms for SVM with outliers~\cite{conf/aaai/XuCS06,icml2014c2_suzumura14,ding2015random}.  
 \vspace{0.1in}

\textbf{Flat fitting with outliers.} Given an integer $j\geq 0$ and a set of points in $\mathbb{R}^d$, the flat fitting problem is to find a $j$-dimensional flat having the smallest maximum distance to the input points~\cite{DBLP:conf/focs/Har-PeledV01}; obviously, the MEB problem is a special case with $j=0$.  In high dimensions, Har-Peled and Varadarajan~\cite{DBLP:journals/dcg/Har-PeledV04} provided a linear time algorithm if $j$ is assumed to be fixed; their running time was further reduced by Panigrahy~\cite{panigrahy2004minimum} based on a core-set approach. There also exist several methods considering flat fitting with outliers but only for low-dimensional case~\cite{har2004shape,agarwal2008robust}.

 \vspace{0.1in}

\textbf{Optimizations under stability.} Bilu and Linial~\cite{DBLP:journals/cpc/BiluL12} showed that the Max-Cut problem becomes easier if the given instance is stable with respect to perturbation on edge weights. Ostrovsky {\em et al.}~\cite{ostrovsky2012effectiveness} proposed a separation condition for $k$-means clustering which refers to the scenario where  the clustering cost of $k$-means  is significantly lower than that of $(k-1)$-means  for a given instance, and demonstrated the effectiveness of the Lloyd heuristic \cite{lloyd1982least} under the separation condition. Balcan {\em et al.}~\cite{balcan2013clustering} introduced the concept of approximation-stability for finding the ground-truth of $k$-median and $k$-means clustering. 
Awasthi {\em et al.}~\cite{DBLP:conf/focs/AwasthiBS10} introduced another notion of clustering stability and gave a PTAS for $k$-median and $k$-means clustering. More clustering algorithms  under stability assumption were studied in~\cite{DBLP:journals/ipl/AwasthiBS12,DBLP:journals/siamcomp/BalcanL16,DBLP:conf/icalp/BalcanHW16,DBLP:conf/colt/BalcanB09,kumar2010clustering}.

 \vspace{0.1in}

\textbf{Sublinear time algorithms.} Besides the aforementioned sublinear MEB algorithm~\cite{DBLP:journals/jacm/ClarksonHW12}, a number of sublinear time algorithms have been studied for the problems like clustering~\cite{DBLP:conf/stoc/Indyk99,DBLP:conf/focs/Indyk99,meyerson2004k,mishra2001sublinear,czumaj2004sublinear} and property testing~\cite{DBLP:journals/jacm/GoldreichGR98,DBLP:books/sp/BhattacharyyaY22}. 
More detailed discussion on sublinear time algorithms can be found in the survey papers~\cite{rubinfeld2006sublinear,czumaj2006sublinear}.

\section{Definitions and Preliminaries}
\label{sec-pre}

We describe and analyze our algorithms in the unit-cost RAM model~\cite{10.5555/211390}. Suppose the input is represented by an $n\times d$ matrix ({\em i.e.,} $n$ points in $\mathbb{R}^d$). As mentioned in~\cite{DBLP:journals/jacm/ClarksonHW12},  it is common to assume that each entry  of the matrix can be recovered in constant time.

We let $|A|$ denote the number of points of a given point set $A$ in $\mathbb{R}^d$, and $||x-y||$ denote the Euclidean distance between two points $x$ and $y$ in $\mathbb{R}^d$. We use $\mathbb{B}(c, r)$ to denote the ball centered at a point $c$ with radius $r>0$. Below, we give the definitions for MEB and the notion of stability. To keep the structure of our paper more compact, we place other necessary definitions for our extensions to Section~\ref{sec-final-extent}, Section~\ref{sec-unknown},  and Section~\ref{sec-ext}, respectively.


\begin{definition}[Minimum Enclosing Ball (MEB)]
\label{def-meb}
Given a set $P$ of $n$ points in $\mathbb{R}^d$, the MEB problem is to find a ball with minimum radius to cover all the points in $P$. The resulting ball and its radius are denoted by $\mathbf{MEB}(P)$ and $\mathbf{Rad}(P)$, respectively.
\end{definition}

\begin{definition}[Radius Approximation and Covering Approximation]
\label{def-app}
Let $0<\epsilon, \delta<1$. A ball $\mathbb{B}(c, r)$ is called a $(1+\epsilon)$-radius approximation of $\mathbf{MEB}(P)$, if the ball covers all points in $P$ and has radius $r\leq (1+\epsilon) \mathbf{Rad}(P)$.  On the other hand, the ball is called a $(1-\delta)$-covering approximation of $\mathbf{MEB}(P)$, if it covers at least $(1-\delta)n$ points in $P$ and has radius $r\leq \mathbf{Rad}(P)$.
\end{definition}
Both radius approximation and covering approximation are single-criterion approximations. When $\epsilon$ ({\em resp.,} $\delta$) approaches to $0$, the $(1+\epsilon)$-radius approximation ({\em resp.,} $(1-\delta)$-covering approximation) will approach to  $\mathbf{MEB}(P)$. The ``covering approximation'' seems to be similar to ``MEB with outliers'', but actually they are quite different  (see Definition~\ref{def-outlier} in Section~\ref{sec-final-extent}).



%

\begin{definition}[($\alpha$, $\beta$)-stable]
\label{def-stable}
Given a set $P$ of $n$ points in $\mathbb{R}^d$ with two parameters $\alpha$ and $\beta$ in $(0,1)$,  $P$ is an ($\alpha$, $\beta$)-stable instance if \textbf{(1)} $\mathbf{Rad}(P\setminus Q)> (1-\alpha)\mathbf{Rad}(P)$ for any $Q\subset P$ with $|Q|<   \beta n $, and \textbf{(2)} there exists a $Q'\subset P$ with $|Q'|= \lceil\beta n\rceil$ having $\mathbf{Rad}(P\setminus Q')\leq (1-\alpha)\mathbf{Rad}(P)$. 
%
\end{definition}

\vspace{0.01in}
\textbf{The intuition of Definition~\ref{def-stable}.} Actually, $\beta$ can be viewed as a function of $\alpha$, and vice versa. For example,  given an $\alpha>0$, there always exists a $\beta\geq \frac{1}{n}$ such that $P$ is an ($\alpha$, $\beta$)-stable instance ($\beta\geq \frac{1}{n}$ because we must remove at least one point). 
The property of stability indicates that $\mathbf{Rad}(P)$ cannot be significantly reduced unless removing a large enough fraction of points from $P$. For a fixed $\alpha$, the larger $\beta$ is, the more stable $P$ should be. Similarly, for a fixed $\beta$,  the smaller $\alpha$ is, the more stable $P$ should be.

 Actually, our stability assumption is quite reasonable in practice. For example, if the radius  can be reduced considerably (say by $\alpha=10\%$) after removing only a very small fraction (say $\beta=1\%$) of  points, 
it is natural to view the small fraction of points as outliers. 
To better understand the notion of stability in high dimensions, we consider the following two examples. 
\vspace{0.01in}

\textbf{ Example (\rmnum{1}).} Suppose that the distribution of $P$ is uniform and dense inside $\mathbf{MEB}(P)$. Let $\alpha\in (0,1)$ be a fixed number, and we  study the corresponding $\beta$ of $P$. If we want the radius of the remaining $(1-\beta)n$ points to be as small as possible, intuitively we should remove the outermost $\beta n$ points (since $P$ is uniform and dense).  Let $Q'$  denote  the set of outermost $\beta n$ points that has $\mathbf{Rad}(P\setminus Q')\leq(1-\alpha)\mathbf{Rad}(P)$. Then we have $\frac{|P\setminus Q'|}{|P|}\approx \frac{Vol\big(\mathbf{MEB}(P\setminus Q')\big)}{Vol\big(\mathbf{MEB}(P)\big)}=\frac{(\mathbf{Rad}(P\setminus Q'))^d}{(\mathbf{Rad}(P))^d}\leq (1-\alpha)^d$, where $Vol(\cdot)$ is the volume function. That is, $1-\beta \leq(1-\alpha)^d$ and it implies $\lim_{d\to\infty}\beta =1$ when $\alpha$ is fixed; that means $P$ tends to be very stable as $d$ increases. 

\textbf{ Example (\rmnum{2}).} Consider a regular $d$-dimensional  simplex $P$ containing $d+1$ points where each pair of points have the pairwise distance equal to $1$. It is not hard to obtain $\mathbf{Rad}(P)=\sqrt{\frac{d}{2(1+d)}}$, and we denote it by $r_d$. If we remove $\beta  (d+1)$ points from $P$, namely it becomes a regular $d'$-dimensional simplex with $d'=(1-\beta  )(d+1)-1$, the new radius $r_{d'}=\sqrt{\frac{d'}{2(1+d')}}$. To achieve $\frac{r_{d'}}{r_d}\leq 1-\alpha$ with a fixed $\alpha$, it is easy to see that $1-\beta  $ should be no larger than $\frac{1}{1+(2\alpha-\alpha^2) d}$; this implies $\lim_{d\to\infty}\beta=1$. Similar to example (\rmnum{1}), the instance  
$P$ tends to be very stable as $d$ increases. 

\begin{remark}
In practice, it is difficult to know the exact value of $\beta$ for a fixed $\alpha$.  However, the value of $\beta$ only affects the sample sizes in our proposed algorithms in Section~\ref{sec-sub}, and thus only assuming a reasonable lower bound $\beta_0< \beta$ is already sufficient. In Section~\ref{sec-final-extent}, we also consider the general case without the stability assumption, where the proposed algorithm does not even need to input  $\beta_0$.  
\end{remark}

\subsection{A More Careful Analysis for Core-set Construction in~\cite{badoiu2003smaller}}
\label{sec-newanalysis}


We first briefly introduce the core-set construction for MEB, since it will be  used in our proposed algorithms. 
Let $0<\epsilon<1$. The algorithm in~\cite{badoiu2003smaller} yields an MEB core-set of size $2/\epsilon$ (for convenience, we always assume that $2/\epsilon$ is an integer). But there is a small issue in their paper. The analysis assumes that the exact MEB of the core-set is computed in each iteration, but in fact 
one may only compute an approximate MEB. Thus, an immediate question is whether the quality is still guaranteed with such a change. Kumar {\em et al.}~\cite{DBLP:journals/jea/KumarMY03} fixed this issue, and showed that computing a $(1+O(\epsilon^2))$-approximate MEB for the core-set in each iteration still guarantees a core-set with  size  $O(1/\epsilon)$, where the hidden constant is larger than $80$. Clarkson~\cite{C10}  showed that the greedy core-set construction algorithm of MEB, as a special case of  the  Frank-Wolfe algorithm, yields a core-set with size slightly larger than $ 4/\epsilon$. Note that there exist several other methods yielding even lower core-set size~\cite{coresets1,DBLP:conf/isaac/KerberS13}, but their construction algorithms are more complicated and thus not applicable to our problems. 
Below we show that it is possible to guarantee a core-set of~\cite{badoiu2003smaller}  with the size being arbitrarily close to $2/\epsilon$, even if we only compute an approximate MEB in each iteration. This improves the core-set sizes of \cite{DBLP:journals/jea/KumarMY03,C10}, and the new analysis is also interesting in its own right.


For the sake of completeness, we first briefly introduce the idea of the core-set construction algorithm in~\cite{badoiu2003smaller}.
Given a point set $P\subset\mathbb{R}^d$, the algorithm is a simple iterative procedure. Initially, it selects an arbitrary point from $P$ and places it into an initially empty set $T$. 
In each of the following $2/\epsilon$ iterations, the algorithm updates the center of $\mathbf{MEB}(T)$ and adds to $T$ the farthest point from the current center of $\mathbf{MEB}(T)$. 
Finally, the center of $\mathbf{MEB}(T)$ induces a $(1+\epsilon)$-approximation for $MEB(P)$. The selected set of $2/\epsilon$ points ({\em i.e.}, $T$) is called the core-set of MEB. To ensure the expected improvement in each iteration, they~\cite{badoiu2003smaller} showed that the following two inequalities hold if the algorithm always selects the farthest point to the current center of $\mathbf{MEB}(T)$:
\begin{eqnarray}
r_{i+1}  \geq  (1+\epsilon)\textbf{Rad}(P)-L_i; \hspace{0.2in} r_{i+1} \geq  \sqrt{r^2_i+L^2_i},\label{for-bc--2}
 \end{eqnarray}
where $r_i$ and $r_{i+1}$ are the radii of $\mathbf{MEB}(T)$ in the $i$-th and $(i+1)$-th iterations, respectively, and $L_i$ is the shifting distance of the center of $\mathbf{MEB}(T)$ from the $i$-th to $(i+1)$-th iteration.



\begin{figure}
\begin{center}
    \includegraphics[width=0.3\textwidth]{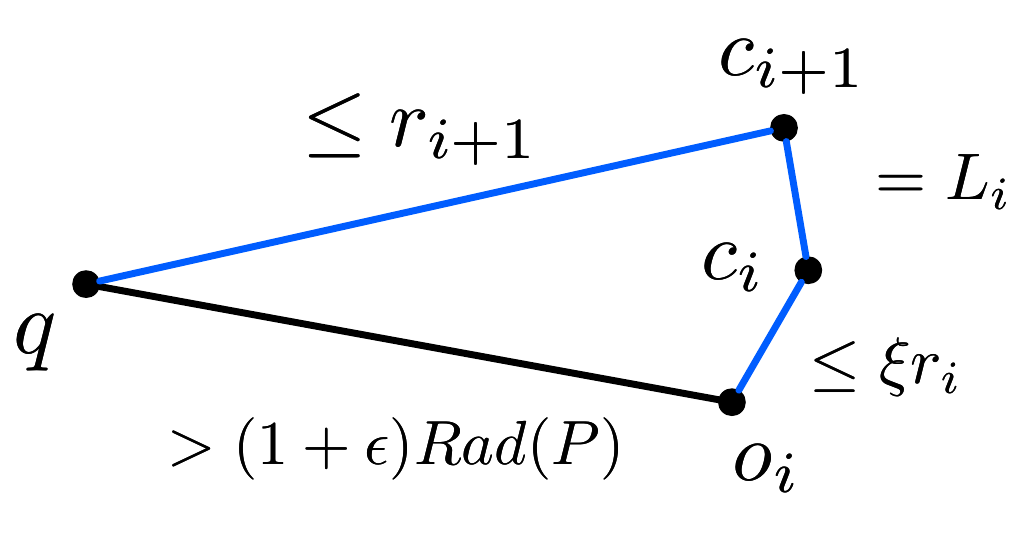}  
    \end{center}
  \caption{An illustration of (\ref{for-bc--1}).}     
   \label{fig-bc}
\end{figure}

As mentioned earlier, we often compute only an approximate $\mathbf{MEB}(T)$ in each iteration. In the $i$-th iteration, we let $c_i$ and $o_i$ denote the centers of the exact and the approximate $\mathbf{MEB}(T)$, 
respectively. Suppose that $||c_i-o_i||\leq \xi r_i$, where $\xi\in (0,\frac{\epsilon}{1+\epsilon})$ (we will see why  this bound is needed later). Using another algorithm proposed in~\cite{badoiu2003smaller}, one can obtain the point $o_i$ in $O(\frac{1}{\xi^2}|T|d)$ time. Note that we only compute $o_i$ rather than $c_i$ in each iteration. Hence we can only select the farthest point (say $q$) to $o_i$. If $||q-o_i||\leq (1+\epsilon)\textbf{Rad}(P)$,  we are done and a $(1+\epsilon)$-approximation of MEB is already obtained. Otherwise, we have
\begin{eqnarray}
(1+\epsilon)\textbf{Rad}(P)< ||q-o_i||\leq ||q-c_{i+1}||+||c_{i+1}-c_i||+||c_i-o_i||\leq r_{i+1}+L_i+\xi r_i \label{for-bc--1}
\end{eqnarray}
by  the triangle inequality (see Figure \ref{fig-bc}). In other words, we should replace the first inequality of (\ref{for-bc--2}) by ``$r_{i+1} > (1+\epsilon)\textbf{Rad}(P)-L_i-\xi r_i$''. Also, the second inequality of (\ref{for-bc--2}) still holds since it depends only on the property of the exact MEB (see~\cite[Lemma 2.1]{badoiu2003smaller}). Thus,  we have 
\begin{eqnarray}
r_{i+1}\geq \max\Big\{\sqrt{r^2_i+L^2_i}, (1+\epsilon)\textbf{Rad}(P)-L_i-\xi r_i\Big\}.\label{for-bc4}
\end{eqnarray}

This leads to the following theorem whose proof can be found 
 in Section~\ref{sec-proof-newbc}.

\begin{theorem}
\label{the-newbc}
In the core-set construction algorithm of~\cite{badoiu2003smaller}, if one computes an approximate MEB for $T$ in each iteration and the resulting center $o_i$ has the distance to $c_i$ less than $\xi r_i= s\frac{\epsilon}{1+\epsilon} r_i$ for some $s\in(0,1)$, the final core-set size is bounded by $z=\frac{2}{(1-s)\epsilon}$. Also, the bound could be arbitrarily close to $2/\epsilon$ when $s$ is small enough. 
\end{theorem}

We can simply set $s$ to be any constant in $(0,1)$; for instance, if $s=1/3$, the core-set size will be bounded by $z=3/\epsilon$. Since $|T|\leq z$ in each iteration, the total running time is $O\Big(z\big(|P|d+\frac{1}{\xi^2}zd\big)\Big)=O\Big(\frac{1}{\epsilon}\big(|P|+\frac{1}{\epsilon^3}\big)d\Big)$. 


\begin{remark}
\label{rem-newbc}
We also want to emphasize a simple observation on the above core-set construction procedure, which will be used in our algorithms and analyses later on. The algorithm always selects the farthest point to $o_i$ in each iteration. However, this is actually not necessary. As long as the selected point has distance at least $(1+\epsilon)\mathbf{Rad}(P)$, the result presented in Theorem~\ref{the-newbc} is still true.
 If no such a point exists ({\em i.e.,} $P\setminus \mathbb{B}\big(o_i, (1+\epsilon)\mathbf{Rad}(P)\big)=\emptyset$), a $(1+\epsilon)$-radius approximate MEB ({\em i.e.,} the ball $\mathbb{B}\big(o_i, (1+\epsilon)\mathbf{Rad}(P)\big)$) has been already  obtained.  
\end{remark}

\begin{remark}[kernels]
\label{rem-newbc2}
If each point $p\in P$ is mapped to $\psi(p)$ in $\mathbb{R}^D$ by some kernel function ({\em e.g.,} as the CVM~\cite{DBLP:journals/jmlr/TsangKC05}), where $D$ could be $+\infty$, we can still run the core-set algorithm of~\cite{badoiu2003smaller}, since the algorithm only needs to compute the distances and the center $o_i$ is always a convex combination of $T$ in each iteration; instead of returning an explicit center, the algorithm will output the coefficients of the convex combination for the center. And similarly, our Algorithm~\ref{alg-meb3} presented in Section~\ref{sec-secondesa} also works fine for  kernels. \end{remark}


\section{Implication of the Stability Property}
\label{sec-stable}


%

In this section, we show an important implication of the stability property of Definition~\ref{def-stable}.

\begin{theorem}
\label{the-stable}
Assume $\epsilon, \epsilon', \beta_0\in (0,1)$. Let $P$ be an $(\epsilon^2, \beta)$-stable instance of the MEB problem with $\beta> \beta_0$, and $o$ be the center of its MEB. 
Let $\tilde{o}$ be a given point in $\mathbb{R}^d$. Assume the number $r\leq (1+\epsilon'^2)\mathbf{Rad}(P)$. If the ball $\mathbb{B}\big(\tilde{o}, r\big)$ covers at least $(1-\beta_0)n$ points from $P$, the following holds
\begin{eqnarray}
||\tilde{o}-o||&<&  (2\sqrt{2} \epsilon+\sqrt{3}\epsilon') \mathbf{Rad}(P).\label{for-stable}
\end{eqnarray}
\end{theorem}

%

Theorem~\ref{the-stable} indicates that if a ball covers a large enough subset of $P$ and its radius is bounded, its center should be close to the center of $\mathbf{MEB}(P)$. 
  Let $P'=\mathbb{B}\big(\tilde{o}, r\big)\cap P$, and assume $o'$ is the center of $\mathbf{MEB}(P')$.
To bound the distance between $\tilde{o}$ and $o$, we bridge them by the point $o'$ (since $||\tilde{o}-o||\leq ||\tilde{o}-o'||+||o'-o||$).  
The following are two key lemmas for proving Theorem~\ref{the-stable}.

\begin{lemma}
\label{lem-stable1}
The distance $||o'-o||\leq \sqrt{2}\epsilon \mathbf{Rad}(P)$.
\end{lemma}

\begin{proof}

We consider two cases: $\mathbf{MEB}(P')$ is totally covered by $\mathbf{MEB}(P)$ and otherwise. 
For the first case (see Figure~\ref{fig-lem1}(a)), it is easy to see that 
\begin{eqnarray}
||o'-o||\leq \mathbf{Rad}(P)-(1-\epsilon^2) \mathbf{Rad}(P)=\epsilon^2 \mathbf{Rad}(P)<\sqrt{2}\epsilon \mathbf{Rad}(P), \label{for-lem-stable1-1}
\end{eqnarray}
where the first inequality comes from the fact that $\mathbf{MEB}(P')$ has radius at least $(1-\epsilon^2) \mathbf{Rad}(P)$ (Definition~\ref{def-stable}). Thus, we can focus on the second case below.

Let $a$ be any point located on the intersection of the two spheres of $\mathbf{MEB}(P')$ and $\mathbf{MEB}(P)$. Then we have the following claim.

\begin{figure}
  \begin{center}
\includegraphics[width=0.9\textwidth]{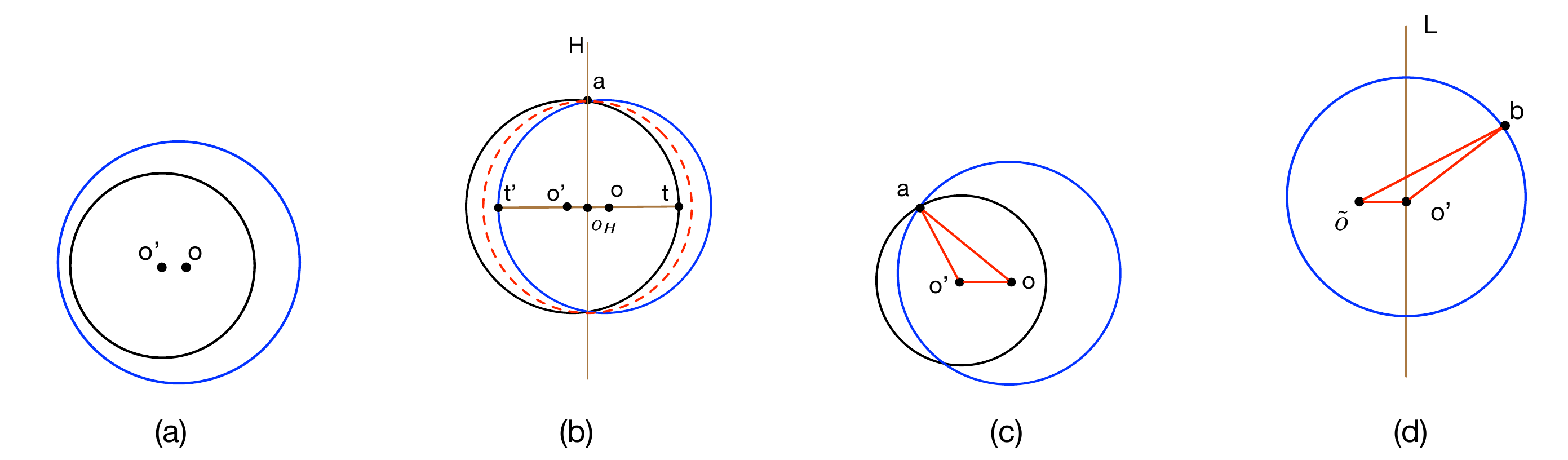}
   \end{center}
   \vspace{-0.1in}
  \caption{(a) The case $MEB(P')\subset MEB(P)$; (b) an illustration under the assumption $\angle ao'o< \pi/2$ in the proof of Claim~\ref{cla-angle}; (c) the angle $\angle ao'o\geq \pi/2$; (d) an illustration of Lemma~\ref{lem-stable2}.}
   \label{fig-lem1}
     \vspace{-0.12in}
\end{figure}

\begin{claim}
\label{cla-angle}
The angle $\angle ao'o\geq \pi/2$.
\end{claim}
\begin{proof}


Suppose  that $\angle ao'o< \pi/2$. Note that $\angle aoo'$ is always smaller than $\pi/2$ since $||o-a||=\mathbf{Rad}(P)\geq \mathbf{Rad}(P')=||o'-a||$. Therefore, $o$ and $o'$ are separated by the hyperplane $H$ that is  orthogonal 
to the segment $\overline{o'o}$ and passes through the point $a$. See Figure~\ref{fig-lem1}(b). 
Now we  show that $P'$ can be covered by a ball smaller than $\mathbf{MEB}(P')$. Let $o_H$ be the point $H\cap \overline{o'o}$, and $t$ ({\em resp.,} $t'$) be the point collinear with $o$ and $o'$ on the right side of the sphere of $\mathbf{MEB}(P')$ ({\em resp.,} left side of the sphere of $\mathbf{MEB}(P)$; see Figure ~\ref{fig-lem1}(b)). Then, we have
\begin{eqnarray}
||t-o_H||+||o_H-o'||&=&||t-o'||=||a-o'||<||o'-o_H||+||o_H-a||\nonumber\\
\Longrightarrow ||t-o_H||&<&||o_H-a||.\label{cla-1}
\end{eqnarray}
Similarly, we have 
$||t'-o_H||<||o_H-a||$. 
Consequently, $\mathbf{MEB}(P)\cap \mathbf{MEB}(P')$ is covered by the ball $\mathbb{B}(o_H, ||o_H-a||)$ (the ``red dotted'' ball in Figure ~\ref{fig-lem1}(b)). Further, because $P'$ is covered by  $\mathbf{MEB}(P)\cap \mathbf{MEB}(P')$ and $||o_H-a||<||o'-a||=\mathbf{Rad}(P')$, $P'$ is covered by the ball $\mathbb{B}(o_H, ||o_H-a||)$ that is smaller than $\mathbf{MEB}(P')$. This  contradicts to the fact that $\mathbf{MEB}(P')$ is the minimum enclosing ball of $P'$. Thus, the claim $\angle ao'o\geq \pi/2$ is true.
 \qed\end{proof}

Given Claim~\ref{cla-angle}, we know that $||o'-o||\leq \sqrt{\big(\mathbf{Rad}(P)\big)^2-\big(\mathbf{Rad}(P')\big)^2}$. See Figure~\ref{fig-lem1}(c). Moreover, Definition~\ref{def-stable} implies that $\mathbf{Rad}(P')\geq (1-\epsilon^2) \mathbf{Rad}(P)$. Therefore, we have 
\begin{eqnarray}
||o'-o|| \leq \sqrt{\big(\mathbf{Rad}(P)\big)^2-\big((1-\epsilon^2) \mathbf{Rad}(P)\big)^2}\leq\sqrt{2}\epsilon \mathbf{Rad}(P).
\end{eqnarray}
 \qed\end{proof}

\begin{lemma}
\label{lem-stable2}
The distance $||\tilde{o}-o'||< (\sqrt{2}\epsilon+\sqrt{3}\epsilon')  \mathbf{Rad}(P)$.
\end{lemma}
\begin{proof}

Let $L$ be the hyperplane orthogonal to
the segment $\overline{\tilde{o}o'}$ and passing through the center $o'$. Suppose $\tilde{o}$ is located on the left side of $L$. Then, there always exists a point $b\in P'$ located on the right closed semi-sphere of $\mathbf{MEB}(P')$ divided by $L$ (this result is from~\cite[Lemma 2.2]{BHI}; for completeness, we state the lemma in Section~\ref{sec-bhi}). See Figure \ref{fig-lem1}(d). That is, the angle $\angle bo'\tilde{o}\geq \pi/2$. As a consequence, we have 
\begin{eqnarray}
||\tilde{o}-o'||\leq\sqrt{||\tilde{o}-b||^2-||b-o'||^2}. \label{for-lem-stable2-1}
\end{eqnarray}
Moreover, since $||\tilde{o}-b||\leq r\leq (1+\epsilon'^2)\mathbf{Rad}(P)$ and $||b-o'||= \mathbf{Rad}(P')\geq (1-\epsilon^2) \mathbf{Rad}(P)$, (\ref{for-lem-stable2-1}) implies that $||\tilde{o}-o'||\leq \sqrt{(1+\epsilon'^2)^2-(1-\epsilon^2)^2} \mathbf{Rad}(P)$, where this upper bound is equal to 
\begin{eqnarray}
\sqrt{2\epsilon'^2+\epsilon'^4+2\epsilon^2-\epsilon^4} \mathbf{Rad}(P)<\sqrt{3\epsilon'^2+2\epsilon^2} \mathbf{Rad}(P)<(\sqrt{2}\epsilon+\sqrt{3}\epsilon') \mathbf{Rad}(P).
\end{eqnarray} 
 \qed\end{proof}

By triangle inequality, Lemmas~\ref{lem-stable1} and~\ref{lem-stable2}, we immediately have 
\begin{eqnarray}
||\tilde{o}-o||&\leq& ||\tilde{o}-o'||+||o'-o||< (2\sqrt{2} \epsilon+\sqrt{3}\epsilon')\mathbf{Rad}(P).
\end{eqnarray}
This completes the proof of Theorem~\ref{the-stable}.

\section{Sublinear Time Algorithms for MEB under Stability Assumption}
\label{sec-sub}
Suppose $\epsilon\in (0,1)$. We assume that the given instance $P$ is an $(\epsilon^2, \beta)$-stable instance where $\beta$ is larger than a given lower bound $\beta_0$ ({\em i.e.,} $\beta>\beta_0$). Using Theorem~\ref{the-stable}, we present two different sublinear time sampling algorithms for computing MEB. Following most of the articles on sublinear time algorithms ({\em e.g.,}~\cite{meyerson2004k,mishra2001sublinear,czumaj2004sublinear}), in each sampling step of our algorithms, we always take the sample \textbf{independently and uniformly at random}.


 \subsection{The First Algorithm}
 \label{sec-firstesa}
 
  \begin{figure}[h] 
\begin{center}
    \includegraphics[width=0.22\textwidth]{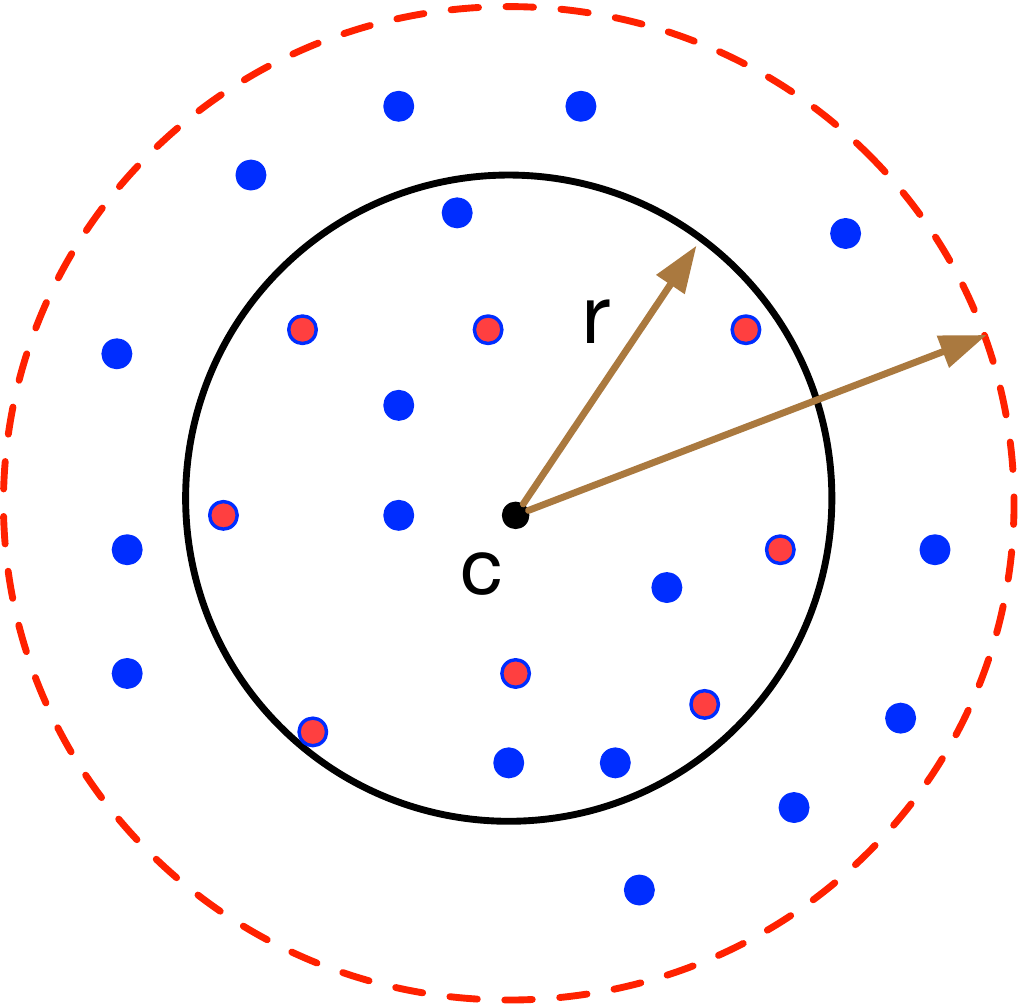}  
    \end{center}
\vspace{-0.13in}
  \caption{An illustration for the first sampling algorithm. The red points are the samples; we expand $\mathbb{B}(c, r)$ slightly and the larger ball is a radius-approximate MEB of the whole input point set.}     
   \label{fig-illustratefirst}
  \vspace{-0.1in}
\end{figure}

The first algorithm is based on the theory of VC dimension and $\epsilon$-nets~\cite{vapnik2015uniform,haussler1987}. 
Roughly speaking, we compute an approximate MEB of a small random sample (say, $\mathbb{B}(c, r)$), and expand the ball slightly; then we prove that this expanded ball is an approximate MEB of the whole data set (see Figure~\ref{fig-illustratefirst}).  Our key idea is to show that $\mathbb{B}(c, r)$ covers at least $(1-\beta_0 )n$ points and therefore $c$ is close to the optimal center by Theorem~\ref{the-stable}. 
As emphasized in Section~\ref{sec-overview}, our result is a single-criterion approximation. If simply applying the uniform sample idea without the stability assumption (as the ideas in~\cite{DBLP:conf/focs/HuangJLW18,DBLP:journals/corr/abs-1901-08219}), it will yield a bi-criteria approximation where the ball has to cover less than $n$ points for achieving the desired bounded radius. 

\renewcommand{\algorithmicrequire}{\textbf{Input:}}
\renewcommand{\algorithmicensure}{\textbf{Output:}}
\begin{algorithm}
   \caption{\textbf{MEB Algorithm \Rmnum{1}}}
   \label{alg-meb1}
\begin{algorithmic}[1]
\REQUIRE Two parameters $0<\epsilon, \eta<1$; an $(\epsilon^2, \beta)$-stable instance $P$ of MEB problem in $\mathbb{R}^d$, where $\beta$ is larger than a given lower bound $\beta_0>0$.
\STATE Sample a set $S$ of $\Theta(\frac{1}{\beta_0}\cdot\max\{\log\frac{1}{\eta},d\log \frac{d}{\beta_0}\})$ points from $P$ uniformly at random.
\STATE Apply any approximate MEB algorithm (such as the core-set based algorithm~\cite{badoiu2003smaller}) to compute a $(1+\epsilon^2)$-radius approximate MEB of $S$, and let the obtained ball be $\mathbb{B}(c, r)$.
\STATE Output the ball $\mathbb{B}\big(c,\frac{1+(2\sqrt{2}+\sqrt{3}) \epsilon}{1-\epsilon^2}r\big)$.
\end{algorithmic}
\end{algorithm}

\begin{theorem}
\label{the-meb1}
With  probability $1-\eta$, Algorithm~\ref{alg-meb1} returns a $\lambda$-radius approximate MEB of $P$, where
\begin{eqnarray}
\lambda=\frac{\big(1+( 2\sqrt{2}+\sqrt{3}) \epsilon\big)(1+\epsilon^2)}{1-\epsilon^2} =1+O( \epsilon).
\end{eqnarray}
\end{theorem}

Before proving Theorem~\ref{the-meb1}, we prove the following lemma first.
\begin{lemma}
\label{lem-net}
Let $S$ be a set of $\Theta(\frac{1}{\beta_0}\cdot\max\{\log\frac{1}{\eta},d\log \frac{d}{\beta_0}\})$ points sampled randomly and independently from a given point set $P\subset\mathbb{R}^d$, and $B$ be any ball covering $S$. Then, with probability $1-\eta$, $|B\cap P|\geq (1-\beta_0)|P|$.
\end{lemma}
\begin{proof}
Consider the range space $\Sigma=(P, \Phi)$ where each range $\phi\in \Phi$ is the complement of a ball in the space. In a range space, a subset $Y\subset P$ is a $\beta_0$-net if 
\begin{eqnarray}
\text{for any $\phi\in \Phi$, 
$\frac{|P\cap\phi|}{|P|}\geq \beta_0\Longrightarrow Y\cap\phi\neq\emptyset$. }
 \end{eqnarray}
The size $|S|=\Theta(\frac{1}{\beta_0}\cdot\max\{\log\frac{1}{\eta},d\log \frac{d}{\beta_0}\})$, and from~\cite{vapnik2015uniform,haussler1987} we know that $S$ is a $\beta_0$-net of $P$ with probability $1-\eta$. Thus, if $|B\cap P|< (1-\beta_0)|P|$, {\em i.e.,} $|P\setminus B|> \beta_0 |P|$, we have $S\cap \big(P\setminus B\big)\neq\emptyset$. This contradicts to  the fact that $S$ is covered by $B$. Consequently, $|B\cap P|\geq (1-\beta_0)|P|$.
 \qed\end{proof}

\begin{proof}(\textbf{of Theorem~\ref{the-meb1}})
Denote by $o$ the center of $\mathbf{MEB}(P)$. Since $S\subset P$ and $\mathbb{B}(c, r)$ is a $(1+\epsilon^2)$-radius approximate MEB of $S$, we know that $r\leq (1+\epsilon^2)\mathbf{Rad}(P)$. Moreover, Lemma~\ref{lem-net} implies that $|\mathbb{B}(c, r)\cap P|\geq (1-\beta_0)|P|$ with probability $1-\eta$. Suppose it is true and let $P'=\mathbb{B}(c, r)\cap P$. Then, we have the distance
\begin{eqnarray}
||c-o||\leq ( 2\sqrt{2}+\sqrt{3}) \epsilon   \mathbf{Rad}(P) \label{for-the-meb1-1}
\end{eqnarray}
via Theorem~\ref{the-stable} (we set $\epsilon'=\epsilon$). For simplicity, we use $x$ to denote $(2\sqrt{2}+\sqrt{3}) \epsilon $. The inequality (\ref{for-the-meb1-1}) implies that the point set $P$ is covered by the ball $\mathbb{B}(c, (1+x)\mathbf{Rad}(P))$. Note that we cannot directly return $\mathbb{B}(c, (1+x)\mathbf{Rad}(P))$ as the final result, since we do not know the value of $\mathbf{Rad}(P)$. Thus, we have to estimate the radius $(1+x)\mathbf{Rad}(P)$.

Since $P'$ is covered by $\mathbb{B}(c, r)$ and $|P'|\geq (1-\beta_0)|P|$, $r$ should be at least $(1-\epsilon^2)\mathbf{Rad}(P)$ due to Definition~\ref{def-stable}. Hence, we have
\begin{eqnarray}
\frac{1+x}{1-\epsilon^2}r\geq (1+x)\mathbf{Rad}(P).
\end{eqnarray}
That is, $P$ is covered by the ball $\mathbb{B}(c, \frac{1+x}{1-\epsilon^2}r)$. Moreover, the radius 
\begin{eqnarray}
 \frac{1+x}{1-\epsilon^2}r\leq  \frac{1+x}{1-\epsilon^2}(1+\epsilon^2)\mathbf{Rad}(P).  
  \end{eqnarray}
This means the ball $\mathbb{B}(c, \frac{1+x}{1-\epsilon^2}r)$ is a $\lambda$-radius approximate MEB of $P$, where
\begin{eqnarray}
\lambda&=&(1+\epsilon^2) \frac{1+x}{1-\epsilon^2}=\frac{\big(1+(2\sqrt{2}+\sqrt{3}) \epsilon \big)(1+\epsilon^2)}{1-\epsilon^2} 
\end{eqnarray}
and $\lambda=1+O( \epsilon )$ if $\epsilon$ is a fixed small number in $(0,1)$. 
\qed\end{proof}

\textbf{Running time of Algorithm~\ref{alg-meb1}.} For simplicity, we assume $\log\frac{1}{\eta}<d\log \frac{d}{\beta_0}$. If we use the core-set based algorithm~\cite{badoiu2003smaller} to compute $\mathbb{B}(c, r)$ (see Remark~\ref{rem-newbc}), the running time of Algorithm~\ref{alg-meb1} is $O\big(\frac{1}{\epsilon^2}(|S|d+\frac{1}{\epsilon^6}d)\big)=O\big(\frac{d^2}{\epsilon^2\beta_0}\log\frac{d}{\beta_0}+\frac{d}{\epsilon^8}\big)=\tilde{O}(d^2)$ where the hidden factor depends on $\epsilon$ and $\beta_0$.  
\begin{remark}
\label{rem-the-meb1}
 If the dimensionality $d$ is too high, the random projection technique {\em Johnson-Lindenstrauss (JL) transform}~\cite{dasgupta2003elementary} can be used to approximately preserve the radius of enclosing ball~\cite{DBLP:conf/compgeom/AgarwalHY07,DBLP:conf/cccg/KerberR15,DBLP:conf/compgeom/Sheehy14}. However, it is not  useful for reducing the time complexity of Algorithm~\ref{alg-meb1}. If we apply the JL-transform on the sampled $\Theta(\frac{d}{\beta_0}\log \frac{d}{\beta_0})$ points in Step 1, the JL-transform step itself already takes $\Omega(\frac{d^2}{\beta_0}\log \frac{d}{\beta_0})$ time.
\end{remark}

\subsection{The Second Algorithm}
\label{sec-secondesa}

Our first algorithm in Section~\ref{sec-firstesa} is simple, but has a sample size ({\em i.e.,} the number of sampled points) depending on the dimensionality $d$,
while \textbf{the second algorithm has a sample size independent of both $n$ and $d$} (it is particularly important when a kernel function is applied, because the new dimension could be very large or even $+\infty$). 
We briefly overview our idea first.


\vspace{0.05in}
\textbf{High level idea of the second algorithm:}  Recall our Remark~\ref{rem-newbc} (\rmnum{2}). If we know the value of $(1+\epsilon)\mathbf{Rad}(P)$, we can perform almost the same core-set construction procedure described in Theorem~\ref{the-newbc} to achieve an approximate center of $\mathbf{MEB}(P)$, where the only difference is that we add a point with distance at least $(1+\epsilon)\mathbf{Rad}(P)$ to $o_i$ in each iteration. In this way, we avoid selecting the farthest point to $o_i$, since this operation will inevitably have a linear time complexity. To implement our strategy in sublinear time, we need to determine the value of  $(1+\epsilon)\mathbf{Rad}(P)$ first. We propose Lemma~\ref{lem-upper} below to estimate the range of $\mathbf{Rad}(P)$, and then perform a binary search on the range to determine the value of  $(1+\epsilon)\mathbf{Rad}(P)$ approximately. Based on the stability property, we observe that the core-set construction procedure can serve as an ``oracle'' to help us to guess the value of  $(1+\epsilon)\mathbf{Rad}(P)$ (see Algorithm~\ref{alg-meb2}). Let $h>0$ be a candidate.  We add a point with distance at least $h$ to $o_i$ in each iteration. We prove that the procedure cannot continue for more than $z$ iterations if $h\geq (1+\epsilon)\mathbf{Rad}(P)$, and will continue more than $z$ iterations with constant probability if $h<(1-\epsilon)\mathbf{Rad}(P)$, where $z$ is the size of core-set described in Theorem~\ref{the-newbc}. Also, during the core-set construction, we add the points to the core-set via random sampling, rather than a deterministic way.  A minor issue here is that we need to replace $\epsilon$ by $\epsilon^2$ in Theorem~\ref{the-newbc}, so as to achieve the overall $(1+O(\epsilon))$-radius approximation in the following analysis. 
\vspace{0.05in}

\begin{lemma}
\label{lem-upper}
Given a parameter $\eta\in (0,1)$, one selects an arbitrary point $p_1\in P$ and takes a  sample $Q\subset P$ with $|Q|=\frac{1}{\beta_0}\log\frac{1}{\eta}$ uniformly at random. Let $p_2=\arg\max_{p\in Q} ||p-p_1||$. Then, with probability $1-\eta$, 
\begin{eqnarray}
\mathbf{Rad}(P)\in [\frac{1}{2}||p_1-p_2||, \frac{1}{1-\epsilon^2}||p_1-p_2||].
\end{eqnarray}
\end{lemma}
\begin{proof}

First, the lower bound of $\mathbf{Rad}(P)$ is obvious since $||p_1-p_2||$ is always no larger than $2\mathbf{Rad}(P)$. Then, we consider the upper bound. Let $\mathbb{B}(p_1, l)$ be the ball covering exactly $(1-\beta_0)n$ points of $P$, and thus $l\geq (1-\epsilon^2)\mathbf{Rad}(P)$ according to Definition~\ref{def-stable}. To complete our proof, we also need the following folklore lemma presented in~\cite{DX14}.

%

\begin{lemma} {\cite{DX14}}
\label{cla-sampling}
Let $N$ be a set of elements, and $N'$ be a subset of $N$ with
size $\left|N'\right|=\tau\left|N\right|$ for some $\tau\in(0,1)$. Given $\eta\in (0,1)$, if one randomly samples $\frac{\ln1/\eta}{\ln1/(1-\tau)}\leq\frac{1}{\tau}\ln\frac{1}{\eta}$ elements from $N$, then with probability at least $1-\eta$, the sample contains at least one element of $N'$.
\end{lemma}

  \begin{figure}[h] 
\begin{center}
    \includegraphics[width=0.25\textwidth]{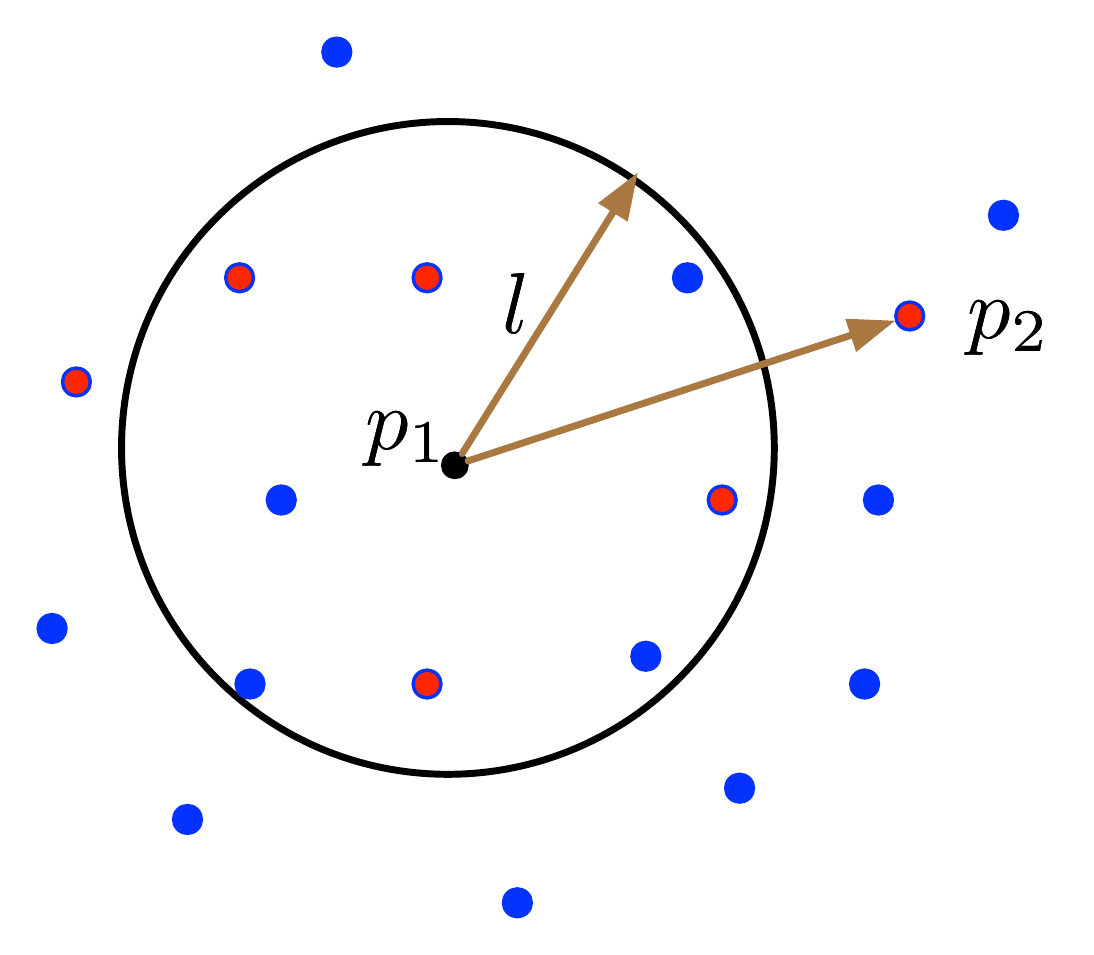}  
    \end{center}
\vspace{-0.13in}
  \caption{An illustration of Lemma~\ref{lem-upper}; the red points are the sampled set $Q$.}     
   \label{fig-illustratefirst-b}
   \vspace{-0.1in}
\end{figure}

In Lemma~\ref{cla-sampling}, let $N$  and $N'$ be the point set $P$ and the subset $P\setminus \mathbb{B}(p_1, l)$, respectively. We know that $Q$ contains at least one point from $N'$ according to Lemma~\ref{cla-sampling} (by setting $\tau=\beta_0$). Namely, $Q$ contains at least one point outside $\mathbb{B}(p_1, l)$. Moreover, because $p_2=\arg\max_{p\in Q} ||p-p_1||$, we have $||p_1-p_2||\geq l\geq (1-\epsilon^2)\mathbf{Rad}(P)$, {\em i.e.}, $\mathbf{Rad}(P)\leq \frac{1}{1-\epsilon^2}||p_1-p_2||$ (see  Figure~\ref{fig-illustratefirst-b} for an illustration).
 \qed\end{proof}

%
%

Algorithm~\ref{alg-meb2} serves as a subroutine in Algorithm~\ref{alg-meb3}. In Algorithm~\ref{alg-meb2}, we simply set $z=\frac{3}{\epsilon^2}$ with $s=1/3$ as described in Theorem~\ref{the-newbc} (as mentioned before, we replace $\epsilon$ by $\epsilon^2$); we compute $o_i$ having distance less than $s\frac{\epsilon^2}{1+\epsilon^2}\mathbf{Rad}(T)$ to the center of $\mathbf{MEB}(T)$ in Step 2(1).

\renewcommand{\algorithmicrequire}{\textbf{Input:}}
\renewcommand{\algorithmicensure}{\textbf{Output:}}

\begin{algorithm}[H]
   \caption{\textbf{MEB Algorithm \Rmnum{2}}}
   \label{alg-meb3}
\begin{algorithmic}[1]
\REQUIRE  Two parameters $0<\epsilon, \eta_0<1$; an $(\epsilon^2, \beta)$-stable instance $P$ of MEB problem in $\mathbb{R}^d$, where $\beta$ is larger than a given lower bound $\beta_0>0$. 
Set  the interval $[a, b]$ for $\mathbf{Rad}(P)$ that is obtained by Lemma~\ref{lem-upper}.
\STATE Among the set $\{(1-\epsilon^2)a, (1+\epsilon^2)(1-\epsilon^2)a, \cdots, (1+\epsilon^2)^w (1-\epsilon^2)a=(1+\epsilon^2)b\}$ where $w=\lceil \log_{1+\epsilon^2}\frac{2}{(1-\epsilon^2)^2}\rceil+1=O(\frac{1}{\epsilon^2})$, perform binary search for the value $h$ by using Algorithm~\ref{alg-meb2} with $z=\frac{3}{\epsilon^2}$ and $\eta=\frac{\eta_0}{2\log w}$. 
\STATE Suppose that Algorithm~\ref{alg-meb2} returns ``no'' when $h=(1+\epsilon^2)^{i_0} (1-\epsilon^2)a$ and returns ``yes'' when $h=(1+\epsilon^2)^{i_0+1} (1-\epsilon^2)a$.
\STATE Run Algorithm~\ref{alg-meb2} again with $h=(1+\epsilon^2)^{i_0+2}a$, $z=\frac{3}{\epsilon^2}$, and $\eta=\eta_0/2$; let $\tilde{o}$ be the obtained ball center of $T$ when the loop stops.
\STATE Return the ball $\mathbb{B}(\tilde{o}, r)$, where $r=\frac{1+(2\sqrt{2}+\frac{2\sqrt{6}}{\sqrt{1-\epsilon^2}})\epsilon}{1+\epsilon^2}h$.
\end{algorithmic}
\end{algorithm}

\vspace{-0.2in}
 \begin{algorithm}[H]
    \caption{\textbf{Oracle for testing $h$}}
   \label{alg-meb2}
\begin{algorithmic}[1]
\REQUIRE  An  instance $P$,  a parameter  $\eta\in (0,1)$, $h>0$, and a positive integer $z$.
\STATE Initially, arbitrarily select a point $p\in P$ and let $T=\{p\}$.
\STATE $i=1$; repeat the following steps:
\begin{enumerate}[(1)]
\item Compute an approximate MEB of $T$ and let the ball center be $o_i$ as described in Theorem~\ref{the-newbc} (replace $\epsilon$ by $\epsilon^2$ and set $s=1/3$).
\item Sample a set $Q\subset P$ with $|Q|=\frac{1}{\beta_0}\log\frac{z}{\eta}$ uniformly at random.
\item Select the point $q\in Q$ that is farthest to $o_i$, and add it to $T$.
\item If $||q-o_i||< h$, stop the loop and output ``yes''.
\item $i=i+1$; if $i>z$, stop the loop and output ``no''.
\end{enumerate}
\end{algorithmic}
\end{algorithm}
\begin{theorem}
\label{the-sample2}
With probability $1-\eta_0$, Algorithm~\ref{alg-meb3} returns a $\lambda$-radius approximate MEB of $P$, where 
\begin{eqnarray}
\lambda=\frac{(1+x_1)(1+x_2)}{1+\epsilon^2}=1+O( \epsilon) \text{\hspace{0.1in} with \hspace{0.1in}} x_1=O\big(\frac{\epsilon^2}{1-\epsilon^2}\big), x_2=O\big(\frac{\epsilon}{\sqrt{1-\epsilon^2}}\big).
\end{eqnarray}
 The running time is $\tilde{O}\big((\frac{1}{\epsilon^2\beta_0}+\frac{1}{\epsilon^8})d\big)$, where  $\tilde{O}(f)=O(f\cdot \mathtt{polylog}(\frac{1}{\epsilon}, \frac{1}{\eta_0}))$. 
\end{theorem}

Before proving Theorem~\ref{the-sample2}, we provide Lemma~\ref{lem-sample2} first. 

\begin{lemma}
\label{lem-sample2}
If $h\geq (1+\epsilon^2)\mathbf{Rad}(P)$, Algorithm~\ref{alg-meb2} returns ``yes''; else if $h<(1-\epsilon^2)\mathbf{Rad}(P)$, Algorithm~\ref{alg-meb2} returns ``no'' with probability at least $1-\eta$.
\end{lemma}
\begin{proof}
First, we assume that $h\geq (1+\epsilon^2)\mathbf{Rad}(P)$. 
Recall the remark following Theorem~\ref{the-newbc}. 
If we always add a point $q$ with distance at least $h\geq (1+\epsilon^2)\mathbf{Rad}(P)$ to $o_i$, the loop 2(1)-(5) cannot continue more than $z$ iterations, {\em i.e.}, Algorithm~\ref{alg-meb2} will return ``yes''.

Now, we consider the case $h<(1-\epsilon^2)\mathbf{Rad}(P)$. Similar to the proof of Lemma~\ref{lem-upper}, we consider the ball $\mathbb{B}(o_i, l)$ covering exactly $(1-\beta_0)n$ points of $P$. According to Definition~\ref{def-stable}, we know that $l\geq (1-\epsilon^2)\mathbf{Rad}(P)>h$. Also, with probability $1-\eta/z$, the sample $Q$ contains at least one point outside $\mathbb{B}(o_i, l)$ due to Lemma~\ref{cla-sampling}. By taking the union bound, with probability $(1-\eta/z)^z\geq 1-\eta$, $||q-o_i||$ is always larger than $h$ and eventually Algorithm~\ref{alg-meb2} will return ``no''.
 \qed\end{proof}

\begin{proof}\textbf{(of Theorem~\ref{the-sample2})}
Since Algorithm~\ref{alg-meb2} returns ``no'' when $h=(1+\epsilon^2)^{i_0} (1-\epsilon^2)a$ and returns ``yes'' when $h=(1+\epsilon^2)^{i_0+1} (1-\epsilon^2)a$, from Lemma~\ref{lem-sample2} we know that 
\begin{eqnarray}
(1+\epsilon^2)^{i_0} (1-\epsilon^2)a&<&(1+\epsilon^2)\mathbf{Rad}(P); \label{for-the-sample2-2}\\
(1+\epsilon^2)^{i_0+1} (1-\epsilon^2)a&\geq& (1-\epsilon^2)\mathbf{Rad}(P).\label{for-the-sample2-1}
\end{eqnarray}
The above inequalities together imply that
\begin{eqnarray}
\frac{(1+\epsilon^2)^3}{1-\epsilon^2}\mathbf{Rad}(P)>(1+\epsilon^2)^{i_0+2}a\geq (1+\epsilon^2)\mathbf{Rad}(P).\label{for-the-sample2-3}
\end{eqnarray}
Thus, when running Algorithm~\ref{alg-meb2} with $h=(1+\epsilon^2)^{i_0+2}a$ in Step 3, the algorithm returns ``yes'' (by the right hand-side of (\ref{for-the-sample2-3})). Then, consider the ball $\mathbb{B}(\tilde{o}, h)$. We claim that $|P\setminus\mathbb{B}(\tilde{o}, h)|<\beta_0 n$. Otherwise, the sample $Q$ contains at least one point outside $\mathbb{B}(\tilde{o}, h)$ with probability $1-\eta/z$ in Step 2(2) of Algorithm~\ref{alg-meb2}, {\em i.e.,} the loop will continue. Thus, it  contradicts to the fact that the algorithm returns ``yes''. Let 
$P'=P\cap\mathbb{B}(\tilde{o}, h)$, and then $|P'|\geq(1-\beta_0)n$. Moreover, the left hand-side of (\ref{for-the-sample2-3}) indicates that
\begin{eqnarray}
h=(1+\epsilon^2)^{i_0+2}a<(1+\frac{8\epsilon^2}{1-\epsilon^2})\mathbf{Rad}(P).\label{for-the-sample2-4}
\end{eqnarray}
Now, we can apply Theorem~\ref{the-stable}, where we set ``$\epsilon'$'' to be ``$\sqrt{\frac{8\epsilon^2}{1-\epsilon^2}}$'' in the theorem. Let $o$ be the center of $\mathbf{MEB}(P)$. Consequently, we have
\begin{eqnarray}
||\tilde{o}-o||<(2\sqrt{2}+ 2\sqrt{6}/\sqrt{1-\epsilon^2})\epsilon\cdot \mathbf{Rad}(P).\label{for-the-sample2-5}
\end{eqnarray}

For simplicity, we let $x_1=\frac{8\epsilon^2}{1-\epsilon^2}$ and $x_2=(2\sqrt{2}+ 2\sqrt{6}/\sqrt{1-\epsilon^2})\epsilon $. Hence, $h\leq (1+x_1)\mathbf{Rad}(P)$ and $||\tilde{o}-o||\leq x_2\mathbf{Rad}(P)$ in (\ref{for-the-sample2-4}) and (\ref{for-the-sample2-5}). From (\ref{for-the-sample2-5}), we know that $P\subset\mathbb{B}(\tilde{o}, (1+x_2)\mathbf{Rad}(P))$.  From the right hand-side of (\ref{for-the-sample2-3}), we know that $(1+x_2)\mathbf{Rad}(P)\leq\frac{1+x_2}{1+\epsilon^2}h$. Thus, we have 
$P\subset\mathbb{B}\Big(\tilde{o}, \frac{1+x_2}{1+\epsilon^2}h\Big)$  
where $\frac{1+x_2}{1+\epsilon^2}h=\frac{1+(2\sqrt{2}+\frac{2\sqrt{6}}{\sqrt{1-\epsilon^2}})\epsilon}{1+\epsilon^2}h$. Also, the radius
\begin{eqnarray}
\frac{1+x_2}{1+\epsilon^2}h&\underbrace{\leq}_{\text{by (\ref{for-the-sample2-4})}}&\frac{(1+x_2)(1+x_1)}{1+\epsilon^2}\mathbf{Rad}(P)=\lambda \cdot\mathbf{Rad}(P).
\end{eqnarray}
Thus  $\mathbb{B}\Big(\tilde{o}, \frac{1+x_2}{1+\epsilon^2}h\Big)$ is a $\lambda$-radius approximate MEB of $P$, and $\lambda=1+O( \epsilon )$ if $\epsilon$ is a fixed small number in $(0,1)$. 

\vspace{0.05in}
\textbf{Success probability.} The success probability of Algorithm~\ref{alg-meb2} is $1-\eta$. In Algorithm~\ref{alg-meb3}, we set $\eta=\frac{\eta_0}{2\log w}$ in Step 1 and $\eta=\eta_0/2$ in Step 3, respectively. We take the union bound and the success probability of Algorithm~\ref{alg-meb3} is $(1-\frac{\eta_0}{2\log w})^{\log w} (1-\eta_0/2)>1-\eta_0$.

\vspace{0.05in}
\textbf{Running time.} As the subroutine, Algorithm~\ref{alg-meb2} runs in $O(z(\frac{1}{\beta_0}(\log\frac{z}{\eta}) d+\frac{1}{\epsilon^6}d))$ time; Algorithm~\ref{alg-meb3} calls the subroutine $O\big(\log(\frac{1}{\epsilon^2})\big)$ times. Note that $z=O(\frac{1}{\epsilon^2})$. Thus, the total running time  is $\tilde{O}\big((\frac{1}{\epsilon^2\beta_0}+\frac{1}{\epsilon^8})d\big)$. 
\qed\end{proof}

\section{Sublinear Time Algorithm for General MEB}
\label{sec-final-extent}

%
In Section~\ref{sec-sub}, we propose the sublinear time algorithms under the stability assumption. Specifically, we assume that the given instance is $(\epsilon^2, \beta)$-stable and $\beta$ is  larger than a  reasonable known lower bound $\beta_0$.  However, when $\beta_0$'s value is unknown, we cannot not determine the sample size for the algorithm; or we may only know a trivial lower bound, {\em e.g.,} $\frac{1}{n}$, and then the sample size could be too large. So 
in this section we  consider the general case without the stability assumption.

\textbf{High-level idea.} An interesting observation is that the ideas developed for stable instance can even help us to develop a hybrid approach for MEB  when the stability assumption does not hold. First, we ``suppose'' the input instance is $(\alpha, \beta)$-stable where ``$\alpha$'' and ``$\beta$'' are  designed based on the pre-specified radius error bound $\epsilon$ and covering error bound $\delta$, and compute a ``potential'' $(1+\epsilon)$-radius approximation (say a ball $B_1$); then we  compute a ``potential'' $(1-\delta)$-covering approximation (say a ball $B_2$), where the  definition of ``covering approximation'' is given in Definition~\ref{def-app}; finally, we determine the final output based on the ratio of their radii. Specifically, we set  a threshold $\tau$ that is determined by the given radius error bound $\epsilon$.  If the ratio is no larger than  $\tau$, we can infer that $B_1$ is a ``true''  $(1+\epsilon)$-radius approximation and return it; otherwise, we return $B_2$ that is a ``true''  $(1-\delta)$-covering approximation.  Moreover, for the latter case ({\em i.e.,} returning a $(1-\delta)$-covering approximation),  we will show that our  proposed algorithm yields a radius not only being strictly smaller than $\mathbf{Rad}(P)$, but also having a gap of $\Theta(\epsilon^2)\cdot\mathbf{Rad}(P)$ to $\mathbf{Rad}(P )$ ({\em i.e.,} the returned radius is at most $\big(1-\Theta(\epsilon^2)\big)\cdot\mathbf{Rad}(P)$). 
%
%
Our algorithm  only needs  uniform sampling and a single pass over the input data, where the space complexity in memory is $O(d)$ (the hidden factor depends on $\epsilon$ and $\delta$); if the input data matrix is sparse ({\em i.e.,} $M=o(nd)$), the time complexity is sublinear.

%

Before presenting our algorithms, we need to show the formal definitions for the problem of MEB with outliers first, since it will be used for computing the $(1-\delta)$-covering approximation. 

\begin{definition}[MEB with Outliers]
\label{def-outlier}
Given a set $P$ of $n$ points in $\mathbb{R}^d$ and a small parameter $\gamma\in [0,1)$, the MEB with outliers problem is to find the smallest ball that covers $(1-\gamma)n$ points. Namely, the task is to find a subset of $P$ with size $(1-\gamma)n$ such that the resulting MEB is the smallest among all possible choices of the subset. The obtained ball is denoted by $\mathbf{MEB}(P, \gamma)$.
\end{definition}

For convenience, we  use $P_{\textnormal{opt}}$ to denote the optimal subset of $P$ with respect to $\mathbf{MEB}(P, \gamma)$. Namely, 
$P_{\textnormal{opt}}=\arg_{Q}\min\Big\{\mathbf{Rad}(Q)\mid Q\subset P, \left|Q\right|= (1-\gamma)n\Big\}$. 
From Definition~\ref{def-outlier}, we can see that the main challenge is to determine the subset of $P$. Similar to Definition~\ref{def-app}, we also define the radius approximation and covering approximation for MEB with outliers. 

\begin{definition}[Radius Approximation and Covering Approximation]
\label{def-app-outlier}
Let $0<\epsilon, \delta<1$. A ball $\mathbb{B}(c, r)$ is called a $(1+\epsilon)$-radius approximation of $\mathbf{MEB}(P,\gamma)$, if the ball covers $(1-\gamma)n$ points of $P$ and has radius $r\leq (1+\epsilon) \mathbf{Rad}(P_{\textnormal{opt}})$.  On the other hand, the ball is called a $(1-\delta)$-covering approximation of $\mathbf{MEB}(P,\gamma)$, if it covers at least $(1-\delta-\gamma)n$ points in $P$ and has radius $r\leq \mathbf{Rad}(P_{\textnormal{opt}})$.

A bi-criteria $(1+\epsilon, 1-\delta)$-approximation  is a ball that covers at least $\big(1-\delta-\gamma\big)n$ points and has radius at most $(1+\epsilon)\mathbf{Rad}(P_{\textnormal{opt}})$.

\end{definition}

%
%

\textbf{Roadmap.} First, we introduce two random sampling techniques in Section~\ref{sec-twolemma}, which are the keys for designing the sublinear bi-criteria approximation algorithm for MEB with outliers in Section~\ref{sec-outlier-general}. Based on the bi-criteria approximation of Section~\ref{sec-outlier-general}, we can solve the general MEB problem in Section~\ref{sec-unknown2}.


\subsection{Two Key Lemmas for Handling Outliers}
\label{sec-twolemma}


To shed some light on our ideas, consider using the core-set construction method in Section~\ref{sec-newanalysis} to compute a bi-criteria $(1+\epsilon,1-\delta)$-approximation for an instance $(P, \gamma)$ of MEB with outliers. \textbf{Let $o_i$ be the obtained ball center in the current iteration, and $Q$ be the set of $(\delta+\gamma) n$ farthest points to $o_i$ from $P$.} A key step for updating $o_i$ is finding a point in the set $P_{\text{opt}}\cap Q$ (the formal analysis is given in Section~\ref{sec-outlier-general}). Actually, this can be done by performing a random sampling from $Q$. However, it requires to compute the set $Q$ in advance, which takes an $\Omega(nd)$ time complexity. To keep the running time to be sublinear, we need to find a point from $P_{\text{opt}}\cap Q$ by a more sophisticated way. 
Since $P_{\text{opt}}$ is mixed with outliers in the set $Q$, simple uniform sampling cannot realize our goal.
To solve this issue, we propose  a ``two level'' sampling procedure
 which is called ``\textbf{Uniform-Adaptive Sampling}''. Roughly speaking, we take a random sample $A$ of size $n'$ first ({\em i.e.,} the uniform sampling step), and then randomly select a point from $Q'$, the set of the farthest $\frac{3}{2}(\delta+\gamma) n'$  points from $A$ to $o_i$ ({\em i.e.,} the adaptive sampling step). According to Lemma~\ref{lem-outlier-general1}, with probability at least $(1-\eta_1) \frac{\delta}{3(\delta+\gamma)}$, the selected point belongs to $P_{\text{opt}}\cap Q$; more importantly, the sample size $n'$ is independent of $n$ and $d$.  The key to prove Lemma~\ref{lem-outlier-general1} is to show that the size of the intersection $Q'\cap\big(P_{\text{opt}}\cap Q\big)$ is large enough. By setting an appropriate value for $n'$, we can prove a lower bound of $|Q'\cap\big(P_{\text{opt}}\cap Q\big)|$. 

\begin{lemma}[Uniform-Adaptive Sampling]
\label{lem-outlier-general1}
Let $\eta_1\in(0,1)$.  If we sample $n'=O(\frac{1}{\delta}\log\frac{1}{\eta_1})$ points independently and uniformly at random from $P$ and let $Q'$ be the set of farthest $\frac{3}{2}(\delta+\gamma)n'$ points to $o_i$ from the sample, then, with probability at least $1-\eta_1$, the following holds 
\begin{eqnarray}
\frac{\Big|Q'\cap\big(P_{\text{opt}}\cap Q\big)\Big|}{|Q'|}\geq \frac{\delta}{3(\delta+\gamma)}.
\end{eqnarray}
\end{lemma}

\begin{proof}
Let $A$ denote  the set of sampled $n'$ points from $P$. First, we know $|Q|=(\delta+\gamma) n$ and $|P_{\text{opt}}\cap Q|\geq \delta n$ (since there are at most $\gamma n$ outliers in $Q$). For ease of presentation, let $\lambda=\frac{|P_{\text{opt}}\cap Q|}{n}\geq \delta$. Let $\{x_i\mid 1\leq i\leq n'\}$ be $n'$ independent random variables with $x_i=1$ if the $i$-th sampled point of $A$ belongs to $P_{\text{opt}}\cap Q$, and $x_i=0$ otherwise. Thus, $E[x_i]= \lambda$ for each $i$. Let $\sigma$ be a small parameter in $(0,1)$. By using the Chernoff bound, we have  $\textbf{Pr}\Big(\sum^{n'}_{i=1}x_i\notin (1\pm\sigma)\lambda n'\Big)\leq e^{-O(\sigma^2 \lambda n')}$. That is,
\begin{eqnarray}
\textbf{Pr}\Big(|A\cap\big(P_{\text{opt}}\cap Q\big)|\in (1\pm\sigma)\lambda n'\Big)\geq 1-e^{-O(\sigma^2 \lambda n')}. \label{for-outlier-e3-1}
\end{eqnarray}
Similarly, we have
\begin{eqnarray}
\textbf{Pr}\Big(|A\cap  Q|\in (1\pm\sigma)(\delta+\gamma) n'\Big)\geq 1-e^{-O(\sigma^2 (\delta+\gamma) n')}. \label{for-outlier-e4-1}
\end{eqnarray}
Note that $n'=O(\frac{1}{\delta}\log\frac{1}{\eta_1})$. By setting $\sigma<1/2$ for (\ref{for-outlier-e3-1}) and (\ref{for-outlier-e4-1}), we have
\begin{eqnarray}
\Big|A\cap\big(P_{\text{opt}}\cap Q\big)\Big|> \frac{1}{2}\delta n' \text{\hspace{0.2in} and \hspace{0.2in}} \Big|A\cap  Q\Big|< \frac{3}{2}(\delta+\gamma) n' \label{for-outlier-general1-1}
\end{eqnarray}
with probability $1-\eta_1$. 
Note that $Q$ contains all the farthest $(\delta+\gamma) n$ points to $o_i$. Denote by $l_i$ the $\big((\delta+\gamma) n+1\big)$-th largest distance from $P$ to $o_i$. Then we have  
\begin{eqnarray}
A\cap Q=\{p\in A\mid ||p-o_i||>l_i\}.\label{for-aq}
\end{eqnarray}
Also, since $Q'$ is the set of the farthest $\frac{3}{2}(\delta+\gamma) n'$ points to $o_i$ from $A$, there exists some $l'_i>0$ such that
\begin{eqnarray}
 Q'=\{p\in A\mid ||p-o_i||>l'_i\}.\label{for-qprime}
\end{eqnarray}
(\ref{for-aq}) and (\ref{for-qprime}) together imply that either $(A\cap Q)\subseteq Q'$ or $Q'\subseteq (A\cap Q)$. Since $\big|A\cap  Q\big|< \frac{3}{2}(\delta+\gamma) n'$ and $|Q'|=\frac{3}{2}(\delta+\gamma) n'$, we know $\Big(A\cap  Q\Big)\subseteq Q'$. Therefore, 
\begin{eqnarray}
\Big(A\cap\big(P_{\text{opt}}\cap Q\big)\Big)=\Big(P_{\text{opt}}\cap \big(A\cap  Q\big)\Big)\subseteq Q'.\label{for-v9-1} 
\end{eqnarray}
Also, it is obvious that 
\begin{eqnarray}
\Big(A\cap\big(P_{\text{opt}}\cap Q\big)\Big)\subseteq  \big(P_{\text{opt}}\cap Q\big).\label{for-v9-4} 
\end{eqnarray}
The above (\ref{for-v9-1}) and (\ref{for-v9-4}) together imply  
\begin{eqnarray}
\Big(A\cap\big(P_{\text{opt}}\cap Q\big)\Big)\subseteq \Big(Q'\cap\big(P_{\text{opt}}\cap Q\big)\Big).\label{for-v9-2} 
\end{eqnarray}
Moreover, since $Q'\subseteq A$, we have 
\begin{eqnarray}
\Big(Q'\cap\big(P_{\text{opt}}\cap Q\big)\Big)\subseteq \Big(A\cap\big(P_{\text{opt}}\cap Q\big)\Big).\label{for-v9-3} 
\end{eqnarray}
Consequently, (\ref{for-v9-2}) and (\ref{for-v9-3}) together imply $Q'\cap\big(P_{\text{opt}}\cap Q\big)=A\cap\big(P_{\text{opt}}\cap Q\big)$ and hence 
\begin{eqnarray}
\frac{\Big|Q'\cap\big(P_{\text{opt}}\cap Q\big)\Big|}{|Q'|}&=&\frac{\Big|A\cap\big(P_{\text{opt}}\cap Q\big)\Big|}{|Q'|}\geq \frac{\delta}{3(\delta+\gamma)},
\end{eqnarray}
where the  inequality comes from the first inequality of (\ref{for-outlier-general1-1}) and the fact $|Q'|=\frac{3}{2}(\delta+\gamma) n'$.
\qed\end{proof}

The random sampling method is not always guaranteed to succeed. To boost the overall success probability, we have to repeatedly run the algorithm multiple times and each time the algorithm will generate a candidate solution ({\em i.e.,} the ball center). Consequently we have to select the best one as our final solution. With a slight abuse of notation, we still use $o_i$ to denote a candidate ball center; since our goal is to  achieve a  $(1+\epsilon, 1-\delta)$-approximation, we need to compute  the $\big((\delta+\gamma) n+1\big)$-th largest distance from $P$ to $o_i$, which is denoted as $l_i$. A straightforward way is to compute the value ``$l_i$'' in linear time for each candidate and return the one having the smallest $l_i$. In this section, we propose the ``\textbf{Sandwich Lemma}'' to estimate $l_i$ in sublinear time. Let $B$ be the set of $n''$ sampled points from $P$ in Lemma~\ref{lem-outlier-general2}, and $\tilde{l}_i$ be the $\big((1+\delta/\gamma)^2\gamma n''+1\big)$-th largest distance from  $B$ to $o_i$. If we can prove the inequalities (\ref{for-outlier-general2-1}) and (\ref{for-outlier-general2-2}) of Lemma~\ref{lem-outlier-general2}, then they can imply that  $\tilde{l}_i$ is a qualified estimation of $l_i$: if $\mathbb{B}(o_i, l_i)$ is a $(1+\epsilon, 1-\delta)$-approximation, the ball  $\mathbb{B}(o_i, \tilde{l}_i)$  should be a $(1+\epsilon, 1-O(\delta))$-approximation. 
The key idea is to prove that the ball $\mathbb{B}(o_i, \tilde{l}_i)$ is ``sandwiched'' by two balls $\mathbb{B}(o_i, \tilde{l}'_i)$ and $\mathbb{B}(o_i, l_i)$, where $\tilde{l}'_i$ is a carefully designed value satisfying 
\begin{eqnarray}
\text{(\rmnum{1}) $\tilde{l}'_i\leq \tilde{l}_i\leq l_i$ and (\rmnum{2}) $\Big|P\setminus \mathbb{B}(o_i, \tilde{l}'_i)\Big|\leq (\gamma+O(\delta)) n$.}
\end{eqnarray}
See Figure~\ref{fig-sub-outlier} for an illustration.  These two conditions of $\tilde{l}'_i$ can imply the inequalities (\ref{for-outlier-general2-1}) and (\ref{for-outlier-general2-2}) of Lemma~\ref{lem-outlier-general2}. 
Similar to Lemma~\ref{lem-outlier-general1}, the sample size $n''$ is also independent of $n$ and $d$.

\begin{figure}
\begin{center}
    \includegraphics[width=0.27\textwidth]{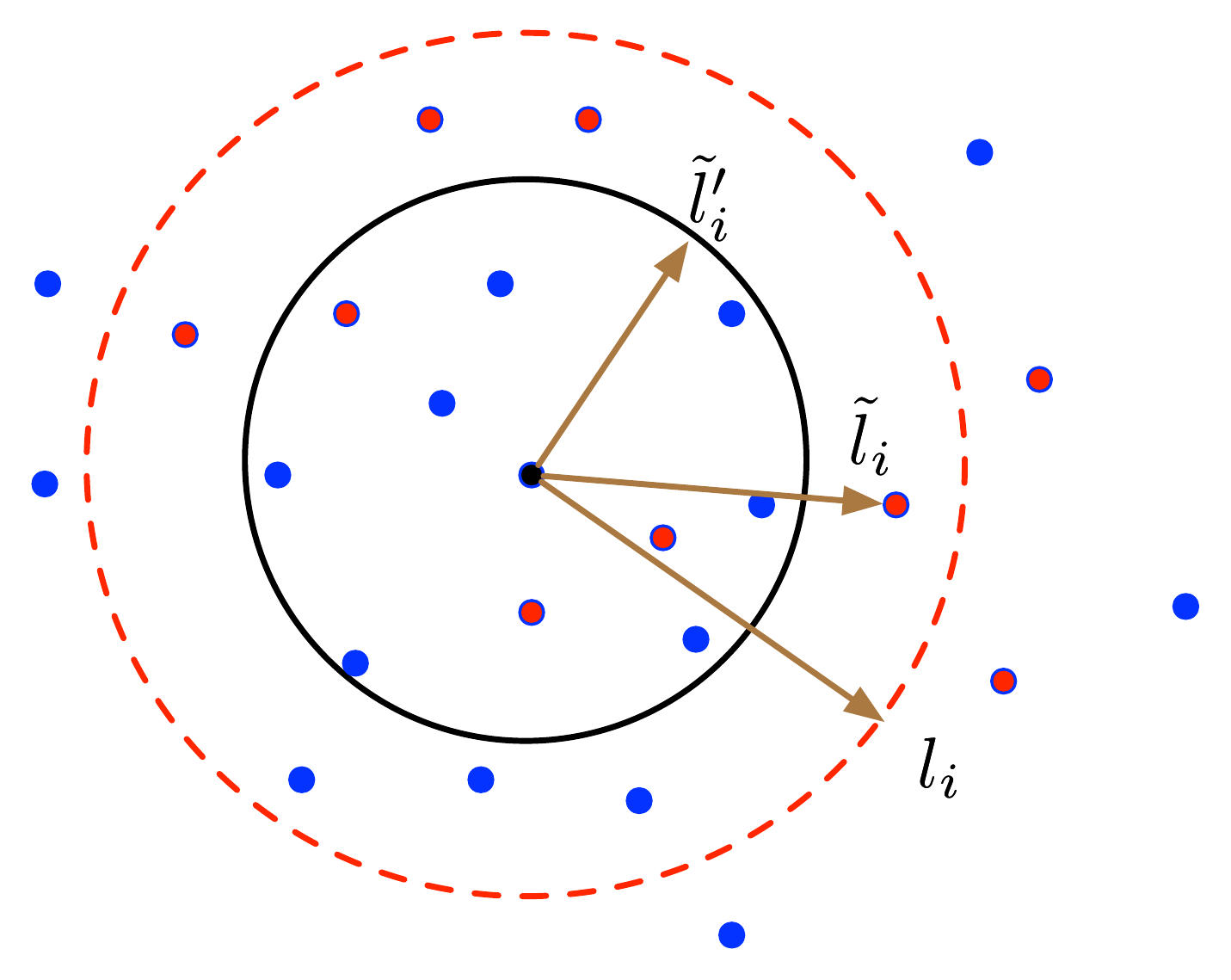}  
    \end{center}
  \caption{The red points are the sampled $n''$ points in Lemma~\ref{lem-outlier-general2}, and the $\big((1+\delta/\gamma)^2\gamma n''+1\big)$-th farthest point is in the ring bounded by the spheres $\mathbb{B}(o_i, \tilde{l}'_i)$ and $\mathbb{B}(o_i, l_i)$.}     
   \label{fig-sub-outlier}
\end{figure}

\begin{lemma} [Sandwich Lemma]
\label{lem-outlier-general2}
Let $\eta_2\in(0,1)$ and assume $\delta<\gamma/3$. If we sample $n''=O\big(\frac{\gamma}{\delta^2}\log\frac{1}{\eta_2}\big)$ points independently and uniformly at random from $P$ and let $\tilde{l}_i$ be the $\big((1+\delta/\gamma)^2\gamma n''+1\big)$-th largest distance from the sample to $o_i$, then, with probability $1-\eta_2$, the following holds
\begin{eqnarray}
\tilde{l}_i&\leq& l_i\label{for-outlier-general2-1};\\
\Big|P\setminus \mathbb{B}(o_i, \tilde{l}_i)\Big|&\leq& (\gamma+5\delta) n.\label{for-outlier-general2-2}
\end{eqnarray}
\end{lemma}

\begin{proof}
Let $B$ denote  the set of sampled $n''$ points from $P$. For simplicity, let $t=(\delta+\gamma) n$. Assume $\tilde{l}'_i>0$ is the value such that $\Big|P\setminus \mathbb{B}(o_i, \tilde{l}'_i)\Big|=\frac{(\gamma+\delta)^2}{\gamma-\delta} n$. Recall that $l_i$ is the $\big(t+1\big)$-th largest distance from $P$ to $o_i$. Since $(\delta+\gamma) n<\frac{(\gamma+\delta)^2}{\gamma-\delta} n$, it is easy to know $\tilde{l}'_i\leq l_i$. Below, we aim to prove that the $\big((1+\delta/\gamma)^2\gamma n''+1\big)$-th farthest point from $B$ is in the ring bounded by the spheres $\mathbb{B}(o_i, \tilde{l}'_i)$ and $\mathbb{B}(o_i, l_i)$ (see Figure~\ref{fig-sub-outlier}).

Note the size $|B|=n''=O\big(\frac{\gamma}{\delta^2}\log\frac{1}{\eta_2}\big)$. Again, using the Chernoff bound  (let $\sigma=\delta/2$) and the same idea for proving (\ref{for-outlier-general1-1}),   we have 
\begin{eqnarray}
\Big|B\setminus \mathbb{B}(o_i, \tilde{l}'_i)\Big|&\geq& (1-\frac{\delta}{2\gamma})\frac{(\gamma+\delta)^2}{\gamma-\delta} n''>(1-\frac{\delta}{\gamma})\frac{(\gamma+\delta)^2}{\gamma-\delta}n''= (1+\delta/\gamma)^2 \gamma n'' ;\label{for-outlier-general2-3}\\
\Big|B\cap Q\big|&\leq& (1+\frac{\delta}{2\gamma})\frac{t}{n}n''< (1+\delta/\gamma)\frac{t}{n}n''=(1+\delta/\gamma)^2 \gamma n'', \label{for-outlier-general2-4}
\end{eqnarray}
with probability $1-\eta_2$. Suppose that (\ref{for-outlier-general2-3}) and (\ref{for-outlier-general2-4}) both hold. Recall that $\tilde{l}_i$ is the $\big((1+\delta/\gamma)^2\gamma n''+1\big)$-th largest distance from the sampled points $B$ to $o_i$, so $\Big|B\setminus \mathbb{B}(o_i, \tilde{l}_i)\Big|= (1+\delta/\gamma)^2 \gamma n'' $. Together with (\ref{for-outlier-general2-3}), we have $\Big|B\setminus \mathbb{B}(o_i, \tilde{l}_i)\Big|\leq \Big|B\setminus \mathbb{B}(o_i, \tilde{l}'_i)\Big|$, {\em i.e.,}
\begin{eqnarray}
 \tilde{l}_i\geq \tilde{l}'_i.\label{for-v4-b1}
 \end{eqnarray} 

The inequality (\ref{for-outlier-general2-4}) implies that the $\big((1+\delta/\gamma)^2\gamma n''+1\big)$-th farthest point (say $q_x$) from $B$ to $o_i$ is not in $Q$. Then, we claim that $ \mathbb{B}(o_i, \tilde{l}_i)\cap Q=\emptyset$. Otherwise, let $q_y\in  \mathbb{B}(o_i, \tilde{l}_i)\cap Q$. 
Then we have 
\begin{eqnarray}
||q_y-o_i||\leq\tilde{l}_i=||q_x-o_i||.\label{for-v9-5}
\end{eqnarray} 
Note that $Q$ is the set of farthest $t$ points to $o_i$ of $P$. So $q_x\notin Q$ implies 
\begin{eqnarray}
||q_x-o_i||<\min_{q\in Q}||q-o_i|| \leq ||q_y-o_i||
\end{eqnarray}
which is in contradiction to (\ref{for-v9-5}). Therefore, $ \mathbb{B}(o_i, \tilde{l}_i)\cap Q=\emptyset$. 
Further, since  $\mathbb{B}(o_i, l_i)$ excludes exactly the farthest $t$ points ({\em i.e.}, $Q$), ``$ \mathbb{B}(o_i, \tilde{l}_i)\cap Q=\emptyset$'' implies 
\begin{eqnarray}
 \tilde{l}_i\leq l_i.\label{for-v4-b2}
  \end{eqnarray}

Overall, we have $\tilde{l}_i\in[\tilde{l}'_i,l_i]$ from (\ref{for-v4-b1}) and (\ref{for-v4-b2}), {\em i.e.,} the $\big((1+\delta/\gamma)^2\gamma n''+1\big)$-th farthest point from $B$ locates in the ring bounded by the spheres $\mathbb{B}(o_i, \tilde{l}'_i)$ and $\mathbb{B}(o_i, l_i)$ as shown in Figure~\ref{fig-sub-outlier}. Also, $\tilde{l}_i\geq \tilde{l}'_i$ implies 
\begin{eqnarray}
\Big|P\setminus \mathbb{B}(o_i, \tilde{l}_i)\Big|&\leq& \Big|P\setminus \mathbb{B}(o_i, \tilde{l}'_i)\Big|=\frac{(\gamma+\delta)^2}{\gamma-\delta} n<(\gamma+5\delta) n,
\end{eqnarray}
where the last equality comes from the assumption $\delta<\gamma/3$. So (\ref{for-outlier-general2-1}) and (\ref{for-outlier-general2-2}) are true in Lemma~\ref{lem-outlier-general2}. 
\qed
\end{proof}

\begin{remark}
Actually our proposed Uniform-Adaptive Sampling method and Sandwich lemma are quite generic, and we will show that they can be generalized to solve a broader range of enclosing with outliers problems in Section~\ref{sec-ext}.
\end{remark}

\subsection{Sublinear Time Algorithm for Bi-criteria Approximation}
\label{sec-outlier-general}
%
%
%
%
%

In this section, we propose a sublinear time algorithm for computing a bi-criteria $(1+\epsilon, 1-\delta)$-approximation for the input instance $(P, \gamma)$; that is, the returned ball covers at least $\big(1-\delta-\gamma\big)n$ points and has radius at most $(1+\epsilon)\mathbf{Rad}(P_{\textnormal{opt}})$. 


Recall the remark following Theorem~\ref{the-newbc}. 
%
As long as the selected point has a distance to the center of $\textbf{MEB}(T)$ larger than $(1+\epsilon)$ times the optimal radius, the expected 
improvement will always be guaranteed. Following this observation, we investigate the following approach. 
%
%
%
%
%
%
%
%
%
%
%
Suppose we run the core-set construction procedure decribed in Theorem~\ref{the-newbc} (we should replace $P$ by $P_{\text{opt}}$ in our following analysis). 
In the $i$-th step, we add an arbitrary point from $P_{\textnormal{opt}}\setminus \mathbb{B}(o_i, (1+\epsilon)\textbf{Rad}(P_{\text{opt}}))$ to $T$. We know that a $(1+\epsilon)$-approximation is obtained after at most $ \frac{2}{(1-s)\epsilon} $ steps, that is, $P_{\text{opt}}\subset \mathbb{B}\big(o_i, (1+\epsilon)\textbf{Rad}(P_{\text{opt}})\big)$ for some $i\leq   \frac{2}{(1-s)\epsilon}  $. 


However, we need to solve two key issues for realizing 
the above approach: \textbf{(\rmnum{1})} how to determine the value of $\textbf{Rad}(P_{\text{opt}})$ and \textbf{(\rmnum{2})} how to correctly select a point from $P_{\textnormal{opt}}\setminus\mathbb{B}(o_i, (1+\epsilon)\textbf{Rad}(P_{\text{opt}}))$. Actually, we can implicitly avoid the first issue via replacing $(1+\epsilon)\textbf{Rad}(P_{\text{opt}})$ by the $t$-th largest distance from the points of $P$ to $o_i$, where we set $t=(\delta+\gamma) n$ for guaranteeing a $(1+\epsilon, 1-\delta)$-approximation. For the second issue, we randomly select one point from the farthest $t$ points of $P$ to $o_i$, and show that it belongs to $P_{\textnormal{opt}}\setminus\mathbb{B}(o_i, (1+\epsilon)\textbf{Rad}(P_{\text{opt}}))$ with a certain probability.



Based on the above idea, we present a sublinear time  $(1+\epsilon, 1-\delta)$-approximation algorithm in this section. To better  understand the algorithm, we show a linear time algorithm first (Algorithm~\ref{alg-outlier} in Sections~\ref{sec-quality}). Note that B\u{a}doiu {\em et al.}~\cite{BHI} also presented a $(1+\epsilon, 1-\delta)$-approximation algorithm but with a higher complexity, and please see our detailed analysis on the running time at the end of Sections~\ref{sec-quality}. More importantly, we can improve the running time of Algorithm~\ref{alg-outlier} to be sublinear. 
%
%
%
For this purpose, we need to avoid computing the farthest $t$ points to $o_i$, since this operation will take 
linear time. 
Also, Algorithm~\ref{alg-outlier} generates a set of candidates for the solution and we need to select the best one. This process also costs linear time.  By using the techniques proposed in Section~\ref{sec-twolemma}, we can solve these issues and develop 
%
a sublinear time algorithm that has the sample complexity independent of $n$ and $d$, in Section~\ref{sec-oulier-improve}.



  \subsubsection{A Linear Time Algorithm}
\label{sec-quality}

In this section, we present our linear time $(1+\epsilon,1-\delta)$-approximation algorithm for MEB with outliers.


\renewcommand{\algorithmicrequire}{\textbf{Input:}}
\renewcommand{\algorithmicensure}{\textbf{Output:}}
\begin{algorithm}
   \caption{$(1+\epsilon,1-\delta)$-approximation Algorithm for MEB with Outliers}
   \label{alg-outlier}
\begin{algorithmic}[1]
\REQUIRE A point set $P$ with $n$ points in $\mathbb{R}^{d}$, the fraction of outliers $\gamma\in(0,1)$, and the parameters $0<\epsilon,\delta<1$, $z\in\mathbb{Z}^{+}$.
\STATE Let $t=(\delta+\gamma) n$.
\STATE Initially, randomly select a point $p\in P$ and let $T=\{p\}$. 
\STATE $i=1$; repeat the following steps until $i>z$:
\begin{enumerate}[(1)]
\item  Denote by $c_i$   the exact center   of $\textbf{MEB}(T)$.
Compute the approximate center $o_i$ with a distance to $c_i$ of less than $\xi \textbf{Rad}(T)=s\frac{\epsilon}{1+\epsilon}\textbf{Rad}(T)$  as described in Theorem~\ref{the-newbc}, where $s$ is set to be $\frac{\epsilon}{2+\epsilon}$.

\item Let $Q$ be the set of farthest $t$ points from $P$ to $o_i$; denote by $l_i$ the $(t+1)$-th largest distance from $P$ to $o_i$.
\item Randomly select a point $q\in Q$, and add it to $T$. 
\item $i=i+1$.
\end{enumerate}
\STATE Output the ball $\mathbb{B}(o_{\hat{i}}, l_{\hat{i}})$ where $\hat{i}=\arg_{i}\min\{l_i\mid 1\leq i\leq z\}$.
\end{algorithmic}
\end{algorithm}

\begin{theorem}
\label{the-outlier}
If the input parameter $z= \frac{2}{(1-s)\epsilon}$ (we assume it is an integer for convenience), then with probability $(1-\gamma)(\frac{\delta}{\gamma+\delta})^{z}$, Algorithm~\ref{alg-outlier} outputs a $(1+\epsilon,1-\delta)$-approximation for the MEB with outliers problem.
\end{theorem}

Before proving Theorem~\ref{the-outlier}, we present the following two lemmas first.

\begin{lemma}
\label{lem-outlier-1}
With probability $(1-\gamma)(\frac{\delta}{\gamma+\delta})^{z}$, after running $z$ rounds in Step 3 of Algorithm~\ref{alg-outlier}, the obtained set $T\subset P_{\text{opt}}$.
\end{lemma}
\begin{proof}
Initially, because $|P_{\text{opt}}|/|P|=1-\gamma$, the first selected point in Step 2 belongs to $P_{\text{opt}}$ with probability $1-\gamma$. In each of the $z$ rounds in Step 3, the selected point belongs to $P_{\text{opt}}$ with probability $\frac{\delta}{\gamma+\delta}$, since 
\begin{eqnarray}
\frac{|P_{\text{opt}}\cap Q|}{|Q|}&=& 1-\frac{|Q\setminus P_{\text{opt}}|}{|Q|}\geq 1-\frac{|P\setminus P_{\text{opt}}|}{|Q|}=1-\frac{\gamma n}{(\delta+\gamma) n}=\frac{\delta}{\delta+\gamma}. \label{for-lem-outlier-1}
\end{eqnarray}
Therefore, with probability $(1-\gamma)(\frac{\delta}{\gamma+\delta})^{z}$ the whole set $T\subset P_{\text{opt}}$.
\qed\end{proof}

\begin{lemma}
In the $i$-th round of Step 3 for $1\leq i\leq z$, at least one of the following two events happens: (1) $o_i$ is the ball center of a $(1+\epsilon,1-\delta)$-approximation; (2) $r_{i+1}>(1+\epsilon)\textbf{Rad}(P_{\text{opt}})-||c_i-c_{i+1}||-\xi r_i$, where $r_i$ is the exact radius of $\textbf{MEB}(T)$ is the $i$-th round. 
\label{lem-outlier-2}
\end{lemma}
\begin{proof}
If $l_i\le(1+\epsilon)\textbf{Rad}(P_{\text{opt}})$, then we are done. That is, the ball $\mathbb{B}(o_i, l_i)$ covers $(1-\delta-\gamma)n$ points with radius $l_i\leq (1+\epsilon)\textbf{Rad}(P_{\text{opt}})$ (the first event happens). Otherwise, $l_i> (1+\epsilon)\textbf{Rad}(P_{\text{opt}})$ and we consider the second event.
Let $q$ be the point added to $T$ in the $i$-th round. Using the triangle inequality, we have
\begin{eqnarray}
||o_i-q||\leq ||o_i-c_i||+||c_i-c_{i+1}||+|c_{i+1}-q||\leq \xi r_i+||c_i-c_{i+1}||+r_{i+1}. \label{for-lemma5-1}
\end{eqnarray}
 Since $l_i> (1+\epsilon)\textbf{Rad}(P_{\text{opt}})$ and $q$ lies outside  of $\mathbb{B}(o_i, l_i)$, {\em i.e,} $||o_i-q||\geq l_i> (1+\epsilon)\textbf{Rad}(P_{\text{opt}})$, (\ref{for-lemma5-1}) implies that the second event happens and the proof is completed.
\qed\end{proof}

\begin{proof} \textbf{(of Theorem~\ref{the-outlier})}
Suppose that the first event of Lemma~\ref{lem-outlier-2} never happens. As a consequence, we obtain a series of inequalities for each pair of radii $ r_{i+1}$ and $r_{i}$, {\em i.e.,} $r_{i+1}>(1+\epsilon)\textbf{Rad}(P_{\text{opt}})-||c_i-c_{i+1}||-\xi r_i$. Assume that $T\subset P_{\text{opt}}$ in Lemma~\ref{lem-outlier-1}, {\em i.e.,} each time the algorithm correctly adds a point from $P_{\text{opt}}$ to $T$. 
Using the almost identical idea for proving Theorem~\ref{the-newbc} in Section~\ref{sec-newanalysis}, 
we know that a $(1+\epsilon)$-approximate MEB of $P_{\text{opt}}$ is obtained after at most $z$ rounds.  The success probability directly comes from Lemma~\ref{lem-outlier-1}. Overall, we obtain Theorem~\ref{the-outlier}. 
\qed\end{proof}

Theorem~\ref{the-outlier} directly implies the following corollary.

\begin{corollary}
\label{cor-outlier}
If one repeatedly runs Algorithm~\ref{alg-outlier} $O(\frac{1}{1-\gamma}(1+\frac{\gamma}{\delta})^z)$ times, with constant probability, the algorithm outputs a $(1+\epsilon,1-\delta)$-approximation for the problem of MEB with outliers.
\end{corollary}

\textbf{Running time.} 
In Theorem~\ref{the-outlier}, we set $z=\frac{2}{(1-s)\epsilon}$ and $s\in(0,1)$. To keep $z$ small, according to Theorem~\ref{the-newbc}, we set $s=\frac{\epsilon}{2+\epsilon}$ so that $z=\frac{2}{\epsilon}+1$ (only larger than the lower bound $\frac{2}{\epsilon}$ by $1$). 
For each round of Step 3, we need to compute an approximate center $o_i$ that has a distance to the exact one less than $\xi r_i=s\frac{\epsilon}{1+\epsilon}r_i=O(\epsilon^2)r_i$. Using the algorithm proposed in~\cite{badoiu2003smaller}, this can be done in $O(\frac{1}{\xi^2}|T|d)=O(\frac{1}{\epsilon^5}d)$ time. Also, the set $Q$ can be obtained in linear time by the algorithm in~\cite{blum1973time}. 
In total, the time complexity for obtaining a $(1+\epsilon,1-\delta)$-approximation in Corollary~\ref{cor-outlier} is 
\begin{eqnarray}
O\big(\frac{C}{\epsilon}(n+\frac{1}{\epsilon^5})d\big) \label{for-time-bi},
\end{eqnarray}
where $C=O(\frac{1}{1-\gamma}(1+\frac{\gamma}{\delta})^{\frac{2}{\epsilon}+1})$.  As mentioned before, B\u{a}doiu {\em et al.}~\cite{BHI} also proposed a linear time bi-criteria approximation. However, the hidden constant of their running time is exponential in $\Theta(\frac{1}{\epsilon\delta})$  that is much larger than $\frac{2}{\epsilon}+1$.

%

\subsubsection{Improvement on Running Time}
\label{sec-oulier-improve}

In this section, we show that the running time of Algorithm~\ref{alg-outlier} can be further improved to be independent of the number of points $n$. First, we observe that it is not necessary to compute the set $Q$ of the farthest $t$ points in Step 3(2) of the algorithm. Actually, as long as the selected point $q$ is part of $ P_{\text{opt}}\cap Q$ in Step 3(3), a $(1+\epsilon, 1-\delta)$-approximation is still guaranteed. The Uniform-Adaptive Sampling procedure proposed in Section~\ref{sec-twolemma} can help us to obtain a point $q\in P_{\text{opt}}\cap Q$ without computing the set $Q$. Moreover, in Lemma~\ref{lem-outlier-general2}, we show that the radius of each candidate solution can be estimated via random sampling. Overall, we achieve a sublinear time algorithm (Algorithm~\ref{alg-outlier2}). Following the analysis in Section~\ref{sec-quality}, we set $s=\frac{\epsilon}{2+\epsilon}$ so that $z=\frac{2}{(1-s)\epsilon}=\frac{2}{\epsilon}+1$. We present the results in Theorem~\ref{the-outlier2} and Corollary~\ref{cor-outlier2}. Comparing with Theorem~\ref{the-outlier}, we have an extra $(1-\eta_1)(1-\eta_2)$ in the success probability in Theorem~\ref{the-outlier2}, due to the probabilities from Lemmas~\ref{lem-outlier-general1} and \ref{lem-outlier-general2}. Another minor issue is that the covering approximation error is increased from $\delta$  to $5\delta$ when applying Lemma~\ref{lem-outlier-general2}. Actually this issue can be easily solved by replacing $\delta$ by $\delta/5$ in the parameters $n'$, $t'$, $n''$, and $t''$, and the asymptotic complexity does not change.

\renewcommand{\algorithmicrequire}{\textbf{Input:}}
\renewcommand{\algorithmicensure}{\textbf{Output:}}
\begin{algorithm}
\vspace{-0.05in}
   \caption{Sublinear Time $(1+\epsilon,1-\delta)$-approximation Algorithm for MEB with Outliers}
   \label{alg-outlier2}
\begin{algorithmic}[1]
\REQUIRE A point set $P$ with $n$ points in $\mathbb{R}^{d}$, the fraction of outliers $\gamma\in(0,1)$, and the parameters $\epsilon,\eta_1, \eta_2\in (0,1)$, $\delta\in (0,1/3\gamma)$, and $z\in\mathbb{Z}^{+}$.
\STATE Let $n'=O(\frac{1}{\delta}\log\frac{1}{\eta_1})$, $n''=O\big(\frac{\gamma}{\delta^2}\log\frac{1}{\eta_2}\big)$, $t'=\frac{3}{2}(\delta/5+\gamma) n'$, and $t''=(1+\frac{\delta}{5\gamma})^2\gamma n''$.
\STATE Initially, randomly select a point $p\in P$ and let $T=\{p\}$. 
\STATE $i=1$; repeat the following steps until $j=z$:
\begin{enumerate}[(1)]
\item Compute the approximate MEB center $o_i$ of $T$.

\item Sample $n'$ points uniformly at random from $P$, and let $Q'$ be the set of farthest $t'$ points to $o_i$ from the sample.

\item Randomly select a point $q\in Q'$, and add it to $T$. 

\item Sample $n''$ points uniformly at random from $P$, and let $\tilde{l}_i$ be the $(t''+1)$-th largest distance from the sampled points to $o_i$.
\item $i=i+1$.
\end{enumerate}
\STATE Output the ball $\mathbb{B}(o_{\hat{i}}, \tilde{l}_{\hat{i}})$ where $\hat{i}=\arg_{i}\min\{\tilde{l}_i\mid 1\leq i\leq z\}$.
\end{algorithmic}
\end{algorithm}


\begin{theorem}
\label{the-outlier2}
If the input parameter $z= \frac{2}{\epsilon}+1$, then with probability $(1-\gamma)\big((1-\eta_1)(1-\eta_2)\frac{\delta/5}{3(\gamma+\delta/5)}\big)^{z}$, Algorithm~\ref{alg-outlier2} outputs a $(1+\epsilon,1-\delta)$-approximation for the problem of MEB with outliers.
\end{theorem}

To boost the success probability in Theorem~\ref{the-outlier2}, we need to repeatedly run Algorithm~\ref{alg-outlier2} and output the best candidate. However, we need to be careful on setting the parameters. The success probability in Theorem~\ref{the-outlier2} consists of two parts, $\mathcal{P}_1=(1-\gamma)\big((1-\eta_1)\frac{\delta/5}{3(\gamma+\delta/5)}\big)^{z}$ and $\mathcal{P}_2=(1-\eta_2)^{z}$, where $\mathcal{P}_1$ indicates the probability that $\{o_1, \cdots, o_z\}$ contains a qualified candidate, and  $\mathcal{P}_2$ indicates the success probability of Lemma~\ref{lem-outlier-general2} over all the $z$ rounds. Therefore, if we run Algorithm~\ref{alg-outlier2} $N=O(\frac{1}{\mathcal{P}_1})$ times, with constant probability (by taking the union bound), the set of all the generated candidates contains at least one that yields a $(1+\epsilon,1-\delta)$-approximation; moreover, to guarantee that we can correctly estimate the resulting radii of all the candidates via the Sandwich Lemma with constant probability, we need to set $\eta_2=O(\frac{1}{zN})$ (because there are $O(zN)$ candidates).

%
%
%

\begin{corollary}
\label{cor-outlier2}
If one repeatedly runs Algorithm~\ref{alg-outlier2} $N=O\Big(\frac{1}{1-\gamma}\big(\frac{1}{1-\eta_1}(3+\frac{3\gamma}{\delta/5})\big)^{z}\Big)$ times with setting $\eta_2=O(\frac{1}{zN})$, with constant probability, the algorithm outputs a $(1+\epsilon,1-\delta)$-approximation for the problem of MEB with outliers.
\end{corollary}

The calculation of running time is similar to (\ref{for-time-bi}) in Section~\ref{sec-quality}. We just replace $n$ by $\max\{n', n''\}=O\big(\frac{\gamma}{\delta^2}\log\frac{1}{\eta_2}\big)=O\big(\frac{\gamma}{\delta^2}\log(zN)\big)=\tilde{O}\big(\frac{\gamma}{\delta^2\epsilon}\big)$~\footnote{The asymptotic notation $\tilde{O}(f)=O\big(f\cdot \mathtt{polylog}(\frac{\gamma}{\eta_1\delta(1-\gamma)})\big)$.}, and change the value of $C$ to be $O\Big(\frac{1}{1-\gamma}\big(\frac{1}{1-\eta_1}(3+\frac{3\gamma}{\delta/5})\big)^{\frac{2}{\epsilon}+1}\Big)$. So the total running time is independent of $n$.


\subsection{General MEB Problem}
\label{sec-unknown2}
Now we consider solving the general MEB problem without the stability assumption in this Section. 
Let $0<\epsilon$, $\delta<1$ be two given parameters. 
First, we view the input $P$ as an instance $(P, \delta/2)$  of MEB with outliers ({\em i.e.,} $\gamma=\delta/2$). Then, we apply the algorithm of Section~\ref{sec-outlier-general} to obtain a bi-criteria $(1+\epsilon^2/2, 1-\delta/2)$-approximation solution $\mathbb{B}(c, r_c)$ (we replace the ``$\epsilon$'' by $\epsilon^2/2$ and replace the ``$\delta$'' by $\delta/2$). The obtained ball $\mathbb{B}(c, r_c)$ covers at least $(1- \delta/2- \delta/2)n=(1-\delta)n$ points of $P$, and the radius
\begin{eqnarray}
r_c\leq (1+\frac{1}{2}\epsilon^2)\cdot r_{-\delta/2}, \label{for-unknown2-1}
\end{eqnarray}
where $r_{-\delta/2}$ stands for the radius of the smallest ball that covers at least $(1-\delta/2)n$ points of $P$. 

Second, we assume that the input $P$ is an $(\alpha, \beta)$-stable instance with $\alpha=\epsilon^2$ and $\beta=\delta/2$; then run Algorithm~\ref{alg-meb3} to obtain a candidate ball center $\tilde{o}$ (of course, we can also use Algorithm~\ref{alg-meb1}, where the only difference is that the sample complexity will be higher). To compute the real radius $r_{\tilde{o}}$ yielded from $\tilde{o}$ (since $P$ may not be a real $(\alpha, \beta)$-stable instance), we just need to read the whole dataset $P$ in one pass. Finally, we determine the final output based on the ratio $r_{\tilde{o}}/r_c$.


 \begin{algorithm}
   \caption{\textbf{Hybrid Approximation for MEB}}
   \label{alg-relaxoutlier-2}
\begin{algorithmic}[1]
\REQUIRE An instance $P$ of MEB   problem in $\mathbb{R}^d$; two parameters $0<\epsilon, \delta<1$.
\STATE View the input as a $(P, \delta/2)$ instance of MEB with outliers; apply the method of Corollary~\ref{cor-outlier2} to obtain a bi-criteria $(1+\epsilon^2/2, 1-\delta/2)$-approximation solution $\mathbb{B}(c, r_c)$ on $(P, \delta/2)$.

\STATE Assume that the input $P$ is an $(\alpha, \beta)$-stable instance with $\alpha=\epsilon^2$ and $\beta=\delta/2$; then run Algorithm~\ref{alg-meb3} to obtain a candidate ball center $\tilde{o}$.

\STATE Read the whole input dataset $P$ in one-pass, and compute the radius $r_{\tilde{o}}=\max_{p\in P}||\tilde{o}-p||$.

\STATE If the ratio $\frac{r_{\tilde{o}}}{r_{c}}\leq\frac{1+\epsilon}{1-\epsilon^2/2}$, return the ball $\mathbb{B}(\tilde{o}, r_{\tilde{o}})$ and say ``it is a $(1+\epsilon)$-radius approximation''.

\STATE Else, return the ball $\mathbb{B}(c, r_c)$ and say ``it is a $(1-\delta)$-covering approximation''.  
 \end{algorithmic}
\end{algorithm}

\begin{theorem}
\label{the-relaxoutlier-2}
With constant success probability, Algorithm~\ref{alg-relaxoutlier-2} returns either a $(1+\epsilon)$-radius approximation or a $(1-\delta)$-covering approximation, and the running time is $O\Big(\big(n+h(\epsilon, \delta)\big)\cdot d\Big)$, where $h(\epsilon, \delta)=O\big(\frac{1}{1-\delta/2}\exp(O(1/\epsilon^2))\big)$. The algorithm only needs uniform sampling and a single pass over the input data, and the space complexity in memory is $O(h(\epsilon, \delta)\cdot d)$. Moreover, if the input data matrix (the $n\times d$ matrix representing the input $P$) has at most $M\ll nd$ non-zeros entries, the total running time will be $O\big(n+h(\epsilon, \delta)\cdot d+M\big)$.
\end{theorem}
\begin{remark}
\label{rem-the-relaxoutlier}
In the following proof, we will see that when the algorithm returns a $(1-\delta)$-covering approximation, the returned radius is  not only $
\leq \mathbf{Rad}(P)$, but also at most $\big(1-\Theta(\epsilon^2)\big)\cdot\mathbf{Rad}(P)$ (see (\ref{for-the-relaxoutlier-3-1}) and (\ref{for-the-relaxoutlier-4-1})).
\end{remark}

\begin{proof}
%
%

We study the time and space complexities first. The method of Corollary~\ref{cor-outlier2} only needs uniform samplings, and Step 2 of  Algorithm~\ref{alg-relaxoutlier-2} is a single pass over the input data.  
According to Corollary~\ref{cor-outlier2}, we know the space complexity is $O(h(\epsilon, \delta)\cdot d)$ with $h(\epsilon, \delta)=O\big(\frac{1}{1-\delta/2}\exp(O(1/\epsilon^2))\big)$. The total running time is $O\Big(\big(n+h(\epsilon, \delta )\big)\cdot d\Big)$. 
Furthermore, we consider the case that the input matrix is sparse. In Step~3, we need to compute the value  $r_{\tilde{o}}=\max_{p\in P}||\tilde{o}-p||$. For  each point $p\in P$, we know 
\begin{eqnarray}
||\tilde{o}-p||^2=||\tilde{o}||^2+||p||^2-2\langle \tilde{o}, p\rangle, \label{for-the-relaxoutlier-1}
\end{eqnarray}
where $\langle \tilde{o}, p\rangle$ stands for their inner product. The value of $||\tilde{o}||^2$ can be obtained in $O(d)$ time, and if 
the input data matrix  has  at most $M\ll nd$ non-zeros entries, the complexity for computing the values $\{||p||^2-2\langle \tilde{o}, p\rangle\mid p\in P\}$ 
is $O(n+M)$. Overall, the complexity of Algorithm~\ref{alg-relaxoutlier-2} is $O\big(n+h(\epsilon, \delta)\cdot d+M\big)$.

Now, we prove the solution quality. We let  $\alpha=\epsilon^2$ and $\beta=\delta/2$, and consider the following two cases. 

\textbf{Case 1:} the instance $P$ is $(\alpha, \beta)$-stable. Then we directly have
\begin{eqnarray}
r_{\tilde{o}}\leq  (1+\epsilon)\cdot\mathbf{Rad}(P). \label{for-the-relaxoutlier-2-1}
\end{eqnarray}
If  $\frac{r_{\tilde{o}}}{r_{c}}>\frac{1+\epsilon}{1-\epsilon^2/2}$, together with (\ref{for-the-relaxoutlier-2-1}), we have
\begin{eqnarray}
r_{c}<\Big(1-\epsilon^2/2\Big)\cdot\mathbf{Rad}(P). \label{for-the-relaxoutlier-3-1}
\end{eqnarray}
Then we can return the ball $\mathbb{B}(c, r_{c})$ and say ``it is a $(1-\delta)$-covering approximation''.  On the other hand, when $\frac{r_{\tilde{o}}}{r_{c}}\leq\frac{1+\epsilon}{1-\epsilon^2/2}$, from (\ref{for-the-relaxoutlier-2-1}) we can return the ball $\mathbb{B}(\tilde{o}, r_{\tilde{o}})$ and say ``it is a $(1+\epsilon)$-radius approximation''.

\textbf{Case 2:}  $P$ is not an $(\alpha, \beta)$-stable instance. Then, from the definition of stability we know the optimal radius of the instance $(P, \delta/2)$ is no larger than 
\begin{eqnarray}
(1-\epsilon^2)\cdot\mathbf{Rad}(P).\label{for-the-relaxoutlier-6-1}
\end{eqnarray}
So we have 
\begin{eqnarray}
r_{c}&<&(1+\frac{1}{2}\epsilon^2)(1- \epsilon^2)\cdot\mathbf{Rad}(P)<\Big(1-\epsilon^2/2\Big)\cdot\mathbf{Rad}(P). \label{for-the-relaxoutlier-4-1}
\end{eqnarray}
If $\frac{r_{\tilde{o}}}{r_{c}}\leq\frac{1+\epsilon}{1-\epsilon^2/2}$,  together with (\ref{for-the-relaxoutlier-4-1}), it implies 
\begin{eqnarray}
r_{\tilde{o}}<(1+\epsilon)\cdot\mathbf{Rad}(P). \label{for-the-relaxoutlier-5-1}
\end{eqnarray}
Then we can return the ball $\mathbb{B}(\tilde{o}, r_{\tilde{o}})$ and say ``it is a $(1+\epsilon)$-radius approximation''.  On the other hand,  when $\frac{r_{\tilde{o}}}{r_{c}}>\frac{1+\epsilon}{1-\epsilon^2/2}$,  from  (\ref{for-the-relaxoutlier-4-1}) we can return  the ball $\mathbb{B}(c, r_{c})$ and say ``it is a $(1-\delta)$-covering approximation''.

Since the success probability of the method of Section~\ref{sec-outlier-general} is constant, the overall success probability of Algorithm~\ref{alg-relaxoutlier} is constant as well. 
\qed\end{proof}

\textbf{More analysis on the result of Algorithm~\ref{alg-relaxoutlier-2}.} We further consider an ``inverse'' question: can we infer the stability degree of the given instance $P$ from the output of Algorithm~\ref{alg-relaxoutlier-2}?  In Step 2, we assume that $P$ is an $(\epsilon^2, \delta/2)$-stable instance, but this may not be true in reality. 
 Recall the definition of ``$(\alpha, \beta)$-stable'' in Definition~\ref{def-stable}. We know that there always exists a value $\hat{\alpha}\in [0, 1)$ such that $P$ is  a $(\hat{\alpha}, \delta/2)$-stable.  We can use ``$\hat{\alpha}$'' to indicate the stability degree of $P$, for the fixed ``$\delta/2$''. The following theorem shows that we can infer the value of $\hat{\alpha}$ through  Algorithm~\ref{alg-relaxoutlier-2}.  
   
\begin{theorem}
\label{the-infer}
If  Algorithm~\ref{alg-relaxoutlier-2} returns a $(1+\epsilon)$-radius approximation, then $\hat{\alpha}<\epsilon$; otherwise, the algorithm returns a $(1-\delta)$-covering approximation and it implies $\hat{\alpha}>\frac{\epsilon^2}{2}$. 

In other words, the algorithm can distinguish the case $\hat{\alpha}\geq\epsilon$ (it must returns a $(1-\delta)$-covering approximation) and the case $\hat{\alpha}\leq\frac{\epsilon^2}{2}$ (it must returns a $(1+\epsilon)$-radius approximation); but if $\frac{\epsilon^2}{2}<\hat{\alpha}<\epsilon$, the algorithm can return either  a $(1-\delta)$-covering approximation or a $(1+\epsilon)$-radius approximation.
\end{theorem}
\begin{proof}
Recall we set $\alpha=\epsilon^2$ and $\beta=\delta/2$ in Algorithm~\ref{alg-relaxoutlier-2}. 
First, we suppose the output is a $(1+\epsilon)$-radius approximation. One possible case is the instance $P$ is a real $(\alpha, \beta)$-stable instance, and then $\hat{\alpha}=\alpha<\epsilon$. 
The other possible case is that $P$ is not   $(\alpha, \beta)$-stable but the ratio $\frac{r_{\tilde{o}}}{r_{c}}\leq\frac{1+\epsilon}{1-\epsilon^2/2}$. Together with (\ref{for-unknown2-1}), we have 
\begin{eqnarray}
\frac{\mathbf{Rad}(P)}{r_{-\delta/2}}\leq \frac{r_{\tilde{o}}}{\frac{1}{1+\epsilon^2/2}r_{c}}\leq \frac{(1+\epsilon)(1+\epsilon^2/2)}{1-\epsilon^2/2}. 
\end{eqnarray}
So $\hat{\alpha}=1-\frac{r_{-\delta/2}}{\mathbf{Rad}(P)}\leq 1-\frac{1-\epsilon^2/2}{(1+\epsilon)(1+\epsilon^2/2)}<\epsilon$. Overall, as long as the output is a $(1+\epsilon)$-radius approximation, $\hat{\alpha}$ should be smaller than $\epsilon$.

 Then we suppose the output is a $(1-\delta)$-covering approximation. One possible case is the instance $P$ is not $(\alpha, \beta)$-stable, and then $\hat{\alpha}>\alpha=\epsilon^2$. 
The other possible case is that $P$ is   $(\alpha, \beta)$-stable but the ratio $\frac{r_{\tilde{o}}}{r_{c}}>\frac{1+\epsilon}{1-\epsilon^2/2}$. Together with (\ref{for-the-relaxoutlier-2-1}), we have  
 \begin{eqnarray}
\frac{\mathbf{Rad}(P)}{r_{-\delta/2}}\geq \frac{\frac{1}{1+\epsilon}r_{\tilde{o}}}{r_{c}}> \frac{1}{1-\epsilon^2/2}. 
\end{eqnarray}
So $\hat{\alpha}=1-\frac{r_{-\delta/2}}{\mathbf{Rad}(P)}> 1-(1-\epsilon^2/2)=\epsilon^2/2$. Overall, as long as the output is a $(1-\delta)$-covering approximation, $\hat{\alpha}>\min\{\epsilon^2, \epsilon^2/2\}=\frac{\epsilon^2}{2}$. 
  \qed
 \end{proof}


\section{ Extension \Rmnum{1}: Hybrid Approximation for MEB with Outliers}
\label{sec-unknown}

In this section, we extend the idea of Section~\ref{sec-unknown2} to present a hybrid approximation algorithm for the MEB with outliers problem $(P, \gamma)$. 
First, we  extend Definition~\ref{def-stable} of MEB to MEB with outliers.

\begin{definition}[($\alpha$, $\beta$)-stable for MEB with Outliers]
\label{def-outlier-stable}
Let $0<\alpha, \beta<1$. Given an instance $(P, \gamma)$ of the  MEB with outliers problem in Definition~\ref{def-outlier}, $(P, \gamma)$ is an ($\alpha$, $\beta$)-stable instance if (1) $\mathbf{Rad}(P\setminus Q)> (1-\alpha)\mathbf{Rad}(P_{\textnormal{opt}})$ for any $Q\subset P$ with $|Q|< \big(\gamma+\beta\big)n$, and (2) there exists a $Q'\subset P$ with $|Q'|=\lceil(\beta+\gamma)n\rceil$ having $\mathbf{Rad}(P\setminus Q')\leq (1-\alpha)\mathbf{Rad}(P_{\textnormal{opt}})$. 
\end{definition}

Definition~\ref{def-outlier-stable} directly implies the following claim.

\begin{claim}
\label{cla-outlier-stable}
If $(P, \gamma)$ is an ($\alpha$, $\beta$)-stable instance of the problem of MEB with outliers, the corresponding $P_{\textnormal{opt}}$ is an ($\alpha$, $\tilde{\beta}$)-stable instance of MEB with $\tilde{\beta}\geq\frac{\beta}{1-\gamma}$.
\end{claim}

Note that Definition~\ref{def-outlier-stable} implicitly requires $\beta<1-\gamma$. So it implies the lower bound $\frac{\beta}{1-\gamma}$ of $\tilde{\beta}$ in Claim~\ref{cla-outlier-stable} cannot be larger than $1$. 
To see the correctness of 
Claim~\ref{cla-outlier-stable}, we can use contradiction. 
Suppose that there exists a subset $P'\subset P_{\textnormal{opt}}$ such that $|P'|> (1-\frac{\beta}{1-\gamma})|P_{\textnormal{opt}}|=(1-\gamma-\beta)n$ and $\mathbf{Rad}(P')\leq(1-\alpha)\mathbf{Rad}(P_{\textnormal{opt}})$. Then, it is in contradiction to the fact that $(P, \gamma)$ is an $(\alpha, \beta)$-stable instance of MEB with outliers. 

To apply the idea of Section~\ref{sec-unknown2},  a significant challenge is that the set $P_{\textnormal{opt}}$ is mixed with the outliers, and thus we cannot easily obtain a $(1+\epsilon)$-radius approximation as Algorithm~\ref{alg-relaxoutlier-2}. 
Our starting point is still the sublinear time bi-criteria approximation algorithm proposed in Section~\ref{sec-outlier-general}. Specifically, given any two small parameters $0<\epsilon$, $\delta<1$, the algorithm  returns a set of  candidate ball centers via the uniform-adaptive sampling procedure. We use $\Xi$ to denote this set.  With constant probability, as least one candidate from $\Xi$, say $s$, satisfies the following inequality: 
\begin{eqnarray}
\big|\mathbb{B}\big(s, (1+\epsilon)\cdot\mathbf{Rad}(P_{\textnormal{opt}})\big)\cap P\big|\geq \big(1-\delta-\gamma\big)n.\label{for-dingesa20}
\end{eqnarray}
Namely, it is a ``$(1+\epsilon, 1-\delta)$-approximation''. To pick such a qualified candidate,  it is possible to estimate the size of $\mathbb{B}\big(s, (1+\epsilon)\cdot\mathbf{Rad}(P_{\textnormal{opt}})\big)\cap P$ by using the uniform sampling based technique ``sandwich lemma'' (instead of reading the whole dataset $P$). 
%
It is worth to note an implicit fact about Theorem~\ref{the-outlier} of Section~\ref{sec-outlier-general}. Actually, in the proof it showed that among the candidate set $\Xi$, there exists one solution $s$  such that the ball $\mathbb{B}\big(s, (1+\epsilon)\cdot\mathbf{Rad}(P_{\textnormal{opt}})\big)$ covers at least $\big(1-\delta-\gamma\big)n$ points from $P_{\textnormal{opt}}$ (since the set $T\subset P_{\textnormal{opt}}$ and the solution $s$ is generated from $T$ (see Lemma~\ref{lem-outlier-1})). So the solution $s$ should satisfy 
\begin{eqnarray}
\big|\mathbb{B}\big(s, (1+\epsilon)\cdot\mathbf{Rad}(P_{\textnormal{opt}})\big)\cap  P_{\textnormal{opt}}\big|\geq \big(1-\delta-\gamma\big)n, \label{for-dingesa20-2}
\end{eqnarray}
which is stronger than (\ref{for-dingesa20}). But the  sandwich lemma may ignore such a stronger solution, since only selecting a solution satisfying (\ref{for-dingesa20})   is already sufficient to guarantee a $(1+\epsilon, 1-\delta)$-approximation. We introduce the following new algorithm for MEB with outliers based on this observation.

\vspace{0.05in}
\textbf{The hybrid approximation algorithm.}
 Let $\epsilon$ and $\delta$ be the two given parameters. First, we apply the method of Section~\ref{sec-outlier-general}. But we do not directly input the couple $(\epsilon, \delta)$ to the bi-criteria approximation algorithm; instead, we use $(\frac{1}{2(2\sqrt{2}+\sqrt{3})^2}\epsilon^2, \delta)$ (we will explain why we have the coefficient ``$\frac{1}{2(2\sqrt{2}+\sqrt{3})^2}$'' in our analysis). That is, we compute a set $\Xi$ of candidate ball centers via the uniform-adaptive sampling of Section~\ref{sec-outlier-general}, and at least one center yields a $(1+\frac{1}{2(2\sqrt{2}+\sqrt{3})^2}\epsilon^2, 1-\delta)$-approximation for the instance $(P, \gamma)$. Then, for each candidate $q\in \Xi$, we define two values:
\begin{eqnarray}
r_q&=& \min  \Big\{ r>0\mid\big|\mathbb{B}(q, r)\cap P\big|\geq (1-\gamma)n\Big\}; \label{for-alg-relax1}\\
r'_q&=& \min \Big\{ r>0\mid\big|\mathbb{B}(q, r)\cap P\big|\geq \big(1-\delta- \gamma\big)n\Big\}.\label{for-alg-relax2}
\end{eqnarray}
We can compute these two values for all the candidates of $\Xi$ by scanning the input $P$ in one pass (instead of using the sandwich lemma). We select the two points $s_1=\arg\min_{q\in\Xi}r_q$ and $s_2=\arg\min_{q\in\Xi}r'_q$ (they may or may not be the same point). If the ratio $\frac{r_{s_1}}{r'_{s_2}}\leq\frac{1+\epsilon}{1-\epsilon^2/\big(2(2\sqrt{2}+\sqrt{3})^2\big)}$, return the ball $\mathbb{B}(s_1, r_{s_1})$ and say ``it is a $(1+\epsilon)$-radius approximation''; else, return the ball $\mathbb{B}(s_2, r'_{s_2})$ and say ``it is a $(1-\delta)$-covering approximation''.


\begin{algorithm}
   \caption{\textbf{Hybrid Approximation for MEB with Outliers}}
   \label{alg-relaxoutlier}
\begin{algorithmic}[1]
\REQUIRE An instance $(P,\gamma)$ of MEB with outliers problem in $\mathbb{R}^d$; two parameters $0<\epsilon, \delta<1$.
\STATE Apply the uniform-adaptive sampling method of Section~\ref{sec-outlier-general} to obtain a set $\Xi$ of candidate ball centers, where at least one center yields a $(1+\frac{1}{2(2\sqrt{2}+\sqrt{3})^2}\epsilon^2, 1-\delta)$-approximation for the instance $(P, \gamma)$. 

\STATE Read the whole input dataset $P$ in one pass, and compute the values $r_q$ and $r'_q$ as the formulas (\ref{for-alg-relax1}) and (\ref{for-alg-relax2}) for each $q\in \Xi$. 

\STATE Let  $s_1=\arg\min_{q\in\Xi}r_q$ and $s_2=\arg\min_{q\in \Xi}r'_q$.

\STATE If the ratio $\frac{r_{s_1}}{r'_{s_2}}\leq\frac{1+\epsilon}{1-\epsilon^2/\big(2(2\sqrt{2}+\sqrt{3})^2\big)}$, return the ball $\mathbb{B}(s_1, r_{s_1})$ and say ``it is a $(1+\epsilon)$-radius approximation''.

\STATE Else, return the ball $\mathbb{B}(s_2, r'_{s_2})$ and say ``it is a $(1-\delta)$-covering approximation''.  
 \end{algorithmic}
\end{algorithm}

\begin{theorem}
\label{the-relaxoutlier}
With constant success probability, Algorithm~\ref{alg-relaxoutlier} returns either a $(1+\epsilon)$-radius approximation or a $(1-\delta)$-covering approximation, and the running time is $O(g(\epsilon, \delta, \gamma)\cdot nd)$, where $g(\epsilon, \delta, \gamma)=O(\frac{1}{1-\gamma}(\frac{ \gamma+\delta }{\delta})^{O(1/\epsilon^2)})$. The algorithm only needs uniform sampling and a single pass over the input data, and the space complexity in memory is $O(g(\epsilon, \delta, \gamma)\cdot d)$. Moreover, if the input data matrix (the $n\times d$ matrix representing the input $P$) has at most $M\ll nd$ non-zeros entries, the total running time will be $O\big(g(\epsilon, \delta, \gamma)\cdot(n+d+M)\big)$.
\end{theorem}
\begin{remark}
\label{rem-the-relaxoutlier}
Similar to Theorem~\ref{the-relaxoutlier-2}, we will see that when the algorithm returns a $(1-\delta)$-covering approximation, the returned radius is  at most $\big(1-\Theta(\epsilon^2)\big)\cdot\mathbf{Rad}(P_{\textnormal{opt}})$ (see (\ref{for-the-relaxoutlier-3}) and (\ref{for-the-relaxoutlier-4})).
\end{remark}
\begin{proof}\textbf{(of Theorem~\ref{the-relaxoutlier})} 
We study the time and space complexities first.  The method of Corollary~\ref{cor-outlier2}  only needs uniform samplings, and Step 2 of  Algorithm~\ref{alg-relaxoutlier} is a single pass over the input data.  The size of $\Xi$ is $g(\epsilon, \delta, \gamma)=O(\frac{1}{1-\gamma}(\frac{ \gamma+\delta }{\delta})^{O(1/\epsilon^2)})$ based on Corollary~\ref{cor-outlier2}. Overall, the space complexity is $O(g(\epsilon, \delta, \gamma)\cdot d)$. And the complexity for generating $\Xi$ is $O\big(|\Xi|\cdot\mathtt{poly}(\frac{1}{\epsilon}, \frac{1}{\delta})d\big)$ which is sublinear in the input size $nd$. It is easy to see that the complexity of Step 2 dominates the whole complexity. Therefore, the total running time is $O(g(\epsilon, \delta, \gamma)\cdot nd)$. Furthermore, we consider the case that the input matrix is sparse. Similar to the proof of  Theorem~\ref{the-relaxoutlier-2}, we know that 
the complexity of Algorithm~\ref{alg-relaxoutlier} is $O\big(g(\epsilon, \delta, \gamma)\cdot(n+d+M)\big)$ if the input data matrix  has  at most $M\ll nd$ non-zeros entries.

Now, we prove the solution quality. We let  $\alpha=\frac{1}{(2\sqrt{2}+\sqrt{3})^2}\epsilon^2$ and $\beta=(1-\gamma)\delta$, and consider the following two cases. 

\textbf{Case 1:} the instance $(P, \gamma)$ is $(\alpha, \beta)$-stable ({\em i.e.,} $P_{\textnormal{opt}}$ is an $(\alpha, \tilde{\beta})$-stable instance of MEB with $\tilde{\beta}\geq \delta$, according to Claim~\ref{cla-outlier-stable}). Denote by $o$ the optimal center of $\mathbf{MEB}(P_{\textnormal{opt}})$. We suppose one candidate ball center $q_0$ of $\Xi$ satisfies the formula (\ref{for-dingesa20-2}). As a consequence, from Theorem~\ref{the-stable}, we know that $||q_0-o||\leq (2\sqrt{2}+\sqrt{3})\sqrt{\alpha}\cdot\mathbf{Rad}(P_{\textnormal{opt}})=\epsilon\cdot\mathbf{Rad}(P_{\textnormal{opt}})$. That is, 
\begin{eqnarray}
r_{s_1}\leq r_{q_0}\leq (1+\epsilon)\cdot\mathbf{Rad}(P_{\textnormal{opt}}). \label{for-the-relaxoutlier-2}
\end{eqnarray}
If  $\frac{r_{s_1}}{r'_{s_2}}>\frac{1+\epsilon}{1-\epsilon^2/\big(2(2\sqrt{2}+\sqrt{3})^2\big)}$, together with (\ref{for-the-relaxoutlier-2}), we have
\begin{eqnarray}
r'_{s_2}<\Big(1-\epsilon^2/\big(2(2\sqrt{2}+\sqrt{3})^2\big)\Big)\cdot\mathbf{Rad}(P_{\textnormal{opt}}). \label{for-the-relaxoutlier-3}
\end{eqnarray}
Then we can return the ball $\mathbb{B}(s_2, r'_{s_2})$ and say ``it is a $(1-\delta)$-covering approximation''.  On the other hand, when $\frac{r_{s_1}}{r'_{s_2}}\leq\frac{1+\epsilon}{1-\epsilon^2/\big(2(2\sqrt{2}+\sqrt{3})^2\big)}$, from (\ref{for-the-relaxoutlier-2}) we can return the ball $\mathbb{B}(s_1, r_{s_1})$ and say ``it is a $(1+\epsilon)$-radius approximation''.

\textbf{Case 2:}  $(P, \gamma)$ is not an $(\alpha, \beta)$-stable instance. Then it implies
\begin{eqnarray}
r'_{s_2}&<&(1+\frac{1}{2(2\sqrt{2}+\sqrt{3})^2}\epsilon^2)(1-\frac{1}{(2\sqrt{2}+\sqrt{3})^2}\epsilon^2)\cdot\mathbf{Rad}(P_{\textnormal{opt}})\nonumber\\
&<&\Big(1-\epsilon^2/\big(2(2\sqrt{2}+\sqrt{3})^2\big)\Big)\cdot\mathbf{Rad}(P_{\textnormal{opt}}). \label{for-the-relaxoutlier-4}
\end{eqnarray}
If $\frac{r_{s_1}}{r'_{s_2}}\leq\frac{1+\epsilon}{1-\epsilon^2/\big(2(2\sqrt{2}+\sqrt{3})^2\big)}$,  together with (\ref{for-the-relaxoutlier-4}), it implies 
\begin{eqnarray}
r_{s_1}<(1+\epsilon)\cdot\mathbf{Rad}(P_{\textnormal{opt}}). \label{for-the-relaxoutlier-5}
\end{eqnarray}
Then we can return the ball $\mathbb{B}(s_1, r_{s_1})$ and say ``it is a $(1+\epsilon)$-radius approximation''.  On the other hand,  when $\frac{r_{s_1}}{r'_{s_2}}>\frac{1+\epsilon}{1-\epsilon^2/\big(2(2\sqrt{2}+\sqrt{3})^2\big)}$,  from (\ref{for-the-relaxoutlier-4}) we can return  the ball $\mathbb{B}(s_2, r'_{s_2})$ and say ``it is a $(1-\delta)$-covering approximation''.

Since the success probability of the method of Section~\ref{sec-outlier-general} is constant, the overall success probability of Algorithm~\ref{alg-relaxoutlier} is constant as well. 
\qed\end{proof}

We also have the following theorem for inferring the stability of the instance $(P, \gamma)$, and the proof is almost identical to the proof of Theorem~\ref{the-infer}.


\begin{theorem}
\label{the-infer2}
Suppose $(P, \gamma)$ is a $(\hat{\alpha}, (1-\gamma)\delta)$-stable instance. If  Algorithm~\ref{alg-relaxoutlier} returns a $(1+\epsilon)$-radius approximation, then $\hat{\alpha}<\epsilon$; otherwise, the algorithm returns a $(1-\delta)$-covering approximation and it implies $\hat{\alpha}>\frac{\epsilon^2}{2(2\sqrt{2}+\sqrt{3})^2}$. 
\end{theorem}

 \section{Extension \Rmnum{2}: Bi-criteria Approximations for MEX With Outliers}
\label{sec-ext}



In this section, we extend Definition~\ref{def-outlier} for MEB with outliers and define a more general problem called \textbf{minimum enclosing ``x'' (MEX) with Outliers}. Then we show that the ideas of Lemma~\ref{lem-outlier-general1} and~\ref{lem-outlier-general2} can be generalized to deal with MEX with outliers problems, as long as the shape ``x'' satisfies several properties. 
\textbf{To describe a shape ``x'', we need to clarify three basic concepts: center, size, and distance function.}

Let $\mathcal{X}$ be the set of specified shapes in $\mathbb{R}^d$. We require that each shape $x\in \mathcal{X}$ is uniquely determined by the following two components: 
``$c(x)$'', the \textbf{center} of $x$, and 
``$s(x)\geq 0$'', the \textbf{size} of $x$. For any two shapes $x_1, x_2\in\mathcal{X}$,  $x_1=x_2$ if and only if $c(x_1)=c(x_2)$ and $s(x_1)=s(x_2)$. Moreover, given a center $o_0$ and a value $l_0\geq 0$, we use $x(o_0, l_0)$ to denote the shape $x$ with $c(x)=o_0$ and $s(x)=l_0$. 
For different shapes, we have different definitions for the center and size. For example, if $x$ is a ball, $c(x)$ and $s(x)$ should be the ball center and the radius respectively; given $o_0\in \mathbb{R}^d$ and  $l_0\geq 0$,  $x(o_0, l_0)$ should be the ball $\mathbb{B}(o_0, l_0)$. As a more complicated example, consider the $k$-center clustering with outliers problem, which is to find $k$ balls to cover the input point set excluding a certain number of outliers and minimize the maximum radius ({\em w.l.o.g.,} we can assume that the $k$ balls have the same radius). For this problem, the shape ``x'' is a union of $k$ balls in $\mathbb{R}^d$; the center $c(x)$ is the set of the $k$ ball centers and the size $s(x)$ is the radius. 


For any point $p\in \mathbb{R}^d$ and any shape $x\in\mathcal{X}$, we also need to define a \textbf{distance function} $f(c(x), p)$ between the center $c(x)$ and $p$.
 For example, if $x$ is a ball, $f(c(x), p)$ is simply equal to $||p-c(x)||$; if $x$ is a union of $k$ balls with the center $c(x)=\{c_1, c_2, \cdots, c_k\}$, the distance should be $\min_{1\leq j\leq k}||p-c_j||$. Note that the distance function is only for ranking the points to $c(x)$, and not necessary to be non-negative ({\em e.g.,} in Section~\ref{sec-svm-general}, we define a distance function $f(c(x), p)\leq 0$ for SVM). By using this distance function, we can define the set ``$Q$'' and the value ``$l_i$'' when generalizing Lemma~\ref{lem-outlier-general1} and~\ref{lem-outlier-general2} below. To guarantee their correctnesses, we also require $\mathcal{X}$ to satisfy the following three properties. 

\begin{property}
\label{prop-1}
{\em For any two shapes $x_1\neq x_2\in\mathcal{X}$, if $c(x_1)=c(x_2)$, then
\begin{eqnarray}
s(x_1)\leq s(x_2)\Longleftrightarrow x_1 \text{ is covered by } x_2, 
\end{eqnarray}
where ``$x_1$ is covered by $x_2$'' means ``for any point $p\in \mathbb{R}^d$, $p\in x_1\Rightarrow p\in x_2$''.}
\end{property}

\begin{property}
\label{prop-2}
{\em Given any shape $x\in\mathcal{X}$ and any point $p_0\in x$, the set 
\begin{eqnarray}
\{p\mid p\in \mathbb{R}^d \text{ and } f(c(x), p)\leq f(c(x), p_0)\}\subseteq x.
\end{eqnarray}}
\end{property}

%

\begin{property}
\label{prop-4}
{\em Given any shape center $o_0$ and any point $p_0\in \mathbb{R}^d$, let $r_0=\min\{r\mid r\geq 0, p_0\in x(o_0, r)\}$. Then  $p_0\in x(o_0, r_0)$ and  $p_0 \notin x(o_0, r)$ for any $r<r_0$. (\textbf{Note:} usually the value $r_0$ is just the distance from $p_0$ to the shape center $o_0$; but for some cases, such as the SVM problem in Section~\ref{sec-svm-general}, the shape size and distance function have different meanings).}
 \end{property}

Intuitively, Property~\ref{prop-1} shows that $s(x)$ defines an order of the shapes sharing the same center $c(x)$. Property~\ref{prop-2} shows that the distance function $f$ defines an order of the points to a given shape center $c(x)$. Property~\ref{prop-4} shows that a center $o_0$ and a point $p_0$ can define a shape just ``touching'' $p_0$. We can take $\mathcal{X}=\{\text{all $d$-dimensional balls}\}$ as an example. For any two concentric balls, the smaller one is always covered by the larger one (Property~\ref{prop-1}); if a point $p_0$ is inside a ball $x$, any point $p$ having the distance $||p-c(x)||\leq ||p_0-c(x)||$ should be inside $x$ too (Property~\ref{prop-2}); also, given a ball center $o_0$ and a point $p_0$,  $p_0\in \mathbb{B}(o_0, ||p_0-o_0||)$ and $p_0\notin \mathbb{B}(o_0, r)$ for any $r<||p_0-o_0||$ (Property~\ref{prop-4}). 
 
%
%
%
Now, we introduce the formal definitions of the MEX with outliers problem and its bi-criteria approximation. 

\begin{definition} [MEX with Outliers]
\label{def-outlier-general}
Suppose the shape set $\mathcal{X}$ satisfies Property~\ref{prop-1}, \ref{prop-2}, and \ref{prop-4}. 
Given a set $P$ of $n$ points in $\mathbb{R}^d$ and a small parameter $\gamma\in (0,1)$, the MEX with outliers problem is to find the smallest shape $x\in \mathcal{X}$ that covers $(1-\gamma)n$ points. Namely, the task is to find a subset of $P$ with size $(1-\gamma)n$ such that its minimum enclosing shape of $\mathcal{X}$ is the smallest among all possible choices of the subset. The obtained solution is denoted by $\mathbf{MEX}(P, \gamma)$.
\end{definition}


\begin{definition}[Bi-criteria Approximation]
\label{def-app-general}
Given an instance $(P, \gamma)$ for MEX with outliers and two small parameters $0<\epsilon, \delta<1$, a $(1+\epsilon, 1-\delta)$-approximation of $(P, \gamma)$ is a solution $x\in\mathcal{X}$ that covers at least $\big(1-\delta-\gamma\big)n$ points and has the size at most $(1+\epsilon)s(x_{\text{opt}})$, where $x_{\text{opt}}$ is the optimal solution.
\end{definition}

It is easy to see that Definition~\ref{def-outlier} of MEB with outliers actually is a special case of Definition~\ref{def-outlier-general}. 
Similar to MEB with outliers, we still use $P_{\text{opt}}$, where $P_{\text{opt}}\subset P$ and $|P_{\text{opt}}|=(1-\gamma)n$, to denote the subset covered by the optimal solution of MEX with outliers. 

%
%
Now, we provide the generalized versions of Lemma~\ref{lem-outlier-general1} and \ref{lem-outlier-general2}. Similar to the core-set construction method in Section~\ref{sec-newanalysis}, we assume that there exists an iterative  algorithm $\Gamma$ to compute MEX (without outliers); actually, this is an important prerequisite to design the sub-linear time algorithms under our framework (we will discuss the iterative algorithms for the MEX with outliers problems   including flat fitting, $k$-center clustering, and SVM, in the following subsections). 
In the $i$-th iteration of $\Gamma$, it maintains a shape center $o_i$.  Also, let $Q$ be the set of $(\delta+\gamma) n$ farthest points from $P$ to $o_i$ with respect to the distance function $f$. First, we need to define the value ``$l_i$'' by $Q$ in the following claim. 

\begin{claim}
\label{claim-prop}
There exists a value $l_i\geq 0$ satisfying $P\setminus x(o_i, l_i)=Q$.
\end{claim}
\begin{proof}
The points of $P$ can be ranked based on their distances to $o_i$. Without loss of generality, let $P=\{p_1, p_2, \cdots, p_n\}$ with $f(o_i, p_1)> f(o_i, p_2)>\cdots > f(o_i, p_n)$ (for convenience, we assume that any two distances are not equal; if there is a tie, we can arbitrarily decide their order to $o_i$). Then the set $Q=\{p_j\mid 1\leq j\leq (\delta+\gamma) n\}$. Moreover, from Property~\ref{prop-4}, we know that each point $p_j\in P$ corresponds to a value $r_j$ that $p_j\in x(o_i, r_j)$ and $p_j\notin x(o_i, r)$ for any $r<r_j$. 
Denote by $x_{j}$ the shape $x(o_i, r_{j})$. We select the  point  $p_{j_0}$ with $j_0=(\delta+\gamma) n+1$. 
From Property~\ref{prop-2}, we know that $p_j\in x_{j_0}$ for any $j\geq j_0$, {\em i.e.,} \textbf{(a)}~$P\setminus Q\subseteq x_{j_0}$. We also need to prove that $p_j\notin x_{j_0}$ for any $j< j_0$. Assume there exists some $p_{j_1}\in x_{j_0}$ with $j_1< j_0$. Then we have $r_{j_1}< r_{j_0}$ and thus $p_{j_0}\notin x_{j_1}$ (by Property~\ref{prop-4}). By Property~\ref{prop-2}, $p_{j_0}\notin x_{j_1}$ implies $f(o_i, p_{j_0})>f(o_i, p_{j_1})$, which is in contradiction to the fact $f(o_i, p_{j_0})<f(o_i, p_{j_1})$. So we have \textbf{(b)}~$Q\cap x_{j_0}=\emptyset$. 

The above  \textbf{(a)} and \textbf{(b)} imply that $\{P\cap x_{j_0}, Q\}$ is a partition of $P$, {\em i.e.,} $(P\cap x_{j_0})\cup Q=P$ and $(P\cap x_{j_0})\cap Q=\emptyset$.  So we know $P\setminus x_{j_0}=Q$. Therefore, we can set the value $l_i=r_{j_0}$ and then $P\setminus x(o_i, l_i)=Q$. 
%
%
%
%
%
%
%
%
%
%
%
%
%
%
%
%
%
%
%
%
%
%
%
%
\qed\end{proof}





\begin{lemma}[Generalized Uniform-Adaptive Sampling]
\label{lem-outlier-general1-general}
Let $\eta_1\in(0,1)$.  If we sample $n'=O(\frac{1}{\delta}\log\frac{1}{\eta_1})$ points independently and uniformly at random from $P$ and let $Q'$ be the set of farthest $\frac{3}{2}(\delta+\gamma) n'$ points to $o_i$ from the sample, then, with probability at least $1-\eta_1$, the following holds 
\begin{eqnarray}
\frac{\Big|Q'\cap\big(P_{\text{opt}}\cap Q\big)\Big|}{|Q'|}\geq \frac{\delta}{3(\gamma+\delta)}.
\end{eqnarray}
\end{lemma}
\begin{proof}
Let $A$ denote  the set of sampled $n'$ points from $P$. Similar to (\ref{for-outlier-general1-1}), we have
\begin{eqnarray}
\Big|A\cap\big(P_{\text{opt}}\cap Q\big)\Big|> \frac{1}{2}\delta n' \text{\hspace{0.2in} and \hspace{0.2in}} \Big|A\cap  Q\Big|< \frac{3}{2}(\delta+\gamma) n' \label{for-outlier-general1-1-2}
\end{eqnarray}
with probability $1-\eta_1$. Similar to (\ref{for-aq}), we have 
\begin{eqnarray}
A\cap Q=\{p\in A\mid f(o_i, p)>f(o_i, p_{j_0})\},\label{for-aq-1}
\end{eqnarray}
where $p_{j_0}$ is the point selected in the proof of Claim~\ref{claim-prop}. By using the same manner of Claim~\ref{claim-prop}, we also can select a point $p_{j'_0}\in A$ with 
\begin{eqnarray}
 Q'=\{p\in A\mid  f(o_i, p)>f(o_i, p_{j'_0})\}.\label{for-qprime-1}
\end{eqnarray}
Then, we can prove 
\begin{eqnarray}
\Big(A\cap\big(P_{\text{opt}}\cap Q\big)\Big)= \Big(Q'\cap\big(P_{\text{opt}}\cap Q\big)\Big).\label{for-v9-2-1} 
\end{eqnarray}
by using the same idea of (\ref{for-v9-2}). Hence, 
\begin{eqnarray}
\frac{\Big|Q'\cap\big(P_{\text{opt}}\cap Q\big)\Big|}{|Q'|}&=&\frac{\Big|A\cap\big(P_{\text{opt}}\cap Q\big)\Big|}{|Q'|}\geq \frac{\delta}{3(\gamma+\delta)},
\end{eqnarray}
where the final inequality comes from the first inequality of (\ref{for-outlier-general1-1-2}) and the fact $|Q'|=\frac{3}{2}(\delta+\gamma) n'$.
\qed\end{proof}

\begin{lemma} [Generalized Sandwich Lemma]
\label{lem-outlier-general2-generalize}
Let $\eta_2\in(0,1)$ and assume $\delta<\gamma/3$. $l_i$ is the value from Claim~\ref{claim-prop}. We sample $n''=O\big(\frac{\gamma}{\delta^2}\log\frac{1}{\eta_2}\big)$ points independently and uniformly at random from $P$ and let $q$ be the $\big((1+\delta/\gamma)^2\gamma n''+1\big)$-th farthest one from the sampled points to $o_i$. If 
$\tilde{l}_i=\min\{r\mid r\geq 0, q\in x(o_i, r)\}$ (similar to the way defining ``$r_0$'' in Property~\ref{prop-4}), then, with probability $1-\eta_2$, the following holds
\begin{eqnarray}
\tilde{l}_i&\leq& l_i\label{for-outlier-general2-g1};\\
\Big|P\setminus x(o_i, \tilde{l}_i)\Big|&\leq& (\gamma+5\delta) n.\label{for-outlier-general2-g2}
\end{eqnarray}
\end{lemma}
\begin{proof}
Let $B$ denote  the set of sampled $n''$ points from $P$. By using the same manner of Claim~\ref{claim-prop}, we know that there exists a value $\tilde{l}'_i>0$ satisfying $\Big|P\setminus x(o_i, \tilde{l}'_i)\Big|=\frac{(\gamma+\delta)^2}{\gamma-\delta}\gamma n$. Similar to the proof of Lemma~\ref{lem-outlier-general2}, we can prove that $\tilde{l}_i\in[\tilde{l}'_i,l_i]$. Due to Property~\ref{prop-1}, we know that $x(o_i, \tilde{l}_i)$ is ``sandwiched'' by the two shapes $x(o_i, \tilde{l}'_i)$ and $x(o_i, l_i)$. Further, since $x(o_i, \tilde{l}'_i)$ is covered by $x(o_i, \tilde{l}_i)$, we have
\begin{eqnarray}
\Big|P\setminus x(o_i, \tilde{l}_i)\Big|&\leq& \Big|P\setminus x(o_i, \tilde{l}'_i)\Big|=\frac{(\gamma+\delta)^2}{\gamma-\delta}\gamma n=(\gamma+5\delta) n,
\end{eqnarray}
where the last equality comes from the assumption $\delta<\gamma/3$. So (\ref{for-outlier-general2-g1}) and (\ref{for-outlier-general2-g2}) are true. 
\qed\end{proof}

By using Lemma~\ref{lem-outlier-general1-general} and Lemma~\ref{lem-outlier-general2-generalize}, we study several applications in the following subsections.

\subsection{$k$-Center Clustering with Outliers}
\label{sec-ext-kcenter}

%
%


Let $\gamma\in(0,1)$. Given a set $P$ of $n$ points in $\mathbb{R}^d$, the problem of \textbf{$k$-center clustering with outliers} is to find $k$ balls to cover $(1-\gamma)n$ points, and the maximum radius of the balls is minimized ({\em w.l.o.g.,} we can assume that the $k$ balls have the same radius). 
Given an instance $(P, \gamma)$, let $\{C_1, \cdots, C_k\}$ be the $k$ clusters forming $P_{\text{opt}}$ (the subset of $P$ yielding the optimal solution), and $r_{\text{opt}}$ be the optimal radius; that is, each $C_j$ is covered by an individual ball with radius $r_{\text{opt}}$. Similar to Section~\ref{sec-outlier-general}, we first introduce a linear time algorithm, and then show how to modify it to be sublinear time by using Lemma~\ref{lem-outlier-general1-general} and \ref{lem-outlier-general2-generalize}.

%


\textbf{Linear time algorithm.} Our algorithm in Section~\ref{sec-quality} can be generalized to be a linear time bi-criteria algorithm for the problem of $k$-center clustering with outliers, if $k$ is assumed to be a constant. Our idea is as follows. In Algorithm~\ref{alg-outlier}, we maintain a set $T$ as the core-set of $P_{\text{opt}}$; here, we instead maintain $k$ sets $T_1, T_2, \cdots, T_k$ as the core-sets of $C_1, C_2, \cdots, C_k$, respectively. Consequently, each $T_j$ for $1\leq j\leq k$ has an approximate MEB center $o^j_i$ in the $i$-th round of Step 3, and we let $O_i=\{o^1_i, \cdots, o^k_i\}$. Initially, $O_0$ and $T_j$ for $1\leq j\leq k$ are all empty; we randomly select a point $p\in P$, and with probability $1-\gamma$, $p\in P_{\text{opt}}$ (w.l.o.g., we assume $p\in C_1$ and add it to $T_1$; thus $O_1=\{p\}$ after this step). 
We let $Q$ be the set of farthest $t=(\delta+\gamma) n$ points to $O_i$, and $l_i$ be the $(t+1)$-th largest distance from $P$ to $O_i$ (the distance from a point $p\in P$ to $O_i$ is $\min_{1\leq j\leq k}||p-o^j_i||$). Then, we randomly select a point $q\in Q$, and with probability $\frac{\delta}{\gamma+\delta}$, $q\in P_{\text{opt}}$ (as (\ref{for-lem-outlier-1}) in Lemma~\ref{lem-outlier-1}). For ease of presentation, we assume that $q\in P_{\text{opt}}$ happens and we have an ``oracle'' to guess which optimal cluster $q$ belongs to, say $q\in C_{j_q}$; then, we add $q$ to $T_{j_q}$ and update the approximate MEB center of $T_{j_q}$. Since each optimal cluster $C_j$ for $1\leq j\leq k$ has the core-set with size $\frac{2}{\epsilon}+1$ (by setting $s=\frac{\epsilon}{2+\epsilon}$ in Theorem~\ref{the-newbc}), after adding at most $k(\frac{2}{\epsilon}+1)$ points, the distance $l_i$ will be smaller than $(1+\epsilon)r_{\text{opt}}$. Consequently, a $(1+\epsilon, 1-\delta)$-approximation solution is obtained when $i\geq k(\frac{2}{\epsilon}+1)$. \textbf{Note} that some ``small'' clusters could be missing from the above random sampling based approach and therefore $|O_i|$ could be less than $k$; however, it always can be guaranteed that the total number of missing inliers is at most $\delta  n$, {\em i.e.,} a $(1+\epsilon, 1-\delta)$-approximation is always guaranteed (otherwise, the ratio $\frac{|P_{\text{opt}}\cap Q|}{|Q|}>\frac{\delta}{\gamma+\delta}$ and we can continue to sample a point from $P_{\text{opt}}$ and then update $O_i$). 

To remove the oracle for guessing the cluster containing $q$, we can enumerate all the possible $k$ cases; since we add $k(\frac{2}{\epsilon}+1)$ points to $T_1, T_2, \cdots, T_k$, it generates $k^{k(\frac{2}{\epsilon}+1)}=2^{k\log k (\frac{2}{\epsilon}+1)}$ solutions in total, and at least one yields  a $(1+\epsilon, 1-\delta)$-approximation with probability $(1-\gamma)(\frac{\delta}{\gamma+\delta})^{k(\frac{2}{\epsilon}+1)}$ (by the same manner for proving Theorem~\ref{the-outlier}).

\begin{theorem}
\label{the-kcenter}
Let $(P, \gamma)$ be an instance of $k$-center clustering with outliers. Given two parameters $\epsilon, \delta\in (0,1)$, there exists an algorithm that outputs a $(1+\epsilon, 1-\delta)$-approximation with probability $(1-\gamma)(\frac{\delta}{\gamma+\delta})^{k(\frac{2}{\epsilon}+1)}$. The running time is $O(2^{k\log k (\frac{2}{\epsilon}+1)}(n+\frac{1}{\epsilon^5})d)$.

If one repeatedly runs the algorithm $O(\frac{1}{1-\gamma}(\frac{\gamma+\delta}{\delta})^{k(\frac{2}{\epsilon}+1)})$ times, with constant probability, the algorithm outputs a $(1+\epsilon,1-\delta)$-approximation solution.

\end{theorem}

Similar to our discussion on the running time for MEB with outliers in Section~\ref{sec-quality}, B\u{a}doiu {\em et al.}~\cite{BHI} also achieved a linear time bi-criteria approximation for the $k$-center clustering with outliers problem (see Section 4 in their paper). However, the hidden constant of their running time is exponential in $(\frac{k}{\epsilon\delta})^{O(1)}$ that is much larger than ``$k\log k (\frac{2}{\epsilon}+1)$'' in Theorem~\ref{the-kcenter}. 

\vspace{0.05in}
\textbf{Sublinear time algorithm.} The linear time algorithm can be further improved to be sublinear time; the idea is similar to that for designing sublinear time algorithm for MEB with outliers in Section~\ref{sec-oulier-improve}. First, we follow Definition~\ref{def-outlier-general} and define the shape set $\mathcal{X}$, where each $x\in\mathcal{X}$ is union of $k$ balls in the space; 
 the center $c(x)$ should be the set of its $k$ ball centers, say $c(x)=\{o^1_x, o^2_x, \cdots, o^k_x\}$, and the size $s(x)$ is the radius, {\em i.e.,} $x=\cup^k_{j=1}\mathbb{B}(o^j_x, s(x))$. Obviously, if $x$ is a feasible solution for the instance $(P, \gamma)$, the size $\Big|P\cap (\cup^k_{j=1}\mathbb{B}(o^j_x, s(x)))\Big|$ should be at least $(1-\gamma)n$. Also, define the distance function $f(c(x), p)=\min_{1\leq j\leq k}||p-o^j_x||$. It is easy to verify that the shape set $\mathcal{X}$ satisfies Property~\ref{prop-1}, \ref{prop-2}, and \ref{prop-4}.
From Lemma~\ref{lem-outlier-general1-general}, we know that it is possible to obtain a point in $P_{\text{opt}}\cap Q$ with probability $(1-\eta_1) \frac{\delta}{3(\gamma+\delta)}$. Further, we can estimate the value $l_i$ and select the best candidate solution based on Lemma~\ref{lem-outlier-general2-generalize}. Overall, we have the following theorem.


\begin{theorem}
\label{the-kcenter2}
Let $(P, \gamma)$ be an instance of $k$-center clustering with outliers. Given the parameters $\epsilon, \delta, \eta_1, \eta_2\in (0,1)$, there exists an algorithm that outputs a $(1+\epsilon, 1-\delta)$-approximation with probability $(1-\gamma)\big((1-\eta_1)(1-\eta_2)\frac{\delta}{3(\gamma+\delta)}\big)^{k(\frac{2}{\epsilon}+1)}$. The running time is $\tilde{O}(2^{k\log k (\frac{2}{\epsilon}+1)}(\frac{\gamma}{\delta^2}+\frac{1}{\epsilon^5})d)$.

If one repeatedly runs the algorithm  $N=O\Big(\frac{1}{1-\gamma}\big(\frac{1}{1-\eta_1}(\frac{3(\gamma+\delta)}{\delta})\big)^{k(\frac{2}{\epsilon}+1)}\Big)$ times with setting $\eta_2=O(\frac{1}{2^{k\log k (\frac{2}{\epsilon}+1)}N})$, with constant probability, the algorithm outputs a $(1+\epsilon,1-\delta)$-approximation solution.

\end{theorem}

\subsection{Flat Fitting with Outliers}
\label{sec-flat}
Let $j$ be a fixed integer between $0$ and $d$. Given a $j$-dimensional flat $\mathcal{F}$ and a point $p\in\mathbb{R}^d$, we define their distance, $dist(\mathcal{F}, p)$, to be the Euclidean distance from $p$ to its projection onto $\mathcal{F}$. Let $P$ be a set of $n$ points in $\mathbb{R}^d$. The problem of \textbf{flat fitting} is to find the $j$-dimensional flat $\mathcal{F}$ that minimizes $\max_{p\in P}dist(\mathcal{F}, p)$. It is easy to see that the MEB problem is the case $j=0$ of the flat fitting problem. Furthermore, given a parameter $\gamma\in (0,1)$, the \textbf{flat fitting with outliers} problem is to find a subset $P'\subset P$ with size $(1-\gamma)n$ such that  $\max_{p\in P'}dist(\mathcal{F}, p)$ is minimized. Similar to MEB with outliers, we also use $P_{\text{opt}}$ to denote the optimal subset. Before presenting our algorithms for flat fitting with outliers, we first introduce the linear time algorithm from Har-Peled and Varadarajan~\cite{DBLP:journals/dcg/Har-PeledV04} for the vanilla version (without outliers). 

We start from the case $j=1$, {\em i.e.,} the flat $\mathcal{F}$ is a line in the space. Roughly speaking, their algorithm is an iterative procedure to update the solution round by round, until it is close enough to the optimal line $l_{\text{opt}}$. There are two parts in the algorithm. \textbf{(1)} It picks an arbitrary point $p_\Delta\in P$ and let $q_\Delta$ be the farthest point of $P$ from $p_\Delta$; it can be proved that the line passing through $p_\Delta$ and $q_\Delta$, denoted as $l_0$, is a good initial solution that yields a $4$-approximation with respect to the objective function. \textbf{(2)} In each of the following  rounds, the algorithm updates the solution from $l_{i-1}$ to $l_i$ where $i\geq 1$ is the current number of rounds: let $p_i$ be the farthest point of $P$ from $l_{i-1}$ and let $h_i$ denote the $2$-dimensional flat spanned by $p_i$ and $l_{i-1}$; then the algorithm computes a set of $O(\frac{1}{\epsilon^8}\log^2\frac{1}{\epsilon})$ lines on $h_i$, and picks one of them as $l_i$ via an ``oracle''. They proved that the improvement from $l_{i-1}$ to $l_i$ is significant enough; thus, after running $\nu=O(\frac{1}{\epsilon^3}\log\frac{1}{\epsilon})$ rounds, it is able to achieve a $(1+\epsilon)$-approximation. To remove the ``oracle'', the algorithm can enumerate all the $O(\frac{1}{\epsilon^8}\log^2\frac{1}{\epsilon})$ lines on $h_i$, and thus the total running time is $O\big(2^{\frac{1}{\epsilon^3}\log^2\frac{1}{\epsilon}} n d\big)$. 

\vspace{0.05in}

\textbf{Linear time algorithm.} Now we consider to adapt the above algorithm to the case with outliers, where in fact the idea is similar to the idea proposed in Section~\ref{sec-quality} for MEB with outliers. For simplicity, we still use the same notations as above. Consider the part \textbf{(1)} first. If we randomly pick a point $p_{\Delta}$ from $ P$, with probability $1-\gamma$, it belongs to $P_{\text{opt}}$; further, we randomly pick a point, denoted as $q_\Delta$, from the set of $(\delta_0+\gamma) n$ farthest points of $P$ from $p_{\Delta}$, where the value of $\delta_0$ will be determined below. Obviously, with probability $\frac{\delta_0}{\gamma+\delta_0}$, $q_\Delta\in P_{\text{opt}}$. Denote by $P_0=\{p\in P_{\text{opt}}\mid ||p-p_\Delta||\leq ||q_\Delta-p_\Delta||\}$. Then we have the following lemma.

\begin{lemma}
\label{lem-flat-1}
Denote by $l_0$ the line passing through $p_\Delta$ and $q_\Delta$. Then, with probability $(1-\gamma)(\frac{\delta_0}{\gamma+\delta_0})$, 
\begin{eqnarray}
\max_{p\in P_0}dist(l_0, p)\leq 4 \max_{p\in P_{0}}dist(l_{\text{opt}}, p)\leq 4 \max_{p\in P_{\text{opt}}}dist(l_{\text{opt}}, p).\label{for-flat-1}
\end{eqnarray}
Also, the size of $P_0$ is at least $\big(1-(\delta_0+\gamma)\big)n$. 
\end{lemma}
It is straightforward to obtain the size of $P_0$. The inequality (\ref{for-flat-1}) directly comes from the aforementioned result of~\cite{DBLP:journals/dcg/Har-PeledV04}, as long as $p_\Delta$ and $q_\Delta\in P_{\text{opt}}$. So we can use the line $l_0$ as our initial solution. Then, we can apply the same random sampling idea to select the point $p_i$ in the $i$-th round. Namely, we randomly pick a point as $p_i$ from the set of $(\delta_0+\gamma) n$ farthest points of $P$ from $l_i$. Moreover, we need to shrink the set $P_{i-1}$ to $P_i=\{p\in P_{i-1}\mid dist(l_{i-1}, p)\leq dist(l_{i-1}, p_i)\}$. Similar to Lemma~\ref{lem-flat-1}, we can show that the improvement from $l_{i-1}$ to $l_{i}$ is significant enough with probability $(1-\gamma)(\frac{\delta_0}{\gamma+\delta_0})^{i+1}$, and the size of $P_i$ is at least $\big(1-((i+1)\delta_0+\gamma)\big) n$. After running $\nu$ rounds, we obtain the line $l_\nu$ such that $\max_{p\in P_\nu}dist(l_\nu, p)\leq (1+\epsilon)  \max_{p\in P_{\text{opt}}}dist(l_{\text{opt}}, p)$, and $|P_\nu|\geq \big(1-((\nu+1)\delta_0+\gamma)\big)n$. So if we set $\delta_0=\frac{\delta}{\nu+1}$ with a given $\delta\in (0,1)$, the line $l_\nu$ will be a bi-criteria $(1+\epsilon, 1-\delta)$-approximation of the instance $(P, \gamma)$. By using the idea in~\cite{DBLP:journals/dcg/Har-PeledV04}, we can extend the result to the case $j>1$ with $\nu=\frac{e^{O(j^2)}}{\epsilon^{2j+1}}\log\frac{1}{\epsilon}$. We refer the reader to~\cite{DBLP:journals/dcg/Har-PeledV04} for more details. 

\begin{theorem}
\label{the-flat}
Let $(P, \gamma)$ be an instance of $j$-dimensional flat fitting with outliers. Given two parameters $\epsilon, \delta\in (0,1)$, there exists an algorithm that outputs a $(1+\epsilon, 1-\delta)$-approximation with probability $(1-\gamma)\big(\frac{1}{2}\big)^{g(j, \epsilon)}$ where $g(j, \epsilon)=poly(e^{O(j^2)}, \frac{1}{\epsilon^j})$. The running time is $O(2^{g'(j, \epsilon)}nd)$ where $g'(j, \epsilon)=poly(e^{O(j^2)}, \frac{1}{\epsilon^j})$.

If one repeatedly runs the algorithm $ \frac{2^{g(j, \epsilon)}}{1-\gamma}$ times, with constant probability, the algorithm outputs a $(1+\epsilon,1-\delta)$-approximation solution.

\end{theorem}

\textbf{Sublinear time algorithm.} We can view the flat fitting with outliers problem as an MEX with outliers problem. Let $r\geq 0$ and $\mathcal{F}$ be a $j$-dimensional flat. Then we can define a $j$-dimensional ``slab'' $SL(\mathcal{F}, r)=\{p\in\mathbb{R}^d\mid dist(\mathcal{F}, p)\leq r\}$, where its ``center'' and ``size'' are $\mathcal{F}$ and $r$ respectively ({\em e.g.,} a ball is a $0$-dimensional slab); the distance function $f(\mathcal{F}, p)=dist(\mathcal{F}, p)$. It is easy to see that the shape set of slabs satisfies  Property~\ref{prop-1}, \ref{prop-2}, and \ref{prop-4}. Furthermore, finding the optimal flat is equivalent to finding the smallest slab covering $(1-\gamma)n$ points of $P$. Therefore, by using  Lemma~\ref{lem-outlier-general1-general} and~\ref{lem-outlier-general2-generalize}, we achieve the following theorem.

\begin{theorem}
\label{the-flat2}
Let $(P, \gamma)$ be an instance of $j$-dimensional flat fitting with outliers. Given the parameters $\epsilon, \delta, \eta_1, \eta_2\in (0,1)$, there exists an algorithm that outputs a $(1+\epsilon, 1-\delta)$-approximation with probability $(1-\gamma)\big((1-\eta_1)(1-\eta_2)\frac{\delta}{3(\gamma+\delta)}\big)^{g(j, \epsilon)}$ where $g(j, \epsilon)=poly(e^{O(j^2)}, \frac{1}{\epsilon^j})$. The running time is $O(2^{g'(j, \epsilon,\delta,\gamma)} d)$ where $g'(j, \epsilon)=poly(e^{O(j^2)}, \frac{1}{\epsilon^j}, \frac{1}{\delta}, \frac{1}{\gamma})$.

If one repeatedly runs the algorithm  $N=O\Big(\frac{1}{1-\gamma}\big(\frac{1}{1-\eta_1}(\frac{3(\gamma+\delta)}{\delta})\big)^{g(j, \epsilon)}\Big)$ times with setting $\eta_2=O(\frac{1}{2^{g(j, \epsilon)}N})$, with constant probability, the algorithm outputs a $(1+\epsilon,1-\delta)$-approximation solution.

\end{theorem}

\subsection{One-class SVM with Outliers}
\label{sec-svm-general}

In practice, datasets often contain outliers.  The separating margin of SVM  could be considerably deteriorated by outliers.  As mentioned in~\cite{ding2015random}, most of existing techniques~\cite{conf/aaai/XuCS06,icml2014c2_suzumura14}  for SVM outliers removal are numerical approaches ({\em e.g.,} adding some penalty item to the objective function), and only can guarantee local optimums. Ding and Xu~\cite{ding2015random} modeled SVM with outliers as a combinatorial optimization problem and provided an algorithm called ``Random Gradient Descent Tree''. 
We focus on one-class SVM with outliers first, and explain the extension for two-class SVM with outliers in Section~\ref{sec-two-outlier}. 
Below is the definition of the one-class SVM with outliers problem proposed in~\cite{ding2015random}. 

\begin{definition} [One-class SVM with Outliers]
\label{def-svm}
 Given a set $P$ of $n$ points in $\mathbb{R}^d$ and a small parameter $\gamma\in (0,1)$, the one-class SVM with outliers problem is to find a subset $P'\subset P$ with size $(1-\gamma)n$ and a hyperplane $\mathcal{H}$ separating the origin $o$ and $P'$, such that the distance between $o$ and $\mathcal{H}$ is maximized. 
%
\end{definition}

\begin{algorithm}[tb]
   \caption{Gilbert Algorithm \cite{Gilbert:1966:IPC,ding2015random}}
   \label{alg-gilbert}
\begin{algorithmic}
   \STATE {\bfseries Input:} A point-set $P$ in $\mathbb{R}^d$, and $N\in \mathbb{Z}^+$.
    \STATE {\bfseries Output:} $v_i$ as an approximate solution of the polytope distance between the origin and $P$.
    \begin{enumerate}
   \item Initialize $i=1$ and $v_1$ to be the closest point in $P$ to the origin $o$.
   \item Iteratively perform the following steps until $i=N$.
   \begin{enumerate}
   \item Find the point $p_i \in P$ whose orthogonal projection on the supporting line of segment $\overline{ov_{i}}$ has the closest distance to $o$ (called the projection distance of $p_{i}$), {\em i.e.,} $p_i=\arg\min_{p\in P}\{\frac{\langle p, v_i\rangle}{||v_i||}\}$, 
     where $\langle p, v_i\rangle$ is the inner product of $p$ and $v_i$ (see Figure~\ref{fig-gilbert}).
   \item Let $v_{i+1}$ be the point on segment $\overline{v_i p_i}$ closest to the origin $o$; update $i=i+1$.
    \end{enumerate}
    \end{enumerate}
\end{algorithmic}
\end{algorithm}

\begin{figure}[]
\vspace{-0.1in}
   \centering
  \includegraphics[height=1.2in]{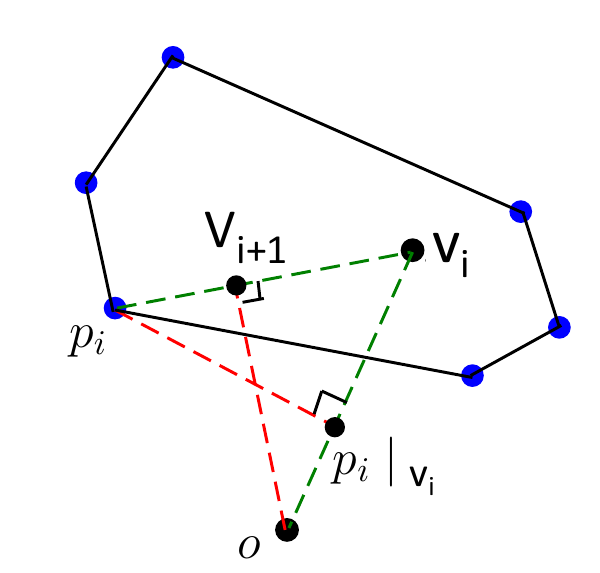}
  \vspace{-0.1in}
      \caption{An illustration of step 2 in Algorithm~\ref{alg-gilbert}; $p_i\mid_{v_i}$ is the projection of $p_i$ on $\overline{o v_i }$.}
  \label{fig-gilbert}
  \vspace{-0.1in}
\end{figure}

\textbf{Linear time algorithm.} We briefly overview the algorithm of~\cite{ding2015random}. They also considered the ``bi-criteria approximation'' with two small parameters $\epsilon, \delta\in (0,1)$: a hyperplane  $\mathcal{H}$ separates the origin $o$ and a subset $P'\subset P$ with size $\big(1- \delta-\gamma\big)n$, where the distance between $o$ and $\mathcal{H}$ is at least $(1-\epsilon)$ of the optimum. The idea of~\cite{ding2015random} is based on the fact that the SVM (without outliers) problem is equivalent to the polytope distance problem in computational geometry~\cite{GJ09}. 

\vspace{0.05in}
{\em Let $o$ be the origin and  $P$ be a given set of points in $\mathbb{R}^d$. The \textbf{polytope distance problem} is to find a point $q$ inside the convex hull of $P$ so that the distance $||q-o||$ is minimized. }
\vspace{0.05in}

For an instance $P$ of one-class SVM, it can be proved that the vector $q_{opt}-o$, if $q_{opt}$ is the optimal solution for the polytope distance between $o$ and $P$, is the normal vector of the optimal hyperplane. We refer the reader to \cite{ding2015random,GJ09} for more details. The polytope distance problem can be efficiently solved by {\em Gilbert Algorithm}~\cite{frank1956algorithm,Gilbert:1966:IPC}. For completeness, we present it in Algorithm~\ref{alg-gilbert}.

Similar to the core-set construction method of MEB in Section~\ref{sec-newanalysis}, the algorithm also greedily improves the current solution by selecting some point $p_i$ in each iteration. Let $\rho$ be the polytope distance between $o$ and $P$, $D=\max_{p,q\in P}||p-q||$, and $E=\frac{D^2}{\rho^2}$. Given $\epsilon\in (0,1)$, it has been proved that a $(1-\epsilon)$-approximation of one-class SVM ({\em i.e.,} a separating margin with the width at least $(1-\epsilon)$ of the optimum) can be achieved by running Algorithm~\ref{alg-gilbert}  at most $2\lceil 2E/\epsilon\rceil$ steps~\cite{GJ09,C10}. To handle outliers, the algorithm of~\cite{ding2015random} follows the similar intuition of Section~\ref{sec-quality}; it replaces the step of greedily selecting the point $p_i$ by randomly sampling a point from a set $Q$, which contains the $(\delta+\gamma) n$ points  having the smallest projection distances ({\em i.e.,} the values of the function $\frac{\langle p, v_i\rangle}{||v_i||}$ in Step 2(a) of Algorithm~\ref{alg-gilbert}). To achieve a $(1-\epsilon, 1-\delta)$-approximation with constant success probability,  the algorithm takes $O\big(\frac{1}{1-\gamma}(1+\frac{\gamma}{\delta})^{z}\frac{D^2}{\epsilon\rho^2}nd\big)$ time, where $z=O(\frac{D^2}{\epsilon\rho^2})$. 



\begin{figure}[]
\vspace{-0.1in}
   \centering
  \includegraphics[height=1.2in]{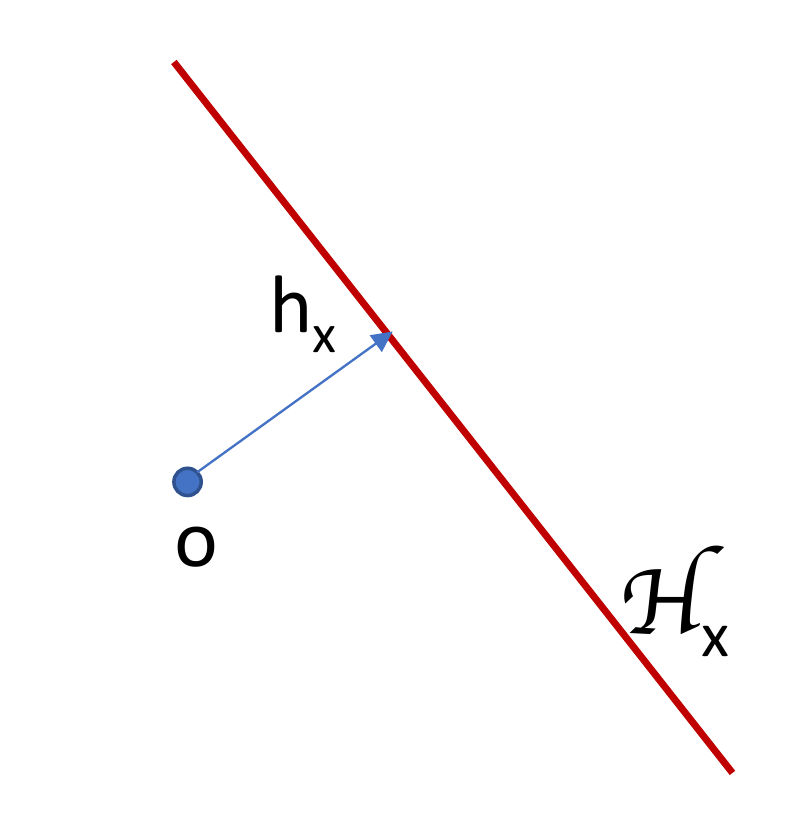}
  \vspace{-0.1in}
      \caption{An illustration for $\mathcal{H}_x$ and $h_x$.}
  \label{fig-nsvm1}
  \vspace{-0.1in}
\end{figure}

\textbf{Sublinear time algorithm.}
We define $\mathcal{X}$ to be the set of all the closed half-spaces not covering the origin $o$ in $\mathbb{R}^d$; for each $x\in \mathcal{X}$, let $\mathcal{H}_x$ be the hyperplane enclosing $x$ and let $h_x$ be the projection of $o$ on $\mathcal{H}_x$ (see Figure~\ref{fig-nsvm1}). We suppose that the given instance $(P, \gamma)$ has feasible solution. That is, there exists at least one half-space $x\in \mathcal{X}$ that the hyperplane $\mathcal{H}_x$ separates the origin $o$ and a subset $P'$ with size $(1-\gamma)n$. We define the center $c(x)= \frac{h_x}{||h_x||}$; since the MEX with outlier problem in Definition~\ref{def-outlier-general} is a minimization problem, we design the size function  $s(x)= \frac{1}{||h_x||}$. 
%
Obviously, a $(1-\epsilon)$-approximation of the SVM with outliers problem  is equivalent to a $\frac{1}{1-\epsilon}$-approximation with respect to the size function $s(x)$. 
 We also define the distance function $f(c(x), p)=-\langle p, \frac{h_x}{||h_x||}\rangle$. It is easy to verify that the shape set $\mathcal{X}$ satisfies Property~\ref{prop-1}, \ref{prop-2}, and~\ref{prop-4}. 
 
Recall that Algorithm~\ref{alg-gilbert} selects the point $p_i=\arg\min_{p\in P}\{\frac{\langle p, v_i\rangle}{||v_i||}\}$ in each iteration. Actually, the vector $\frac{v_i}{||v_i||}$ can be viewed as a shape center 
and $p_i$  is the farthest point to $\frac{v_i}{||v_i||}$ based on the distance function $f(c(x), p)$. Moreover, the set $Q$ mentioned in the previous linear time algorithm actually is the set of the farthest $(\delta+\gamma) n$ points from $P$ to $\frac{v_i}{||v_i||}$. 
Consequently, we can apply  Lemma~\ref{lem-outlier-general1-general} to sample a point from $P_{opt}\cap Q$, and apply Lemma~\ref{lem-outlier-general2-generalize} to estimate the value of $l_i$ for each candidate solution $\frac{v_i}{||v_i||}$. Overall, we can improve the running time of the algorithm of~\cite{ding2015random} to be independent of $n$. 

\begin{theorem}
\label{the-svm}
Let $(P, \gamma)$ be an instance of SVM with outliers. Given the parameters $\epsilon, \delta, \eta_1, \eta_2\in (0,1)$, there exists an algorithm that outputs a $(1-\epsilon, 1-\delta)$-approximation with probability $(1-\gamma)\big((1-\eta_1)(1-\eta_2)\frac{\delta}{3(\gamma+\delta)}\big)^{z}$ where $z=O(\frac{D^2}{\epsilon\rho^2})$. The running time is $\tilde{O}(\frac{D^2\gamma}{\delta^2  \epsilon^2\rho^2} d)$.

If one repeatedly runs the algorithm  $N=O\Big(\frac{1}{1-\gamma}\big(\frac{1}{1-\eta_1}(3+\frac{3\gamma}{\delta})\big)^{z}\Big)$ times with setting $\eta_2=O(\frac{1}{zN})$, with constant probability, the algorithm outputs a $(1-\epsilon,1-\delta)$-approximation solution.

\end{theorem}

\subsection{Two-class SVM with Outliers}
\label{sec-two-outlier}
 
 Below is the definition of the two-class SVM with outliers problem proposed in~\cite{ding2015random}. 

\begin{definition} [Two-class SVM with Outliers]
\label{def-svm2}
 Given two point sets $P_1$ and $P_2$ in $\mathbb{R}^d$ and two small parameters $\gamma_1, \gamma_2\in (0,1)$, the two-class SVM with outliers problem is to find two subsets $P'_1\subset P_1$ and $P'_2\subset P_2$ with $|P'_1|=(1-\gamma_1)|P_1|$ and $|P'_2|=(1-\gamma_2)|P_2|$, and a margin separating  $P'_1$ and $P'_2$, such that the width of the margin is maximized.   
 %
\end{definition}
We use $P^{opt}_1$ and $P^{opt}_2$, where $|P^{opt}_1|=(1-\gamma_1)|P_1|$ and $|P^{opt}_2|=(1-\gamma_2)|P_2|$, to denote the subsets of $P_1$ and $P_2$ which are separated by the optimal margin. 
The ordinary two-class SVM (without outliers) problem is equivalent to computing the polytope distance between the origin $o$ and $\mathcal{M}(P_1, P_2)$, where  $\mathcal{M}(P_1, P_2)$ is the Minkowski difference of $P_1$ and $P_2$~\cite{GJ09}. Note that it is not necessary to compute the set  $\mathcal{M}(P_1, P_2)$ explicitly. Instead, Algorithm~\ref{alg-gilbert} only needs to select one point from $\mathcal{M}(P_1, P_2)$ in each iteration, and overall the running time is still linear in the input size. To deal with two-class SVM with outliers, Ding and Xu~\cite{ding2015random} slightly modified their algorithm for the case of one-class. In each iteration, it considers two subsets $Q_1\subset P_1$ and $Q_2\subset P_2$, which respectively consist of points having the $(\delta+\gamma_1)|P_1|$ smallest  projection distances among all points in $P_{1}$ and the $(\delta+\gamma_2)|P_2|$ largest  projection distances among all points in $P_2$ on the vector $v_i$; then, the algorithm randomly selects two points $p^1_i\in Q_1$ and $p^2_i\in Q_2$, and their difference vector $p^2_i-p^1_i$ will serve as the role of $p_i$ in Step 2(a) of  Algorithm~\ref{alg-gilbert} to update the current solution $v_i$. This approach yields a $(1-\epsilon, 1-\delta)$-approximation in linear time.

\begin{figure}[]
\vspace{-0.1in}
   \centering
  \includegraphics[height=2in]{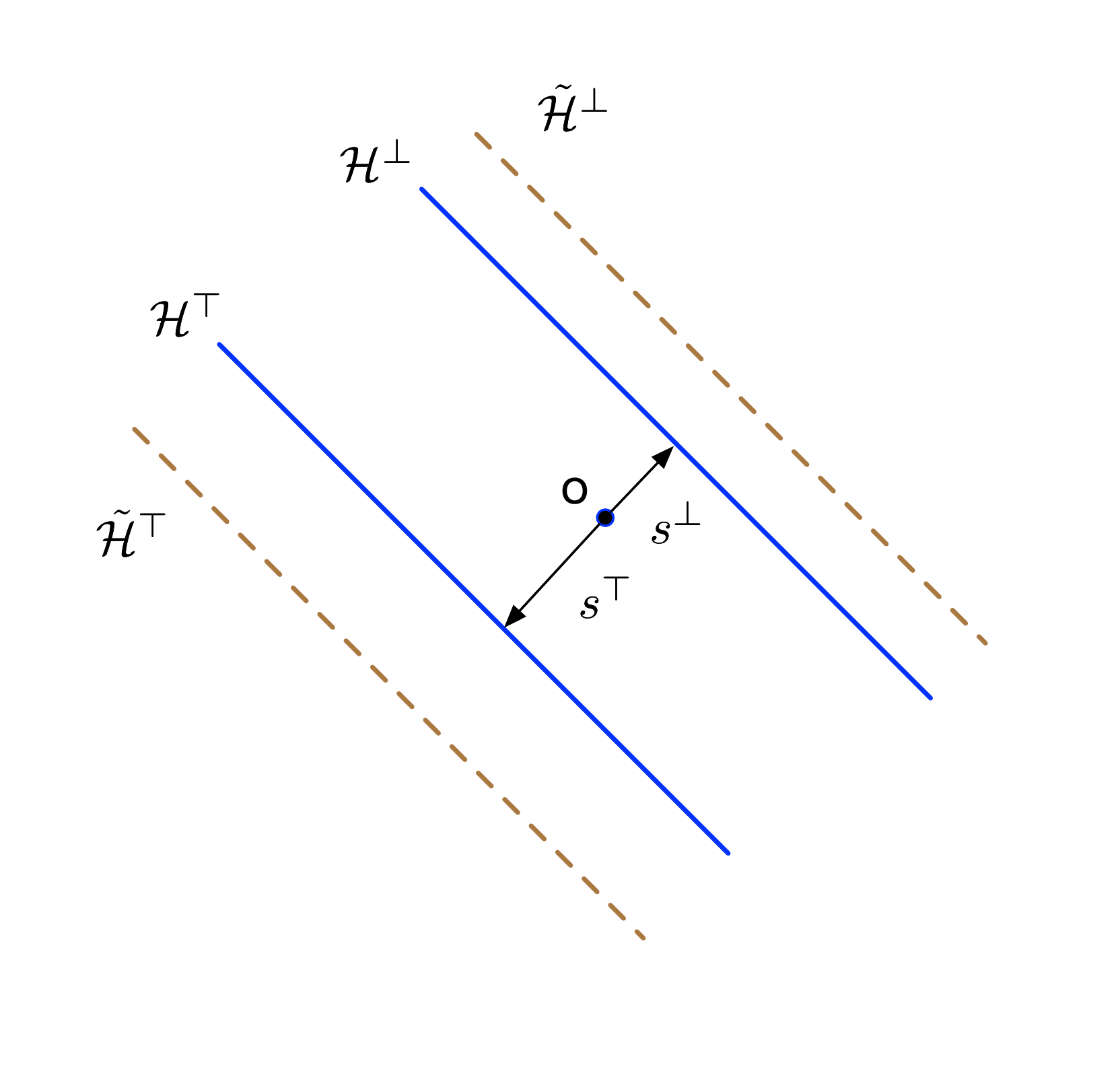}
  \vspace{-0.1in}
      \caption{An illustration for two-class SVM. The distances from $o$ to $\mathcal{H}^\perp$ and $\mathcal{H}^\top$ are $s^\perp$ and $s^\top$, respectively. The hyperplanes  $\tilde{\mathcal{H}}^\perp$ and $\tilde{\mathcal{H}}^\top$ are the estimations of $\mathcal{H}^\perp$ and $\mathcal{H}^\top$, and the distances from $o$ to them are $\tilde{s}^\perp$ and $\tilde{s}^\top$ respectively.}
  \label{fig-twosvm}
  \vspace{-0.1in}
\end{figure}

To improve the algorithm to be sublinear, we need several modifications on our previous idea for the case of one-class. First, we change the distance function to be:
\[ f(p, c)= \left\{ \begin{array}{ll}
       -\langle p, \frac{h_x}{||h_x||}\rangle & \mbox{if $p\in P_1$};\\
       \langle p, \frac{h_x}{||h_x||}\rangle& \mbox{if $p\in P_2$}.\end{array} \right. \] 
By using this new distance function, we can apply Lemma~\ref{lem-outlier-general1-general} to obtain the points $p^1_i\in Q_1\cap P^{opt}_1$ and $p^2_i\in Q_2\cap P^{opt}_2$ separately in sublinear time. Given a vector ({\em i.e.,} candidate center) $\frac{v_i}{||v_i||}$, assume $\mathcal{H}^\perp$ and $\mathcal{H}^\top$ are the parallel hyperplanes orthogonal to $\frac{v_i}{||v_i||}$ that the margin formed by them separates $P'_1$ and $P'_2$, where $P'_1\subset P_1$ and $P'_2\subset P_2$ with $|P'_1|=(1-\gamma_1)|P_1|$ and $|P'_2|=(1-\gamma_2)|P_2|$. 
Without loss of generality, we assume that the origin $o$ is inside the margin. Suppose that the distances from $o$ to $\mathcal{H}^\perp$ and $\mathcal{H}^\top$ are $s^\perp$ and $s^\top$, respectively. Then, we obtain two shapes (closed half-spaces) $x^\perp=(-\frac{v_i}{||v_i||}, \frac{1}{s^\perp})$ and $x^\top=(\frac{v_i}{||v_i||}, \frac{1}{s^\top})$ with $P'_1\subset x^\perp$ and $P'_2\subset x^\top$. Consequently, we can apply  Lemma~\ref{lem-outlier-general2-generalize} twice to obtain two values $ \frac{1}{\tilde{s}^\perp}\leq \frac{1}{s^\perp}$ and $ \frac{1}{\tilde{s}^\top}\leq \frac{1}{s^\top}$ with $\Big| P_1\setminus x(-\frac{v_i}{||v_i||}, \frac{1}{\tilde{s}^\perp})\Big|\leq (O(\delta)+\gamma_1)|P_1|$ and $\Big| P_2\setminus x(\frac{v_i}{||v_i||}, \frac{1}{\tilde{s}^\top})\Big|\leq (O(\delta)+\gamma_2)|P_2|$. Therefore, we can use the value $\tilde{s}^\perp+\tilde{s}^\top$ as an estimation of $s^\perp+s^\top$. See Figure~\ref{fig-twosvm} for an illustration. Overall, we can achieve a $(1-\epsilon, 1-O(\delta))$-approximation in sublinear time.

\section{Future Work}
\label{sec-future}
Following our work, several interesting problems deserve to be studied in future. For example, different from radius approximation, the current research on covering approximation of MEB is still inadequate. In particular, can we provide a lower bound for the complexity of computing covering approximate MEB,  as the lower bound result for radius approximate MEB proved by~\cite{DBLP:journals/jacm/ClarksonHW12}? Also, is it possible to extend the stability notion to other geometric optimization problems with more complicated structures? In Section~\ref{sec-ext}, we only provide the bi-criteria approximations for the MEX with outliers problems. So it is interesting to consider to extend the stability notion to these geometric optimization problems, and then we can design the hybrid approximation algorithms for them.


 \section{Acknowledgements} 
 The research of this work was supported in part by National Key R\&D program of China through
grant 2021YFA1000900 and the Provincial NSF of Anhui through grant 2208085MF163. The author also
want to thank Prof. Jinhui Xu for his helpful comments on this draft. 
 

\bibliographystyle{abbrv}

\bibliography{stability2}

\appendix

\section{Proof of Theorem~\ref{the-newbc}}
\label{sec-proof-newbc}
To ensure the expected improvement in each iteration of the algorithm of~\cite{badoiu2003smaller}, they showed that the following two inequalities hold if the algorithm always selects the farthest point to the current center of $\mathbf{MEB}(T)$:
\begin{eqnarray}
r_{i+1}  \geq  (1+\epsilon)\mathbf{Rad}(P)-L_i; \hspace{0.2in} r_{i+1} \geq  \sqrt{r^2_i+L^2_i},\label{for-bc2}
 \end{eqnarray}
where $r_i$ and $r_{i+1}$ are the radii of $\mathbf{MEB}(T)$ in the $i$-th and $(i+1)$-th iterations, respectively, and $L_i$ is the shifting distance of the center of $\mathbf{MEB}(T)$ from the $i$-th to $(i+1)$-th iteration.

However, we often compute only an approximate $\mathbf{MEB}(T)$ in each iteration. In the $i$-th iteration, we let $c_i$ and $o_i$ denote the centers of the exact and the approximate $\mathbf{MEB}(T)$, 
respectively. Suppose that $||c_i-o_i||\leq \xi r_i$, where $\xi\in (0,\frac{\epsilon}{1+\epsilon})$ (we will see why  this bound is needed later). Note that we only compute $o_i$ rather than $c_i$ in each iteration. As a consequence, we can only select the farthest point (say $q$) to $o_i$. If $||q-o_i||\leq (1+\epsilon)\mathbf{Rad}(P)$,  we are done and a $(1+\epsilon)$-radius approximation of MEB is already obtained. Otherwise, we have
\begin{eqnarray}
(1+\epsilon)\mathbf{Rad}(P)&<& ||q-o_i||\nonumber\\
&\leq& ||q-c_{i+1}||+||c_{i+1}-c_i||+||c_i-o_i||\nonumber\\
&\leq& r_{i+1}+L_i+\xi r_i \label{for-bc-1}
\end{eqnarray}
by  the triangle inequality. In other words, we should replace the first inequality of (\ref{for-bc2}) by $r_{i+1} > (1+\epsilon)\mathbf{Rad}(P)-L_i-\xi r_i$. Also, the second inequality of (\ref{for-bc2}) still holds since it depends only on the property of the exact MEB (see Lemma 2.1 in~\cite{badoiu2003smaller}). Thus,  we have 
\begin{eqnarray}
r_{i+1}\geq \max\Big\{\sqrt{r^2_i+L^2_i}, (1+\epsilon)\mathbf{Rad}(P)-L_i-\xi r_i\Big\}.\label{for-bc4}
\end{eqnarray}

Similar to the analysis in~\cite{badoiu2003smaller}, we let $\lambda_i=\frac{r_i}{(1+\epsilon)\mathbf{Rad}(P)}$. Because $r_i$ is the radius of $\mathbf{MEB}(T)$ and $T\subset P$,  we know $r_i\leq \mathbf{Rad}(P)$ and then $\lambda_i\leq1/(1+\epsilon)$. By simple calculation, we know that when $L_i=\frac{\big((1+\epsilon)\mathbf{Rad}(P)-\xi r_i\big)^2-r^2_i}{2\big((1+\epsilon)\mathbf{Rad}(P)-\xi r_i\big)}$ the lower bound of $r_{i+1}$ in (\ref{for-bc4}) achieves the minimum value. Plugging this value of $L_i$ into (\ref{for-bc4}), we have
\begin{eqnarray}
\lambda^2_{i+1}\geq \lambda^2_i+\frac{\big((1-\xi\lambda_i)^2-\lambda^2_i\big)^2}{4(1-\xi\lambda_i)^2}.\label{for-bc5}
\end{eqnarray}
To simplify  inequality (\ref{for-bc5}), we consider the function $g(x)=\frac{(1-x)^2-\lambda^2_i}{1-x}$, where $0<x<\xi$. Its derivative $g'(x)=-1-\frac{\lambda^2_i}{(1-x)^2}$ is always negative, thus we have
\begin{eqnarray}
g(x)\geq g(\xi)=\frac{(1-\xi)^2-\lambda^2_i}{1-\xi}. \label{for-bc2-1}
\end{eqnarray}
Because $\xi<\frac{\epsilon}{1+\epsilon}$ and $\lambda_i\leq 1/(1+\epsilon)$, we know that  the right-hand side of (\ref{for-bc2-1}) is always non-negative. Using (\ref{for-bc2-1}), inequality (\ref{for-bc5}) can be simplified to 
\begin{eqnarray}
\lambda^2_{i+1}&\geq& \lambda^2_i+\frac{1}{4}\big(g(\xi)\big)^2\nonumber\\
&=&\lambda^2_i+\frac{\big((1-\xi)^2-\lambda^2_i\big)^2}{4(1-\xi)^2}.\label{for-bc2-2}
\end{eqnarray}
(\ref{for-bc2-2}) can be further rewritten as 
\begin{eqnarray}
\Big(\frac{\lambda_{i+1}}{1-\xi}\Big)^2&\geq&\frac{1}{4}\Big(1+(\frac{\lambda_{i}}{1-\xi})^2\Big)^2 \nonumber\\
\Longrightarrow  \frac{\lambda_{i+1}}{1-\xi}&\geq&\frac{1}{2}\Big(1+(\frac{\lambda_{i}}{1-\xi})^2\Big).\label{for-bc2-3}
\end{eqnarray}

Now, we can apply a similar transformation of $\lambda_i$ which was used in~\cite{badoiu2003smaller}. Let $\gamma_i=\frac{1}{1-\frac{\lambda_i}{1-\xi}}$.  We know $\gamma_i>1$ (note $0\leq\lambda_i\leq\frac{1}{1+\epsilon}$ and $\xi<\frac{\epsilon}{1+\epsilon}$). Then, (\ref{for-bc2-3}) implies that 
\begin{eqnarray}
\gamma_{i+1}&\geq&\frac{\gamma_i}{1-\frac{1}{2\gamma_i}}\nonumber\\
&=&\gamma_i\big(1+\frac{1}{2\gamma_i}+(\frac{1}{2\gamma_i})^2+\cdots\big)\nonumber\\
&>&\gamma_i+\frac{1}{2}, \label{for-bc2-4}
\end{eqnarray}
where the equation comes from the fact that $\gamma_i>1$ and thus $\frac{1}{2\gamma_i}\in(0,\frac{1}{2})$. Note that $\lambda_0=0$ and thus $\gamma_0=1$. As a consequence, we have $\gamma_i>1+\frac{i}{2}$. In addition, since $\lambda_i\leq\frac{1}{1+\epsilon}$, that is, $\gamma_i\leq\frac{1}{1-\frac{1}{(1+\epsilon)(1-\xi)}}$, we have
\begin{eqnarray}
i< \frac{2}{\epsilon-\xi-\epsilon\xi}=\frac{2}{(1-\frac{1+\epsilon}{\epsilon}\xi)\epsilon}.\label{for-bc2-5}
\end{eqnarray}

Consequently, we obtain the theorem.

\section{Lemma 2.2 in \cite{BHI}}
\label{sec-bhi}

\begin{lemma}[\cite{BHI}]
Let $\mathbb{B}(c, r)$ be a minimum enclosing ball of a point set $P\subset\mathbb{R}^d$, then any closed half-space that contains $c$, must also contain at least a point from $P$ that is at distance $r$ from $c$.
\end{lemma}

\end{document}